\documentclass[12pt]{article}

\usepackage{amssymb}
\usepackage{indentfirst}
\usepackage{amssymb,epsfig,amsmath,accents,setspace}
\usepackage{theorem}
\usepackage[usenames,dvipsnames]{xcolor}
\usepackage[active]{srcltx}
\usepackage{bm}
\usepackage{multirow}
\usepackage{lscape}
\usepackage{epstopdf}
\usepackage[justification=centering]{caption}

\usepackage{titlesec}

\usepackage{tikz}

\usepackage{url}
\usepackage[breaklinks]{hyperref}
\hypersetup{
    colorlinks=true,       
    linkcolor=red,          
    citecolor=blue,        
}

\usepackage[sort&compress]{natbib}
\bibpunct{(}{)}{;}{a}{,}{,}

\newtheorem{assumption}{Assumption}

\theoremstyle{plain} { \theorembodyfont{\rmfamily}
\newtheorem{definition}{Definition}

}
\newtheorem{proposition}{Proposition}
\newtheorem{lemma}{Lemma}
\newtheorem{corollary}{Corollary}
\newtheorem{theorem}{Theorem}

\newcommand{\bbR}{{\mathbb R}}

\newcommand{\bbX}{{\mathbb X}}

\newcommand{\prob}{\mathbb{P}}

\newcommand{\expect}{\mathbb{E}}

\newcommand{\tran}{^\top}

\def\halmos{\mbox{\quad$\square$}}
\def\Halmos{\mbox{\quad$\square$}}

\newenvironment{proof}{\vspace{1ex}\noindent{\em Proof}\hspace{0.5em}}
    {\vspace{1ex}}

\newenvironment{pfof}[1]{\vspace{1ex}\noindent{\em Proof of #1}\hspace{0.5em}}
    {\vspace{1ex}}

\usepackage{setspace}
 \numberwithin{equation}{section}

 \numberwithin{equation}{section}

\usepackage{xr}

\begin{document}
\title{Portfolio Selection under Median and Quantile Maximization
\thanks{The authors are grateful to Hanqing Jin for his comments. He and Jiang acknowledge financial support from the General Research Fund of the Research Grants Council of Hong Kong SAR (Project No. 14200917).}
}
\author{Xue Dong He\thanks{ Room 505, William M.W. Mong Engineering Building, Department of Systems Engineering and Engineering Management, The Chinese University of Hong Kong, Shatin, N.T., Hong Kong, Telphone: +852-39438336, Email: xdhe@se.cuhk.edu.hk.} and Zhaoli Jiang\thanks{Risk Management Institute,    	National University of Singapore, Singapore 119076,   Email: rmijz@nus.edu.sg.} and Steven Kou\thanks{Corresponding Author. Boston University, Questrom School of Business, Rafik B. Hariri Building, 595 Commonwealth Avenue Boston, MA 02215, Email: kou@bu.edu.}}
\maketitle

\begin{abstract}
Although maximizing median and quantiles is intuitively appealing and has an axiomatic foundation, it is difficult to study the optimal portfolio strategy due to the discontinuity and time inconsistency in the objective function.  We use the intra-personal equilibrium approach to study the problem. Interestingly, we find that the only viable outcome is from the median maximization, because for other quantiles either the equilibrium does not exist or there is no investment in the risky assets. The median maximization strategy gives a simple explanation to why wealthier people invest more percentage of their wealth in risky assets.

\medskip



\end{abstract}

\section{Introduction}
Maximization of the mean of investment return leads to excessive risk taking, so in the modern portfolio selection theory, investors are assumed to be concerned about not only the mean but also the variance of investment return \citep{MarkowitzH:52ps}. On the other hand, median is an alternative to mean to summarize a distribution. Median is also a special case of quantile: The $\alpha$-quantile of a distribution is defined to be the threshold such that the probability of observing a value beyond this threshold is equal to $\alpha$, and median is the $1/2$-quantile.

In contrast to mean-variance analysis and expected utility maximization, the study of median maximization and, more generally, quantile maximization is very limited in the literature. \citet{Ethier2004:KellySystem} show that Kelly portfolio, which was proposed by \citet{Kelly1956:ANewInterpretationInformation} and maximizes the long-run growth rate of wealth, also maximizes the median of wealth in some special market settings among all portfolios that are constant over time. \citet{Manski1988:Ordinal} considers a preference model in which an individual maximizes the $\alpha$-quantile of her utility for some $\alpha\in(0,1)$. Moreover, the author shows $\alpha$ measures the individual's risk attitude with the riskiness of a distribution defined in terms of quantile-preserving spreads. \citet{Chambers2009:Axiomatization} proves that a representation of an individual's preferences for distributions is invariant to ordinal transformation and weakly monotonic with respect to first-order stochastic dominance if and only if it is the $\alpha$-quantile of the individual's utility for some $\alpha\in(0,1)$. \citet{rostek2010quantile} axiomatizes quantile maximization in a Savage setting. Using the framework proposed by \citet{AnscombeAumann1963:DefinitionSubjectiveProb}, \citet{deCastro2019:StaticDynamicQuantile} prove that an individual's preferences are represented by quantile maximization if and only if her tastes and beliefs over consequences are completely separable. \citet{deCastro2019:StaticDynamicQuantile} further provide axiomatization for a dynamic, discrete-time quantile model in which an individual evaluates consumption streams by recursive preferences but using quantile, instead of the standard expected utility, as the representation of risk attitude in the recursive preferences. \citet{deCastro2019:PortfolioSelection} consider a single-period portfolio selection problem in which an agent maximizes the quantile of her portfolio of two risky stocks or of a risk-free asset and a risky stock. \citet{Giovannetti2013:AssetPricing} considers a single-period asset pricing model with a representative agent who maximizes the quantile of her consumption utility. \citet{deCastro2019:DynamicQuantileModels} apply the dynamic quantile model proposed by \citet{deCastro2019:StaticDynamicQuantile} to portfolio selection and asset pricing in a multi-period setting.

In the present paper, we study dynamic portfolio selection under quantile maximization. More precisely, we consider an agent who trades a risk-free asset and multiple   risky assets continuously in time with an objective of maximizing the $\alpha$-quantile of her wealth at certain terminal time. The risk-free rate and the mean return rates and volatility of the risky assets are assumed to be deterministic, i.e. the Black-Scholes model. In addition, the agent faces some cone constraints on her portfolio, an example being the constraint of no short sales. This portfolio problem is time inconsistent in that a dynamic portfolio that, at the current time, maximizes the quantile of the terminal wealth does not necessarily maximize the quantile of the terminal wealth at future time. This is in contrast to the dynamic portfolio selection problem considered by \citet{deCastro2019:DynamicQuantileModels}, where the authors employ the dynamic quantile model proposed by \citet{deCastro2019:StaticDynamicQuantile}. Indeed, in that dynamic quantile model, quantiles are used at the beginning of each period to evaluate certain risk at the end of the period. In other words, quantile maximization in that model is only applied in every single period, so there is no time inconsistency caused by quantiles. In our model, however, the agent evaluates the quantile of her wealth at the terminal time, and there is a positive, continuous time period between the time of the evaluation of the quantile and the terminal time, so time inconsistency arises.

Because of time inconsistency, without self-control or the help of commitment devices, the agent cannot commit her future selves to following the dynamic portfolio that maximizes the quantile of her terminal wealth today. Following the literature on time inconsistency, we consider so-called intra-personal equilibrium portfolio strategies; see for instance \citet{Strotz1955:MyopiaInconsistency}, \citet{EkelandLazrak2006:BeingSeriousAboutNonCommitment}, \citet{Bjork2017:TimeInconsistent}, \citet{HeJiang2019:OnEquilibriumStrategies}, and the references therein. More precisely, we assume that the agent has no self control, so we regard her selves at different time to be different players in a game and seek an equilibrium in this game. As a result, an intra-personal equilibrium is a time-consistent portfolio strategy because at any time the agent is not willing to deviate from it and thus is able to implement it throughout the investment horizon.

 We focus on time-varying affine strategies under which the dollar amount invested in the risky assets are affine functions of the agent's wealth and the intercepts and linear coefficients are time varying. This family of strategies include many commonly used strategies, such as investing a certain proportion of wealth in the risky assets, investing a wealth-independent dollar amount in the risky assets, and the mixture of the above two strategies taken in different time periods. We prove that for $\alpha$-quantile maximization with $\alpha>1/2$, there does not exist an intra-personal equilibrium strategy that is time varying and affine. For $\alpha$-quantile maximization with $\alpha<1/2$, a time-varying affine strategy is an intra-personal equilibrium strategy if and only if it leads to zero investment in the risky assets at all time. For $1/2$-quantile maximization, namely, median maximization, a time-varying affine strategy is an intra-personal equilibrium strategy if and only if it is a portfolio insurance strategy under which the agent sets up a portfolio insurance level $\xi$ and invests any capital in excess of $\xi$ into Kelly portfolio and the remaining into the risk-free asset. In particular, we derive multiple intra-personal equilibrium strategies because different portfolio insurance levels lead to different investment strategies.

Our results show that in continuous-time portfolio selection, the risk attitude of the agent does not vary smoothly as $\alpha$ changes. When $\alpha>1/2$, the agent is too risk seeking to take limited risk. When $\alpha<1/2$, the agent is so risk averse that she takes no risk at all. When $\alpha=1/2$, the agent takes nonzero, limited risk. In this case, we show that with a smaller portfolio insurance level $\xi$, the median of the terminal wealth becomes strictly larger at any time and any wealth level but the agent can end up with lower wealth levels. Therefore, $\xi$ becomes a parameter that trades off the growth of the portfolio, which is measured by the median of the terminal wealth, and the risk of the portfolio, which is measured by the lowest level that the agent's wealth in the future may touch. Thus, we can consider $\xi$ to be the second parameter to represent the agent's risk attitude.

We then compare the intra-personal equilibrium strategy to fractional Kelly strategies and find that neither of them dominates the other in terms of median of the terminal wealth. The intra-personal equilibrium strategy, however, entails less risk than fractional Kelly strategies because it implies a higher level of minimum wealth. We also compare the intra-personal equilibrium strategy to the pre-committed strategy and the naive strategy under median maximization. The pre-committed strategy is one that maximizes the quantile of the terminal wealth at the initial time, so it is optimal for the agent's self at the initial time. The naive strategy is the actual strategy implemented by the agent if she were not aware of the time-inconsistency and thus, at each time, is only able to implement in an infinitesimally period of time the strategy that maximizes the quantile of the terminal wealth at that time. We find that the pre-committed strategy can lead to arbitrarily large holding of the risky assets and leads to capped wealth, so it is less preferable to the intra-personal equilibrium strategy. Under the naive strategy, the agent would take an infinite amount of risk around the terminal time and the median of the terminal wealth is always lower than that of the intra-personal equilibrium strategy. Thus, the naive strategy is not preferable either.

As an application, we find that the intra-personal equilibrium strategy can explain an empirical finding by \citet{WachterYogo2010:Household}: Households with a higher level of wealth tend to have a larger portfolio shares in risky assets (i.e., to invest more percentage of their wealth in risky assets). This empirical finding cannot be explained by expected utility maximization with a power utility function even when the asset prices have stock volatility, stochastic return rates, or jumps, because the resulting optimal percentage of wealth invested in risky assets is independent of the agent's wealth. The mean-variance log return model proposed by \citet{dai2020dynamic} cannot explain the finding either because in their model the percentage of wealth invested in risky assets is independent of the agent's wealth as well;    see Theorem 3 in \citet{dai2020dynamic}. Under the intra-personal equilibrium strategy in our model, which is derived in the simple Black-Scholes model, households with a higher level of wealth indeed have a larger portfolio shares in risky assets. The intuition is as follows: When households become older, their wealth, on average, also becomes larger and thus farther away from the portfolio insurance level, so they invest more in the risky assets.

  The intra-personal equilibrium strategy in our model differs from the portfolio strategies in \citet{deCastro2019:DynamicQuantileModels}, \citet{deCastro2019:PortfolioSelection}, and \citet{Giovannetti2013:AssetPricing}, where the authors apply quantile maximization to evaluate random payoffs in each single period. Indeed, in the latter, the portfolio strategies are to invest a percentage of wealth in risky assets, and the percentage is independent of the wealth, so these strategies cannot explain why more wealthy people hold a larger wealth share of risky assets. By contrast, we derive multiple intra-personal equilibrium strategies, and they do not imply proportional investment in risky assets unless the portfolio insurance level is zero. %

Because of the dynamic setting and time inconsistency in our model, the methodologies in \citet{deCastro2019:DynamicQuantileModels}, \citet{deCastro2019:PortfolioSelection}, and \citet{Giovannetti2013:AssetPricing} cannot apply in our model. Moreover, the general framework of solving intra-personal equilibrium that is proposed by \citet{Bjork2017:TimeInconsistent} and followed by nearly all studies on time-inconsistent problems cannot apply either, because in this framework the objective function is set to be the expectation of certain reward function or a nonlinear transformation of it. To overcome the difficulties, we perform infinitesimal analysis to derive the rate of increment in the quantile of the terminal wealth due to deviation from an intra-personal equilibrium in a small period of time and use this rate to characterize intra-personal equilibrium.

Quantiles can be regarded as a special case of preference representations that involve probability weighting, such as rank-dependent utility \citep{QuigginJ:82rd}, prospect theory \citep{KahnemanDTverskyA:79pt,TverskyKahneman1992:CPT}, and the dual theory of choice \citep{YaariM:87dtc}. Quantiles are also related to risk measures, such as value-at-risk and expected shortfall. Most studies of portfolio selection under preferences with probability weighting or under various risk measures in dynamic settings focus on pre-committed strategies; see for instance \citet{BasakShapiro2001:VaRRiskManagement}, \citet{HeJinZhou2014:PortfolioSelectionRiskMeasure}, \citet{JinHQZhouXY:08bps}, \citet{HeXDZhouXY:2011PortfolioChoiceviaQuantiles,HeXDZhouXY:10hfa}, and \citet{vanBilsenLaeven2020:DynamicConsumption}. On the other hand, \citet{EpsteinLZinS:90fora} and \citet{DeGiorgiLegg2011:NarrowFramingProbabilityWeighting} employ rank-dependent utility and prospect theory to evaluate risks in every single period in their models of recursive preferences, so similar to the dynamic quantile model proposed by \citet{deCastro2019:StaticDynamicQuantile}, their models do not lead to time inconsistency. \citet{MaEtal2019:TimeConsistent} and \citet{HeEtal2019:ForwardRDU} study how the probability weighting function in certain preference model evolves over time so as to make the pre-committed strategies consistent over time.

We are only aware of several works that study equilibrium strategies in models with preferences involving probability weighting. \citet{Barberis2012:Casino} study a casino gambling problem in which an agent with prospect theory preferences decides when to stop playing independent and identically distributed bets, and one of the strategies the author considers in his model is intra-personal equilibrium. His work is extended by \citet{EbertStrack2016:NeverEverGettingStarted} to a continuous-time setting. \citet{HuangEtal2017:StoppingBehaviors} study a similar problem but with a different approach. Note that all the above three works study optimal stopping problems in which the control variable is binary, while in the present paper we consider portfolio selection in which the control variable is continuous. \citet{HuEtal2020:ConsistentInvestmentRDU} study portfolio selection for an agent with rank-dependent utility preferences in a complete, continuous-time market where the mean return rates and volatility of the assets are deterministic. The authors consider intra-personal equilibrium for the agent due to the time-inconsistency caused by the agent's preferences. Their notion of intra-personal equilibrium, however, differs from ours. More precisely, they assume that the dollar amount invested in the risky assets by the future selves of the agent remains the same even if the wealth of the future selves changes. 
In our notion, the agent's future selves are assumed to take fixed investment strategies, e.g., to invest 10\% of wealth in the risky assets. As a result, when the agent's self today takes a different investment strategy, the wealth of the agent's future selves changes and, consequently, the dollar amount invested in the risky assets by the future selves also changes. Our notion of intra-personal equilibrium is consistent with the standard definition of equilibrium in the game theory and resembles those used by most studies on time-inconsistent problems in the literature.

The remainder of the paper is organized as follows: In section \ref{se:Model} we propose the model and define the notion of intra-personal equilibrium strategies. In Section \ref{se:MainResults}, we show the intra-personal equilibrium strategy for quantile maximization. In Section \ref{se:Discussion} we discuss the properties of the intra-personal equilibrium strategy for median maximization and compare it with
  fractional Kelly strategies. In Section \ref{se:HouseholdShares}, we apply our model to explain why more wealthy households have larger shares in risky assets.  Finally, Section \ref{se:Conclusions} concludes. In Appendix \ref{comp pre and naive}, we compare the intra-personal equilibrium strategies to the pre-committed and naive strategies under median maximization. In Appendix \ref{se:MultiTerminalDates}, we extend our model to the case in which the agent is concerned about the quantile of her wealth at multiple times.  All proofs are presented in Appendix \ref{se:Proofs}.



\section{Model}\label{se:Model}


\subsection{Market}
Consider an agent who can trade a risk-free asset and $m$  risky assets continuously in the time period $[0,T]$. The price of the risk-free asset, denoted as $S_{0}(t)$, and the price of risky asset $i$, denoted as $S_{i}(t)$, $i=1,\dots,m$, follow
\begin{align*}
  dS_0(t)&= S_0(t) r(t)dt,\quad t\ge 0,\\
  dS_i(t)&=S_i(t)\big[\big(b_i(t)+r(t)\big)dt+\sum_{j=1}^{d}\sigma_{ij}(t)dW_j(t)\big],\quad t\ge 0,\quad i=1,\dots, m,
\end{align*}
where $W(t):=\big(W_1(t),...,W_d(t)\big)\tran$, $t\ge 0$ is a standard, $d$-dimensional Brownian motion, $r(t)$ is the {\em risk-free rate}, $b(t):=\big(b_1(t),\dots, b_m(t)\big)\tran$ is the {\em excess mean return rate vector} of the risky assets, and $\sigma(t):=\big(\sigma_{i,j}(t)\big)$ is the volatility matrix of the risky assets.

Here and hereafter,
for any nonempty interval $I$, denote by $C(I)$ the set of continuous functions from $I$ to certain metric space $\mathbb{B}$, where $\mathbb{B}$ varies with and will be clear in different contexts. For any $c<d$, denote by $C_{\mathrm{pw}}([c,d))$ the set of functions $g$ from $[c,d)$ to $\mathbb{B}$ such that there exists $c=:t_0<t_1<\dots<t_N:=d$ and $g$ on each $[t_{i-1},t_i)$ can be continuously extended to $[t_{i-1},t_i]$, $i=1,\dots, N$.

We make the following assumption throughout the paper:
\begin{assumption}\label{as:Parameters}
  $r$, $b$, and $\sigma$ are deterministic and belong to $C_{\mathrm{pw}}([0,T))$. Moreover, the following two non-degeneracy conditions hold: (i) $b(t)\neq 0$ for all $t\in[0,T)$ and (ii) there exists $\delta>0$ such that $\sigma(t) \sigma(t)\tran-\delta I$ is positive semi-definite for all $t \in[0,T)$, where $I$ stands for the $m$-dimensional identity matrix.
\end{assumption}

Suppose at each time $t$ an agent invests $\pi_i(t)$ dollars in risky asset $i$, $i=1,\dots, m$ and the remaining of her wealth in the risk-free asset. Then, the dynamics of the agent's wealth, denoted as $X(t)$, follow
\begin{align}\label{eq:WealthSDEDollar}
dX(t)=\big(r(t)X(t)+ \pi(t)\tran b(t)\big) dt+\pi(t)\tran \sigma(t)dW(t), \quad t\in[0,T],
\end{align}
where $\pi(t):=\big(\pi_1(t),\dots,\pi_m(t)\big)\tran$ is referred to as the agent's {\em portfolio}. 

\subsection{Portfolio Selection Problem}\label{subse:QuantileMaxProblem}
Suppose that the agent is endowed with initial wealth $x_0$ at time $0$ and wants to maximize the median of her terminal wealth (i.e., wealth at the end time $T$). The agent faces some portfolio constraints, such as the no-short-selling constraint, represented by $Q\pi(t)\ge 0,t\in[0,T)$ for some $n$-by-$m$ matrix $Q$. Suppose that the agent is going to revisit the portfolio decision at each time $t\in[0,T)$, with the same objective of maximizing the median of the terminal wealth. In contrast to expected utility maximization, a portfolio $\pi(s),s\in[0,T)$ that maximizes the median of the terminal wealth at time 0 does not necessarily maximize  the median of the terminal wealth at time $t$, leading to time-inconsistent behavior. To obtain consistent investment behavior, we follow the literature to consider the so-called equilibrium strategies, in which the agent is assumed to have no control of her selves in the future and thus the selves at different time can be viewed as different players in a game.

Formally, we restrict ourselves to consider Markovian portfolio strategies $\bm{\pi}$ that are mappings from $[0,T)\times \bbR$ to $\bbR^m$: At time $t$ with wealth $x$ at that time, the agent invests $\bm{\pi}(t,x)$ dollars in the risky assets. A portfolio strategy $\bm{\pi}$ is {\em feasible} if (i) for any $t\in[0,T)$ and $x\in \bbR$, the following equation
\begin{align}\label{eq:WealthSDEGeneral}
\left\{
\begin{array}{l}
dX^{\bm{\pi}}_{t,x}(s)=\big(r(s)X^{\bm{\pi}}_{t,x}(s)+ \bm{\pi}(s,X^{\bm{\pi}}_{t,x}(s))\tran b(s)\big) ds+\bm{\pi}(s,X^{\bm{\pi}}_{t,x}(s))\tran \sigma(s)dW(s), \; s\in[t,T],\\
X^{\bm{\pi}}_{t,x}(t) =x,
\end{array}
\right.
\end{align}
which represents the agent's wealth process if she starts at time $t$ with wealth $x$ and follows $\bm{\pi}$ to invest, has a unique solution, and (ii) $Q\bm\pi(s,X^{\bm{\pi}}_{0,x_0}(s))\ge 0,s\in[0,T)$. Denote by $\Pi$ the set of feasible portfolio strategies $\bm{\pi}$.
%
%
Denote by
\begin{align}
  F^{\bm{\pi}}(t,x,y):=\prob(X^{\bm{\pi}}_{t,x}(T)\le y), \quad y\in \bbR
\end{align}
the cumulative distribution function of the terminal wealth given wealth level of $x$ at time $t$, and denote by
\begin{align*}
  G^{\bm{\pi}}(t,x,\alpha):=\sup\{y\in \bbR:F^{\bm{\pi}}(t,x,y)\le \alpha\},\quad \alpha\in(0,1)
\end{align*}
the (right-continuous) quantile function of the terminal wealth given wealth level of $x$ at time $t$. In particular, $G^{\bm{\pi}}(t,x,1/2)$ is the median of the terminal wealth given wealth level of $x$ at time $t$.

Suppose that the agent wants to maximize the $\alpha$-level quantile of her terminal wealth. In particular, when $\alpha=1/2$, the agent is a median maximizer. Because quantile maximization leads to time-inconsistency, we consider so-called intra-personal equilibrium. Formally, suppose that we are given a strategy $\hat {\bm{\pi}}\in \Pi$. Denote by $\bbX_t^{x_0,\hat{\bm{\pi}}}$ the set of {\em reachable} wealth levels at time $t$ from the initial wealth $x_0$ at time $0$ and following the strategy $\hat {\bm{\pi}}$; i.e., $\bbX_t^{x_0,\hat{\bm{\pi}}}$ is defined as follows:
\begin{align}\label{support def}
&\bbX_t^{x_0,\hat{\bm{\pi}}}=\mathrm{int}(\mathbb{S}_{X^{\hat{\bm{\pi}}}_{0,x_0}(t)}) \nonumber\\
& \cup \left\{x\in \partial \mathbb{S}_{X^{\hat{\bm{\pi}}}_{0,x_0}(t)} : \prob\big(X^{\hat{\bm{\pi}}}_{0,x_0}(t)\in B_\delta(x)\cap \partial \mathbb{S}_{X^{\hat{\bm{\pi}}}_{0,x_0}(t)}\big)>0\text{ for all }\delta>0\right\},
\end{align}
where $B_\delta(x)$ denotes the ball with radius $\delta$ and centered at $x$, $\mathbb{S}_{X^{\hat{\bm{\pi}}}_{0,x_0}(t)}$ is  the support of $X^{\hat{\bm{\pi}}}_{0,x_0}(t)$  and $\mathrm{int}(\mathbb{S}_{X^{\hat{\bm{\pi}}}_{0,x_0}(t)})$, $\partial \mathbb{S}_{X^{\hat{\bm{\pi}}}_{0,x_0}(t)}$ denote the interior and the boundary of $\mathbb{S}_{X^{\hat{\bm{\pi}}}_{0,x_0}(t)}$ respectively.
\begin{definition}\label{de:EquilibriumStrategy}
$\hat{\bm{\pi}}\in \Pi$ is an intra-personal equilibrium for $\alpha$-level quantile maximization if for any $t\in[0,T)$, $x\in \bbX_t^{x_0,\hat{\bm{\pi}}}$, and $\pi\neq \hat{\bm{\pi}}(t,x)$ with $Q\pi\ge 0$, there exists $\epsilon_0\in (0,T-t)$ such that
\begin{align}\label{eq:EquiDefDollar}
G^{\hat{\bm{\pi}}_{t,\epsilon,\pi}}(t,x,\alpha)-G^{\hat{\bm{\pi}}}(t,x,\alpha) \le 0,\quad \forall \epsilon\in(0,\epsilon_0],
\end{align}
where
\begin{align}\label{eq:PerturbatedPolicy}
  \hat{\bm{\pi}}_{t,\epsilon,\pi}(s,y) = \begin{cases}
\pi, &  s\in[t,t+\epsilon), y\in \mathbb{R}, \\
\hat{\bm{\pi}}(s,y),&s\notin [t,t+\epsilon), y\in \mathbb{R}. \\
\end{cases}
\end{align}
\end{definition}

Imagine that at time $t$, the agent is only able to control herself for a period of length $\epsilon$, so she can choose to invest any dollar amount $\pi$ in the risky assets in the period $[t,t+\epsilon)$ and after that period she is expected follow certain given strategy, e.g., $\hat{\bm{\pi}}$. As a result, the strategy that the agent will actually implement until the end date is $\hat{\bm{\pi}}_{t,\epsilon,\pi}$ as defined by \eqref{eq:PerturbatedPolicy}. Definition \ref{de:EquilibriumStrategy} then stipulates that $\hat{\bm{\pi}}$ is an intra-personal equilibrium if at any time $t$ with any wealth level $x$ that is reachable at that time under $\hat{\bm{\pi}}$, the objective function, namely, the quantile of the terminal wealth, becomes smaller if she chooses an {\em alternative} amount $\pi$ (satisfying the portfolio constraints) to invest in the risky assets, assuming that she is only able to control herself to invest $\pi$ dollars in the risky assets in an {\em infinitesimally small} period of time.

The above definition of equilibrium strategies is so-called regular equilibrium, which slightly differs from the notion of weak equilibrium that
 is used in most studies of continuous-time time-inconsistent problems in the literature.
%
 As explained in \citet{HeJiang2019:OnEquilibriumStrategies}, the notion of regular equilibrium is preferred to the notion of weak equilibrium because the agent can still be willing to deviate from a weak equilibrium strategy and take a very different alternative strategy.
For completeness, we also studied the intra-personal equilibrium strategy under quantile maximization using the notion of weak equilibrium and found that the results in the present paper remain the same. Because the proof is similar to the one for the case of regular equilibrium, we chose not to present it.

Because the discount factor for the period $[0,s]$, namely $e^{-\int_0^sr(u)du}$, is a deterministic function of $s\in[0,T]$, maximizing the quantile of the terminal wealth is equivalent to maximizing the quantile of the discounted terminal wealth. Thus, for notational simplicity, in the following presentation, we set $r\equiv 0$ without loss of generality.

\section{Main Results}\label{se:MainResults}


Let us first present Kelly's portfolio. Assume the following, which stipulates that at each time, one can find a portfolio with a positive instantaneous mean return:
\begin{assumption}\label{as:Feasibility2}
For any $t\in [0,T)$, the set $\{ v\in \mathbb{R}^m\mid b(t)\tran v>0,  Qv \geq 0 \} $ is nonempty and for any $t\in(0,T]$, the set $\{ v\in \mathbb{R}^m\mid b(t-)\tran v>0, Qv \geq 0 \} $ is nonempty, where $b(t-)$ denotes the left-limit of $b$ at $t$.
\end{assumption}

For each $t\in[0,T)$, denote by $v^*(t)$ the optimal solution to following problem
  \begin{align}\label{general small alpha 1}
\left\{
  \begin{array}{cl}
  \underset{v\in\mathbb{R}^m}{\min} &\frac{1}{2} v\tran \sigma(t)\sigma(t)\tran v-b(t)\tran v,\\
  \text{subject to} & Qv\ge 0.
  \end{array}
\right.
\end{align}
 By Lemma \ref{le:QuadProgSolution} in Appendix \ref{se:Proofs}, $v^*\in C_{\mathrm{pw}}([0,T))$.

It is well known that Kelly's portfolio, namely the one that maximizes the expected logarithmic utility of the terminal wealth, is
\begin{align}\label{eq:Kelly}
  \bm{\pi}_{\mathrm{Kelly}}(t,x) = v^*(t)x,\quad t\in [0,T),x\in\bbR.
\end{align}
See for instance \citet{KaratzasIShreveS:98momf}.

We are particularly interested in the following family of {\em affine strategies}:
\begin{align*}
  \mathbb{A} = \big\{&\bm{\pi}\mid \bm{\pi}(t,x) = \theta_0(t) + \theta_1(t) x,\;t\in[0,T),x\in \bbR
   \\ &\text{ for some } \theta_0,\theta_1\in C_{\mathrm{pw}}([0,T))\text{ taking values in }\bbR^m\big\}.
\end{align*}
Note that this family of strategies include many commonly used strategies, such as investing a certain proportion of wealth in the risky assets, investing a wealth-independent dollar amount in the risky assets, and the mixture of the above two strategies taken in different time periods. Note that Kelly's strategy is an affine one.

 \begin{theorem}\label{th:EquilibriumLowerLimit}
 Suppose Assumptions \ref{as:Parameters} and \ref{as:Feasibility2} hold.
 \begin{enumerate}
\item[(i)]
 Suppose $\alpha=1/2$. Then, $\hat{\bm{\pi}}\in \mathbb{A}$ is an intra-personal equilibrium if and only if
 \begin{align}\label{eq:EquiStrategyPropTarget}
\hat{\bm{\pi}}(t,x) =  v^*(t) (x-\xi),  \quad   t\in[0, T), x\in\mathbb{R}
\end{align}
for some constant $\xi<x_0$.
 \item[(ii)] Suppose $\alpha \in (0, 1/2)$. Then,    $\hat{\bm{\pi}}\in \mathbb{A}$ is an intra-personal equilibrium if and only if it implies zero investment in the risky assets at all time, i.e., if and only if
  \begin{align}\label{tail risk equilibrium Strategy}
\hat{\bm{\pi}}(t,x) = \theta(t) (x-x_0),  \quad \forall t\in[0, T], x\in\mathbb{R}
\end{align}
for some $\theta\in C_{\mathrm{pw}}([0,T))$.
 \item[(iii)] Suppose $\alpha \in (1/2, 1)$. Then,   there does not exist any intra-personal
equilibrium in $\mathbb{A}$.
 \end{enumerate}
\end{theorem}

Theorem \ref{th:EquilibriumLowerLimit}-(i) characterizes all affine strategies that are intra-personal equilibria for median maximization. The wealth process $X_{0,x_0}^{\hat{\bm{\pi}}}$ under intra-personal equilibrium \eqref{eq:EquiStrategyPropTarget} satisfies
\begin{align*}
\left\{
\begin{array}{l}
  d(X^{\hat{\bm{\pi}}}_{0,x_0}(t)-\xi) = (X^{\hat{\bm{\pi}}}_{0,x_0}(t)-\xi)\left[v^*(t)\tran b(t)dt + v^*(t)\tran\sigma(t)dW(t) \right],\; t\in[0,T],\\
  X^{\hat{\bm{\pi}}}_{0,x_0}(0)=x_0>\xi.
  \end{array}
  \right.
\end{align*}
Therefore, $X^{\hat{\bm{\pi}}}_{0,x_0}(t)>\xi$ for all $t\in[0,T]$, and the set of reachable wealth levels at time $t$ is $(\xi,+\infty)$ for $t\in(0,T]$. Thus, $\xi$ stands for the guaranteed wealth level, or a {\em portfolio insurance level}.

By definition, revising the value of $\hat{\bm{\pi}}(t,x)$ for $x\notin \bbX_t^{x_0,\hat{\bm{\pi}}}$ changes neither the wealth process $X^{\hat{\bm{\pi}}}_{0,x_0}$ nor whether or not $\hat{\bm{\pi}}$ is an intra-personal equilibrium. Therefore, any $\tilde{\bm{\pi}}$ such that $\tilde{\bm{\pi}}$ agrees with $\hat{\bm{\pi}}$ as given by \eqref{eq:EquiStrategyPropTarget} for $x\in \bbX_t^{x_0,\hat{\bm{\pi}}},t\in[0,T)$, e.g., $\tilde{\bm{\pi}}(t,x) = v^*(t) (x-\xi)^+,t\in[0,T),x\in \bbR$, is also an intra-personal equilibrium.

Theorem \ref{th:EquilibriumLowerLimit}-(i) also shows that there exist multiple intra-personal equilibria for median maximization, parameterized by the portfolio insurance level $\xi$, and the multiplicity here is generic in that the dollar amount invested in the risky assets, $\hat{\bm{\pi}}(t,X^{\hat{\bm{\pi}}}_{0,x_0}(t))$, differs with respect to different values of $\xi$. Multiplicity of intra-personal equilibria for time-inconsistent problems has been noted in the literature both in discrete settings (see e.g., \citealt{VieilleWeibull2009:MultipleSolutions} and \citealt{CaoWerning2018:SavingDissaving}) and in continuous-time settings (see e.g., \citealt{EkelandLazrak2010:GoldenRule} and \citealt{CaoWerning2016:DynamicSavingsDisagreements}).

  Theorem \ref{th:EquilibriumLowerLimit}-(ii) shows that the only intra-personal equilibrium strategy under $\alpha$-quantile maximization with $\alpha<1/2$ is not to invest in the risky assets. On the other hand, Theorem \ref{th:EquilibriumLowerLimit}-(iii) shows that the intra-personal equilibrium strategy does not exist under $\alpha$-quantile maximization with $\alpha>1/2$. To derive some insight of the above results, let us consider constant $b$ and $\sigma$ and restrict to strategies of investing a constant proportion $v$ of the wealth to the risky assets. Then, given wealth $x$ at time $t$, the $\alpha$-quantile of this strategy is
\begin{align*}
 \left(b\tran v -\frac{1}{2}v\tran \sigma \sigma\tran v\right)(T-t) + \sqrt{v\tran \sigma \sigma\tran v}\sqrt{T-t}\Phi^{-1}(\alpha).
\end{align*}
When $T-t$ is sufficiently small, the second term in the above dominates the first term. As a result, when $\alpha<1/2$, we have $\Phi^{-1}(\alpha)<0$, so in order to maximize the $\alpha$-quantile, the agent would minimize $\sqrt{v\tran \sigma \sigma\tran v}$, implying zero investment in the risky assets. When $\alpha>1/2$, we have $\Phi^{-1}(\alpha)>0$, so in order to maximize the $\alpha$-quantile, the agent would maximize $\sqrt{v\tran \sigma \sigma\tran v}$, implying infinite amount of risk taking. Thus, when the agent is concerned about the lower quantile of her wealth (i.e., $\alpha<1/2$), she is extremely conservative so she decides not to invest in the risky assets. When she is concerned about the upper quantile of her wealth (i.e., $\alpha>1/2$), she is aggressive and thus takes infinite amount of risk, implying nonexistence of intra-personal equilibrium. 



 In Definition \ref{de:EquilibriumStrategy}, a portfolio strategy is an intra-personal equilibrium if at any time and any reachable wealth level, the agent is not willing to deviate from it in the sense of condition \eqref{eq:EquiDefDollar}. It is reasonable to exclude wealth levels that are not reachable in the test of whether a strategy is an intra-personal equilibrium because the actions of the agent's future selves at those wealth levels are irrelevant from the perspective of the agent's self today. In the literature on time inconsistency problems, however, such unreachable states are not excluded in the definition of intra-personal equilibrium.\footnote{The only exception is \citet{HeJiang2019:OnEquilibriumStrategies}; see detailed discussions and the relevant references therein.} The following theorem shows that to derive intra-personal equilibrium in our model, it is necessary to exclude unreachable states.

\begin{theorem}\label{prop:BoundaryPoint}
  Suppose Assumptions \ref{as:Parameters} and \ref{as:Feasibility2} hold. Consider any strategy $\tilde{\bm{\pi}}$ such that $\tilde{\bm{\pi}}$ agrees with $\hat{\bm{\pi}}$ as given by \eqref{eq:EquiStrategyPropTarget} for $x\in \bbX_t^{x_0,\hat{\bm{\pi}}},t\in[0,T)$ and that $\tilde{\bm{\pi}}(t,x)$ is continuous in $x$. Then, for any $t\in (0,T)$, there exists $\epsilon_0\in (0,T-t)$ such that $G^{\tilde{\bm{\pi}}_{t,\epsilon,v^*(t)}}(t,\xi,1/2)>\xi = G^{\hat{\bm{\pi}}}(t,\xi,1/2)$ for any $\epsilon\in(0,\epsilon_0)$, where $\tilde{\bm{\pi}}_{t,\epsilon,\pi}$ is defined similarly as in \eqref{eq:PerturbatedPolicy} with $\hat{\bm{\pi}}$ therein replaced by $\tilde{\bm{\pi}}$.
\end{theorem}

  Theorem \ref{prop:BoundaryPoint}  shows that if we mechanically force the condition \eqref{eq:EquiDefDollar} to hold for all wealth levels, even for those that are not reachable, then for the median maximization problem the portfolio strategies we derive in Theorem \ref{th:EquilibriumLowerLimit} are no longer intra-personal equilibrium. Theorem \ref{prop:BoundaryPoint}  also shows why we define the set of reachable wealth levels at each time $t$ to be \eqref{support def} rather than to be the support of $X^{\hat{\bm{\pi}}}_{0,x_0}(t)$: for the strategy \eqref{eq:EquiStrategyPropTarget}, $\xi$ is in the support of $X^{\hat{\bm{\pi}}}_{0,x_0}(t)$, but it does not satisfy the condition \eqref{eq:EquiDefDollar} and is actually not visited by the wealth process.

Finally, due to different life objectives, such as education, kids, and retirement, some investors may concern their wealth at multiple time points that respectively correspond to those objectives; see for instance \citet{Sironi2016:Fintech}. To account for the above multiple objectives, we can extend our model to maximize the weighted average of the quantile of wealth at multiple time points. In this extended model, we derive the same intra-personal equilibrium strategies as those in Theorem \ref{th:EquilibriumLowerLimit}. The details of the extended model are presented in Appendix \ref{se:MultiTerminalDates}.\footnote{
Let us comment that we do not impose any lower bound on the discounted wealth level in our formulation of the portfolio selection problem. Our results, however, still hold if one imposes such a lower bound, e.g., certain $I<x_0$: In this case, Theorem \ref{th:EquilibriumLowerLimit}-(ii) and -(iii) and Theorem \ref{th:QuantileMaximizationMultiple}-(ii) and -(iii) still hold, and Theorem \ref{th:EquilibriumLowerLimit}-(i), Theorem \ref{th:QuantileMaximizationMultiple}-(i), and Theorem \ref{prop:BoundaryPoint} hold by restricting $\xi\in [I,x_0)$.
}

\section{Discussion of Consistent Portfolio Strategies for Median Maximization}\label{se:Discussion}
In this section, we further discuss the intra-personal equilibria for median maximization as given by \eqref{eq:EquiStrategyPropTarget}. To highlight the dependence of this intra-personal equilibria on $\xi$, we denote it as $\hat{\bm{\pi}}_\xi$ in the following.

\subsection{Portfolio Insurance}
\begin{proposition}\label{prop:PortfolioInsurance}
For any $\xi<x_0$,
 \begin{align}\label{eq:MedianEquilibrium}
   G^{\hat{\bm{\pi}}_\xi}(t,x,1/2) = \xi + (x-\xi)e^{\frac{1}{2}\int_t^T\|\sigma(s)\tran v^*(s)\|^2ds},\quad x\in \bbX_t^{x_0,\hat{\bm{\pi}}_\xi},\; t\in [0,T).
 \end{align}
 Moreover, for any $\xi_1<\xi_2<x_0$, $G^{\hat{\bm{\pi}}_{\xi_1}}(t,x,1/2)>G^{\hat{\bm{\pi}}_{\xi_2}}(t,x,1/2)$ for any $t\in[0,T)$ and $x\in \bbX_t^{x_0,\hat{\bm{\pi}}_{\xi_2}}\supseteq \bbX_t^{x_0,\hat{\bm{\pi}}_{\xi_1}}$.
\end{proposition}

The multiplicity of intra-personal equilibria for median maximization raises a question of which equilibrium strategy to choose. Proposition \ref{prop:PortfolioInsurance} shows that with a smaller portfolio insurance level $\xi$, at any time and any wealth level, the median of the terminal wealth becomes strictly larger. On the other hand, with a smaller portfolio insurance level $\xi$, the agent's wealth in the future can reach lower wealth levels. Therefore, $\xi$ becomes a parameter that trades off the growth of the portfolio, which is measured by the median of the terminal wealth, and the risk of the portfolio, which is measured by the lowest level the agent's wealth in the future may touch. As a comparison, the mean-variance portfolio selection problem features a tradeoff between the growth and risk of the portfolio that are measured respectively by the expectation and variance of the portfolio return.

Now, imagine that an investor specifies a maximum amount of loss $L$, e.g., 20\% of the initial wealth, she can tolerate. Moreover, she wants to maximize the median of her terminal wealth and to have a consistent investment plan. Then, portfolio \eqref{eq:EquiStrategyPropTarget} with $\xi=x_0-L$ can be recommended to her.

\subsection{Comparison with Fractional Kelly}
One of the critiques of Kelly's strategy is that it entails too much risk, and to address this issue, the so-called fractional Kelly strategies have been proposed in the literature; see for instance \citet{maclean1992growth}. Formally, fixing $\gamma > 0$, a {\em $\gamma$-fractional Kelly strategy} is defined to be
 \begin{align}\label{eq:FractionalKelly}
   \bm{\pi}_{\gamma-\mathrm{Kelly}}(t,x) = \gamma v^*(t)x,\quad t\in[0,T),x\in \bbR.
 \end{align}
 It is well known that in the market setting in the present paper, the $\gamma$-fractional Kelly strategy is the one that maximizes the expected utility of terminal wealth with a constant relative risk aversion degree $1/\gamma$; see for instance \citet{KaratzasIShreveS:98momf}. For $\gamma\in (0,1)$, the $\gamma$-fractional Kelly strategy leads to less investment in the  risky assets compared to the Kelly strategy.

 Now, for the intra-personal equilibrium $\hat{\bm{\pi}}_\xi$ for median maximization with $\xi\in (0, x_0)$, we have
 \begin{align*}
   \hat{\bm{\pi}}_\xi(t,x)/x = \big((x-\xi)/x\big)v^*(t),\quad x>\xi.
 \end{align*}
 Therefore, compared to Kelly's strategy, $\hat{\bm{\pi}}_\xi$ implies less investment in the  risky assets because $(x-\xi)/x<1$. In the following, we compare the intra-personal equilibrium $\hat{\bm{\pi}}_\xi$ with the fractional Kelly strategy in terms of their growth and risk.

\begin{proposition}\label{prop:CompwithFracKelly}
 For any $\gamma>0$,
 \begin{align}\label{eq:MedianFractionalKelly}
   G^{\bm{\pi}_{\gamma-\mathrm{Kelly}}}(t,x,1/2) = xe^{(\gamma-\frac{1}{2}\gamma^2)\int_t^T\|\sigma(s)\tran v^*(s)\|^2ds},\quad x>0,\; t\in [0,T).
 \end{align}
 Moreover, for fixed $\xi\in (0,x_0)$, $\gamma\in(0,1)$, and $t\in[0,T)$, we have
 \begin{align}
   a_{t,\gamma}:=\frac{e^{\frac{1}{2}\int_t^T\|\sigma(s)\tran v^*(s)\|^2ds}-1}{e^{\frac{1}{2}\int_t^T\|\sigma(s)\tran v^*(s)\|^2ds} -e^{(\gamma-\frac{1}{2}\gamma^2)\int_t^T\|\sigma(s)\tran v^*(s)\|^2ds} }>1,
 \end{align}
 and $G^{\bm{\pi}_{\gamma-\mathrm{Kelly}}}(t,x,1/2)$ is strictly larger than (strictly smaller than, respectively) $G^{\hat{\bm{\pi}}_\xi}(t,x,1/2)$ if and only if $x<a_{t,\gamma}\xi$ ($x>a_{t,\gamma}\xi$, respectively).
\end{proposition}

Proposition \ref{prop:CompwithFracKelly} shows that at time $t\in(0,T)$, neither of the intrapersonal equilibrium $\hat{\bm{\pi}}_{\xi}$ and the fractional Kelly strategy dominates the other in terms of median of the terminal wealth: the former implies a higher median of the terminal wealth than the latter when the wealth level is high and vice versa when the wealth level is low. At the initial time (with initial wealth $x_0$), which of the above two strategies imply a higher median of the terminal wealth depends on the value of $\xi$ and $\gamma$. On the risk side, the intrapersonal equilibrium $\hat{\bm{\pi}}_{\xi}$ entails less risk than the fractional Kelly strategy in that the former implies a higher level of minimum wealth. Finally, as implied by Theorem \ref{th:EquilibriumLowerLimit}, the fractional Kelly strategy (except for the case $\gamma=1$) is {\em not} an intrapersonal equilibrium for median maximization; i.e., it is an inconsistent investment strategy for median maximization.

\subsection{Expected Utility Maximization}

\begin{proposition}\label{eq:EUMaximization}
  The intra-personal equilibrium $\hat{\bm{\pi}}_\xi$ for median maximization is the optimal portfolio strategy that maximizes $\expect[\ln (X^{\bm{\pi}}(T)-\xi)]$.
\end{proposition}

Proposition \ref{eq:EUMaximization} shows that the intra-personal equilibrium $\hat{\bm{\pi}}_\xi$ under median maximization is the same as the portfolio that maximizes the expected utility of the terminal wealth with utility function $\log(x-\xi)$. This utility function is a special case of the so-called hyperbolic absolute risk aversion (HARA); see for instance Section 6 of \citet{MertonR:71oc} and \citet{KimOmberg}. It is worth emphasizing that we are only able to prove the above equivalence between the intra-personal equilibrium under median maximization and expected utility maximization with a logarithmic utility function in the market setting in the present paper. 

Although median maximization and expected utility maximization with the above particular utility function yield the same portfolio strategies, these two models differ. Median of a distribution is easy to understand while utility functions and thus the expected utility of a distribution are not. One can apply median maximization without calibration while we need to infer the utility function before we can employ the expected utility theory. Thus, median maximization is a more accessible model than expected utility maximization for investors. Moreover, the parameter $\xi$ in the intra-personal equilibrium under median maximization is not a priori given. It represents a portfolio insurance level and trades off the growth and risk of investment. In the expected utility theory, however, the parameter $\xi$ is exogenously given and does not have an economic meaning. It needs to be calibrated from investors' preferences under risk.

%

\section{Application: Household Portfolio Shares}\label{se:HouseholdShares}
\citet{WachterYogo2010:Household} find that households with a higher level of wealth tend to have a larger portfolio shares in risky assets (i.e., to invest more percentage of their wealth in risky assets). More precisely, in one of their studies, the authors consider a representative sample, provided by the Board of Governors of the Federal Reserve System, of approximately 3,000 households. In this sample, the net worth, i.e., the wealth, and the portfolio share in risky assets of each household is observed. The authors conduct linear regression with the log net worth as the explanatory variable and the portfolio share in risky assets as the dependent variable in the cross-section of households. The authors divide the households in the sample into four age groups: 26--35, 36--45, 56--65, 66--75, and find that the portfolio share in risky assets is more positively correlated with the log net worth for elder age groups; see Table 4 of \citet{WachterYogo2010:Household}.

The above empirical finding cannot be explained by expected utility maximization with a power utility function even when the asset prices have stochastic volatility, stochastic return rates, or jumps, because the resulting optimal percentage of wealth invested in risky assets is independent of the agent's wealth. The mean-variance log return model proposed by \citet{dai2020dynamic} cannot explain the finding either because in their model the percentage of wealth invested in risky assets is independent of the agent's wealth as well. \citet{WachterYogo2010:Household} develop a life-cycle consumption and portfolio choice model to explain this empirical finding.

The intra-personal equilibrium under median maximization is consistent with the above empirical finding,
 so our model, which is simpler than the life-cycle consumption and portfolio choice model proposed by \citet{WachterYogo2010:Household}, is an alternative to explain this empirical finding. Indeed, suppose
that the households follow $\hat{\bm{\pi}}_\xi$ for some $\xi>0$ to do investment. Then, when households become older, their wealth $x$ also becomes larger on average, so the portfolio shares in risky assets, which is $(x-\xi)/x$ under $\hat{\bm{\pi}}_\xi$, also become larger. To confirm the above intuition, we conduct a numerical analysis in the following.

Suppose that there are 3,000 households in the market, indexed by $j=1,\dots, 3000$, and their age is between 26 and 35. Each household $j$ is endowed with initial wealth $x_{0,j}$ and has a portfolio insurance level $\beta$, e.g., 60\%, proportion of her initial wealth, i.e., $\xi_{j} = \beta x_{0,j}$, $j=1,\dots, 3,000$. As a result, the portfolio shares in risky assets for household $j$ are $(x_{0,j}-\xi_j)/x_{0,j} = 1-\beta$, which are the same for different households. Now, imagine that after $t$ years, household $j$'s investment in risky assets generate a gross return rate $R_{t,j}$, so her wealth becomes $X_{t,j}=\xi_{j} + (x_{0,j}-\xi_j)R_{t,j}$ and thus her portfolio shares in risky assets becomes
\begin{align*}
p_{t,j}:=\frac{\xi_{j} + (x_{0,j}-\xi_j)R_{t,j}-\xi_j}{\xi_{j} + (x_{0,j}-\xi_j)R_{t,j}} = \frac{(1-\beta)R_{t,j}}{\beta + (1-\beta)R_{t,j}}.
\end{align*}
Note that after $t=10$, 20, 30, and 40 years, the households' age become 36--45, 46--55, 56--65, and 66--75, respectively. Thus, for each $t\in \{10,20,30,40\}$, we follow \citet{WachterYogo2010:Household} to run linear regression with $\ln X_{t,j}$ as the explanatory variable and $p_{t,j}$ as the dependent variable.

We simulate the initial wealth of the 3,000 households by setting $x_{0,j}$ to be the $j$-th sample of $\bar x_0 e^{\varrho U}$, where $U$ is a standard normal random variable. Thus, $\ln \bar x_0$ and $\varrho$ represents respectively the average log net worth and the standard deviation of the log net worth across households with in the age group 26--35. We use the sample provided by Survey of Consumer Finances that tracks the wealth of US households every three years from 1989 to 2016 to estimate $\bar x_0$ and $\varrho$.\footnote{The sample is available at \url{https://www.federalreserve.gov/econres/scfindex.htm}.} Following the study in \citet{WachterYogo2010:Household}, we exclude households with non-positive net worth or with no risky-asset holding from the sample. Using the data in 2016, we obtain the following estimates: $\bar x_0=61811.8$ and $\varrho=0.0569$.

On the other hand, according the survey data by Survey of Consumer Finances from the same source as above, the percentage of net worth invested in risky assets for households in the age group 26--35 ranges from $38\%$ to $58\%$ across different survey years. Thus, we set the value of $1-\beta$ to be in the range $40\%$--$60\%$, so we choose three values of $\beta$: $40\%$, $50\%$ and $60\%$.

We simulate the gross return rate $R_{t,j}$ of the 3,000 households from the following distribution: $(1+\mu t)e^{-\frac{1}{2} \varpi^2 t+ \varpi \sqrt{t}Z }$, where $Z$ is a standard normal random variable that is independent of $U$, $\mu=4\%$, and $\varpi>0$. In other words, we set the average excess return rate per year across households to be 4\%, and $\varpi$ measures the standard deviation of the annual log return rate across households. Because we do not have the data to estimae $\varpi$, we simply choose three values of $\varpi$ in the following: 0.65\%, 0.70\%, and 0.75\%.

%

Finally, we run linear regression with $X_{t,j}$ to be the explanatory variable and $p_{t,j}$ to be the dependent variable to obtain the coefficient of $\ln X_{t,j}$. We repeat the simulation for 2,000 times and report the mean and standard deviation (in parentheses) of the coefficient of $\ln X_{t,j}$ in Table \ref{ta:PortfolioShares}, where different rows and columns refer to different values of $t\in \{10,20,30,40\}$ (corresponding to age 36--45, 46--55, 56--65, and 66--75, respectively) and  $\varpi\in \{0.65\%,0.70\%,0.75\%\}$. We can see that the coefficient is indeed more positive for larger with $t$, which is consistent with the empirical finding in \citet{WachterYogo2010:Household}.

\begin{table}
	\begin{center}
		\begin{tabular}{|l|l|c|c|c|c|c|}
			\hline
			\multicolumn{2}{|c|}{} & 26--35 & 36--45 & 46--55 & 56--65 & 66--75 \\ \hline
			& $\beta=40\%$ & 0 & 1.26 (0.09) & 2.28 (0.11) & 3.02 (0.11) & 3.52 (0.11) \\
			\cline{2-7}
			$\varpi=0.65 \% $ & $\beta=50\%$ & 0 & 1.23 (0.11) & 2.41 (0.13) & 3.38 (0.15) & 4.13 (0.14) \\
			\cline{2-7}
			& $\beta=60\%$ & 0 & 1.06 (0.11) & 2.27 (0.15) & 3.41 (0.17) & 4.39 (0.18) \\
			\hline
			& $\beta=40\%$  & 0 & 1.45 (0.10) & 2.60 (0.12) & 3.40 (0.11) & 3.92 (0.12) \\
			\cline{2-7}
			$\varpi=0.70 \%$ & $\beta=50\%$ & 0 & 1.41 (0.11) & 2.74 (0.14) & 3.82 (0.15) & 4.62 (0.16) \\
			\cline{2-7}
			& $\beta=60\%$ & 0 & 1.22 (0.12) & 2.60 (0.16) & 3.88 (0.18) & 4.96 (0.19) \\
			\hline
			& $\beta=40\%$ & 0 & 1.65 (0.11) & 2.93 (0.13) & 3.79 (0.12) & 4.32 (0.12) \\
			\cline{2-7}
			$\varpi=0.75 \%$ & $\beta=50\%$ & 0 & 1.61 (0.12) & 3.10 (0.15) &  4.28 (0.16) & 5.12 (0.16) \\
			\cline{2-7}
			& $\beta=60\%$ & 0 & 1.39 (0.12) & 2.95 (0.17) & 4.37 (0.19) & 5.53 (0.19) \\
			\hline
			\multicolumn{2}{|c|}{\citet{WachterYogo2010:Household}} &  0.52 & 1.84 & 3.56 & 3.88 & 4.32 \\
			\hline
		\end{tabular}
	\end{center}
	\caption{\small Sensitivity (in percentage) of portfolio shares in risky assets $p_{t,j}$ with respect to wealth $\ln X_{t,j}$ in the cross-section of households. The second to sixth columns refer the age groups 26--35, 36--45, 46--55, 56--65, 66--75, respectively, which correspond to $t=0$, 10, 20, 30, and 40, respectively. The number of households is 3,000, and their initial net worth is simulated from $\bar x_0 e^{\varrho U}$ with $\bar x_0=61811.8$ and $\varrho=0.0569$, where $U$ is a standard normal random variable. The value of $\beta$ is set to be 40\%, 50\%, and 60\%. The gross return rate of the households are simulated from $(1+\mu t)e^{-\frac{1}{2} \varpi^2 t+ \varpi \sqrt{t}Z }$ with $\mu=4\%$ and $\varpi$ to one of the three values: =0.65\%, 0.70\%, and 0.75\%, where $Z$ is a standard normal random variable independent of $U$. The numbers in parentheses are standard error of the estimates of the sensitivities. The last row reports the sensitivity of portfolio shares in risky assets with respect to net worth in the empirical study of \citet[Table 4]{WachterYogo2010:Household}.
	}
	\label{ta:PortfolioShares}
\end{table}

\section{Conclusions}\label{se:Conclusions}
Median is a popular alternative to mean as a summary statistic of a distribution. In this paper, we studied portfolio selection under the $\alpha$-quantile maximization, particularly under median maximization when $\alpha$ is set to be $1/2$. We considered an agent who trades a risk-free asset and multiple  risky assets continuously in time with an objective of maximizing the $\alpha$-quantile of her wealth at certain terminal time, and the mean return rates and volatility of the assets are assumed to be deterministic. Because of time inconsistency, we considered intra-personal equilibrium strategies.

We found that in the class of time-varying, affine portfolio strategies, the intra-personal equilibrium does not exist  when $\alpha>1/2$ and leads to zero investment in the risky assets when $\alpha<1/2$. For the case of $\alpha=1/2$, namely the case of median maximization, a time-varying affine strategy is an intra-personal equilibrium strategy if and only if it is a portfolio insurance strategy. Different choices of the portfolio insurance level then induce different intra-personal equilibria and can be interpreted as different degrees of risk attitude of the agent.

We compared the intra-personal equilibrium strategy under median maximization to fractional Kelly strategies   and showed that the former is better in terms of trading off growth and risk.  We also showed that the intra-personal equilibrium strategy can explain why households with a higher level of wealth tend to invest more percentage of their wealth in risky assets.   Finally, in the Appendices, we compared the intra-personal equilibrium strategy with the pre-committed and naive strategies and extended our model to the case in which the agent is concerned about median of wealth at multiple times and derived similar results.

\appendix

\section{Comparison with Pre-Committed and Naive Strategies}\label{comp pre and naive}

When facing time inconsistency, some individuals may commit their future selves to follow the plans they set up today that are optimal under today's decision criteria, and such plans are called pre-committed strategies. For instance, one can delegate her investment to a portfolio manager and asks the manager to maximize her decision criterion today. In the following, we compare the intra-personal equilibrium $\hat{\bm{\pi}}_\xi$ under median maximization, which is a rational choice of an agent who is not able to commit her future selves to following her plan today, with the pre-committed strategy under median maximization. To facilitate the comparison, we assume the same portfolio insurance level $\xi$ in the derivation of the pre-committed strategy.

\begin{proposition}\label{prop:PreCommitted}
\begin{enumerate}
  \item[(i)]The pre-committed strategy $\bm{\pi}_{0,\mathrm{pc}}$ that maximizes the time-0 median of the terminal wealth is
\begin{align}
  \bm{\pi}_{0,\mathrm{pc}}(t,x) &= \Delta_{0,\mathrm{pc}} (t,x) v^*(t)(x-\xi),\quad t\in[0,T),x>\xi,\label{eq:PreCommittedStrategy}\\
  \Delta_{0,\mathrm{pc}} (t,x):&=\frac{1}{\sqrt{\int_t^T\|\sigma(\tau)\tran v^*(\tau)\|^2d\tau}}\times \frac{\Phi'\big(d(t,\bm{z}_0(t,x))\big)}{\Phi\big(d(t,\bm{z}_0(t,x))\big)},\notag\\
  d(t,z):&=\frac{-\int_t^T \|\sigma(\tau)\tran v^*(\tau)\|^2d\tau+z}{\sqrt{\int_t^T\|\sigma(\tau)\tran v^*(\tau))\|^2d\tau}}\notag
\end{align}
with $\bm{z}_0(t,x)$ uniquely determined by
\begin{align*}
  \frac{x-\xi}{\Phi\big(d(t,\bm{z}_0(t,x))\big)} = \frac{x_0-\xi}{\Phi\left(-\sqrt{\int_0^T\|\sigma(\tau)\tran v^*(\tau)\|^2d\tau}\right)}=:k_0^*.
\end{align*}
Moreover,
\begin{align*}
  X_{0,x_0}^{\bm{\pi}_{0,\mathrm{pc}}}(T) = \xi + k_0^*\mathbf{1}_{\int_0^Tv^*(\tau)\tran \sigma(\tau)dW(\tau)\ge 0},
\end{align*}
and $\Delta_{0,\mathrm{pc}} (t,x)$ is strictly decreasing, continuous in $x$ and satisfies
\begin{align*}
  \lim_{x\downarrow \xi}\Delta_{0,\mathrm{pc}} (t,x) = +\infty,\quad \lim_{x\uparrow \xi+k_0^*}\Delta_{0,\mathrm{pc}} (t,x)=0.
\end{align*}
\item[(ii)] $G^{\bm{\pi}_{0,\mathrm{pc}}}(0,x_0,1/2) =  \xi + k_0^*$, and
\begin{align*}
   G^{\bm{\pi}_{0,\mathrm{pc}}}(t,x,1/2) = \begin{cases}
      \xi + k_0^*, & x\in [x_0,\xi+k_0^*),\\
     \xi, & x\in (\xi, x_0),
   \end{cases}
   \quad t\in (0,T).
\end{align*}
In addition, $G^{\bm{\pi}_{0,\mathrm{pc}}}(0,x_0,1/2)>G^{\hat{\bm{\pi}}_\xi}(0,x_0,1/2)$, and for any $t\in (0,T)$,
\begin{align}
    & G^{\bm{\pi}_{0,\mathrm{pc}}}(t,x,1/2)
   \begin{cases}
     \ge G^{\hat{\bm{\pi}}_\xi}(t,x,1/2), & x_0\le x \le \xi + (x_0-\xi)\tilde a_t,\\
     < G^{\hat{\bm{\pi}}_\xi}(t,x,1/2), & x\in (\xi,x_0)\cup(\xi + (x_0-\xi)\tilde a_t,\xi+k_0^*),
   \end{cases}\label{eq:EquiPreComMedianComparison}
   \\
  &\tilde a_t :=e^{-\frac{1}{2}\int_t^T\|\sigma(s)\tran v^*(s)\|^2ds}/\Phi\left(-\sqrt{\int_0^T\|\sigma(\tau)\tran v^*(\tau)\|^2d\tau}\right) \in \big(1,k_0^*/(x_0-\xi)\big).\notag
\end{align}
\end{enumerate}
\end{proposition}

Proposition \ref{prop:PreCommitted}-(i) shows the pre-committed portfolio strategy under median maximization. Recall that under the intra-personal equilibrium $\hat{\bm{\pi}}_\xi$, the agent's dollar amount invested in the risky assets is proportional to the distance between the current wealth and portfolio insurance level, and the proportion is independent of the current wealth level. For the pre-committed portfolio strategy, however, this proportion depends on the current wealth level and can become arbitrarily large when the wealth approaches the portfolio insurance level. In terms of the wealth process, under the intra-personal equilibrium, the agent has potential to attain arbitrarily high wealth levels in the future, but under the pre-committed portfolio strategy, the wealth in the future is capped at certain level $\xi+k_0^*$.

The pre-committed portfolio strategy obviously implies a higher level of time-0 median of the terminal wealth than the intra-personal equilibrium. After the initial time, i.e., at time $t\in (0,T)$, however, the pre-committed portfolio strategy results in smaller median of the terminal wealth than the intra-personal equilibrium when the wealth level at that time is very low or very high; see \eqref{eq:EquiPreComMedianComparison}.

It can be costly to implement pre-committed strategies; for instance, portfolio delegation usually incurs some management fees. In some situations, individuals can be unaware of the time-inconsistency or wrongly believe that they can commit their future selves to the plan set up today. As a result, they may keep re-optimizing and updating their plans over time. In the extreme case, at each instant an agent can only implement her plan for an infinitesimally small time period and re-optimizes and updates the plan afterwards. The resulting strategy that is actually implemented by the agent over time is called the {\em naive strategy}.

\begin{proposition}\label{prop:NaiveStrategy}
The naive portfolio strategy $\bm{\pi}_{\mathrm{na}}$ under median maximization is given by
\begin{align}\label{eq:NaiveStrategy}
  \bm{\pi}_{\mathrm{na}}(t,x) = \Delta_{\mathrm{na}}(t) v^*(t)(x-\xi),
\end{align}
where
\begin{align}
  \Delta_{\mathrm{na}}(t):=\frac{1}{\sqrt{\int_t^T\|\sigma(s)\tran v^*(s)\|^2ds}}\times \frac{\Phi'\left(-\sqrt{\int_t^T\|\sigma(s)\tran v^*(s)\|^2ds}\right)}{\Phi\left(-\sqrt{\int_t^T\|\sigma(s)\tran v^*(s)\|^2ds}\right)}.
\end{align}
Moreover, $\Delta_{\mathrm{na}}(t)>1,t\in [0,T)$ and $\lim_{t\uparrow T}\Delta_{\mathrm{na}}(t)=+\infty$. Furthermore, for any fixed $t\in [0,T)$ and $x>\xi$,
denoting by $G^{\bm{\pi}_{\mathrm{na}}}(t,x,1/2)$ the limit of the median, conditional on time-$t$ wealth level of $x$, of the wealth at time $\tau$ as $\tau$ goes to $T$, we have $G^{\bm{\pi}_{\mathrm{na}}}(t,x,1/2)=\xi$.
\end{proposition}

Proposition \ref{prop:NaiveStrategy} shows that under the naive strategy, the dollar amount invested in the risky assets is also proportional to the distance between the current wealth level and the portfolio insurance level. Moreover, the proportion is always strictly larger than 1, implying higher risky asset holdings than the intra-personal equilibrium. The proportion even goes to infinity when it is near the terminal time, showing that under the naive strategy, the agent would take an infinite amount of risk around the terminal time.

Because under the naive strategy the agent invests infinite amount of money in risky assets around the terminal time, the terminal wealth of the naive strategy is not well defined. We, however, can still study the wealth around the terminal time, in particular the median of the wealth at time $\tau$ when $\tau$ is very close to the terminal time. It turns out that the limit of the median exists when $\tau$ goes to the terminal time $T$, and the limit is $\xi$. This shows that in terms of the median of terminal wealth, the naive strategy always underperforms the intra-personal equilibrium. The reason is because under the naive strategy the agent takes an infinite amount of risk around the terminal time, which significantly reduces the median of the terminal wealth.

\section{Multiple Target Dates}\label{se:MultiTerminalDates}

Suppose that the agent is concerned about the her wealth at not only the terminal time $T$ but also some intermediate moments. More precisely, consider multiple time points $0=:T_0<T_1<\dots <T_{N}:=T$. At time $t\in [T_{n-1},T_{n})$ with wealth $x$, the agent's decision criterion at that time is a weighted average of the $\alpha$-level quantile of her wealth at $T_{n},T_{n+1},\dots, T_{N}$, i.e., is
\begin{align}\label{ObjMultipleT}
  J^{\bm{\pi}}(t,x,\alpha) = \sum_{i=n}^{i=N} w_{n,i}G^{\bm{\pi}}(t,x,\alpha;T_{i}),
\end{align}
where $G^{\bm{\pi}}(t,x,\alpha;T_i)$ stands for the $\alpha$-level quantile of $X^{\bm{\pi}}_{t,x}(T_i)$, $w_{n,i}\ge 0,i=n,\dots, N$ are constants and satisfy $\sum_{i=n}^Nw_{n,i}=1$, and $w_{n,N}>0$.

\begin{definition}\label{de:EquilibriumStrategyMultipleTime}
$\hat{\bm{\pi}}\in \Pi$ is an intra-personal equilibrium for multi-time $\alpha$-level quantile maximization if for any $n=1,\dots, N$, $t\in[T_{n-1},,T_{n})$, $x\in \bbX_t^{x_0,\hat{\bm{\pi}}}$, and $\pi\neq \hat{\bm{\pi}}(t,x)$ with $Q\pi\ge 0$, there exists $\epsilon_0\in (0,T-t)$ such that
\begin{align}\label{eq:EquiDefDollarMultipleTime}
J^{\hat{\bm{\pi}}_{t,\epsilon,\pi}}(t,x,\alpha) -J^{\hat{\bm{\pi}}}(t,x,\alpha) \le 0,\quad \forall \epsilon\in(0,\epsilon_0],
\end{align}
where $\hat{\bm{\pi}}_{t,\epsilon,\pi}$ is given by \eqref{eq:PerturbatedPolicy}.
\end{definition}

  \begin{theorem}\label{th:QuantileMaximizationMultiple}
 Suppose Assumptions \ref{as:Parameters} and \ref{as:Feasibility2} hold.
 \begin{enumerate}
\item[(i)]
 Suppose $\alpha=1/2$. Then, $\hat{\bm{\pi}}\in \mathbb{A}$ is an intra-personal equilibrium for multi-time $\alpha$-level quantile maximization if and only if $\hat{\bm{\pi}}$ is given by \eqref{eq:EquiStrategyPropTarget} for some constant $\xi<x_0$.
 \item[(ii)] Suppose $\alpha \in (0, 1/2)$. Then,     $\hat{\bm{\pi}}\in \mathbb{A}$ is an intra-personal equilibrium for multi-time $\alpha$-level quantile maximization if and only if it implies zero investment in the risky assets at all time, i.e., if and only if $\hat{\bm{\pi}}$ is given by \eqref{tail risk equilibrium Strategy} for some $\theta\in C_{\mathrm{pw}}([0,T))$.
 \item[(iii)] Suppose $\alpha \in (1/2, 1)$. Then,   there does not exist any intra-personal equilibrium for multi-time $\alpha$-level quantile maximization in $\mathbb{A}$.
 \end{enumerate}
\end{theorem}

Theorem \ref{th:QuantileMaximizationMultiple} shows that the intra-personal equilibrium for multi-time quantile maximization is the same as for single-time-point quantile maximization. In particular, for median maximization, although the decision criterion is discontinuous at each time point $T_i$, the portfolio insurance level remains constant over time.

\section{Proofs}\label{se:Proofs}

\subsection{A Lemma}
\begin{lemma}\label{le:QuadProgSolution}
  Suppose Assumptions \ref{as:Parameters} and \ref{as:Feasibility2} hold. For each fixed $t\in [0,T)$, $v^*(t)\neq 0$ and $b(t)\tran v^*(t) = \|\sigma(t)\tran v^*(t)\|^2>0$. Consequently, $v^*\in C_{\mathrm{pw}}([0,T))$, $\inf_{t\in[0,T)}\|v^*(t)\|>0$, and $\inf_{t\in[0,T)}\|\sigma(t)\tran v^*(t)\|>0$.
\end{lemma}

\begin{pfof}{Lemma \ref{le:QuadProgSolution}}
By Assumption \ref{as:Parameters}, there exists $0=:t_0<t_1<\dots<t_N:=T$ such that $\sigma$ and $b$ are continuous on $[t_{i-1},t_i)$ and can be continuously extended to $[t_{i-1},t_i]$, $i=1,\dots, N$. Thus, in the following, we only need to fix $i$ and consider the continuous extension of $b$ and $\sigma$ on $[t_{i-1},t_i]$.

  Fix any $t\in [t_{i-1},t_{i}]$. It is obvious that the optimal solution of \eqref{general small alpha 1} uniquely exists. Moreover, by Assumption \ref{as:Feasibility2}, there exists $v_0\in \mathbb{R}^m$ with $Qv\ge 0$ and $b(t)\tran v_0>0$. For sufficiently small $\epsilon>0$, $v_\epsilon:=\epsilon v_0$ satisfies $Qv_\epsilon\ge 0$ and $\frac{1}{2}v_\epsilon\tran \sigma(t)\sigma(t)\tran v_\epsilon-b(t)\tran v_\epsilon = \epsilon\big(\frac{1}{2}v_0\tran \sigma(t)\sigma(t)\tran v_0\epsilon -b(t)\tran v_0\big)<0$, so the optimal solution of \eqref{general small alpha 1}, namely $v^*(t)$, cannot be 0.

  The Lagrange dual theory implies the follow equations:
  \begin{align}\label{eq:LeQuadProgEq1}
    \left\{
    \begin{array}{l}
    \sigma(t)\sigma(t)\tran v^*(t) - b(t) - Q\tran \lambda = 0,\\
    \lambda\tran Q v^*(t) = 0,\\
    \lambda\ge 0,Qv^*(t)\ge 0.
    \end{array}
    \right.
  \end{align}
  Multiplying by $v^*(t)\tran$ from left on both sides of the first equation of \eqref{eq:LeQuadProgEq1} and recalling the second equation of \eqref{eq:LeQuadProgEq1}, we immediately conclude that $b(t)\tran v^*(t) =\|\sigma(t)\tran v^*(t)\|^2$. Moreover, because $v^*(t)\neq 0$ and $\sigma(t)\sigma(t)\tran$ is positive definite, we have $\|\sigma(t)\tran v^*(t)\|^2>0$.

    Finally, by Assumption \ref{as:Parameters} and Theorem 4.4 in \citet{daniel1973stability}, we immediately conclude that $v^*\in C([t_{i-1},t_i])$. As a result, $\inf_{t\in[t_{i-1},t_i]}\|v^*(t)\|>0$ and $\inf_{t\in[t_{i-1},t_i]}\|\sigma(t)\tran v^*(t)\|^2>0$.
%
%
%
%
%
%
  \halmos
\end{pfof}

\subsection{Proof of Theorem \ref{th:EquilibriumLowerLimit}}
We introduce some notations to be used in the following proof.

For any interval $[a,b)$ and open set $O$ in $\bbR^l$, denote by $C^{0,\infty}([a,b)\times O)$ the set of functions $g(t,z)$ from $[a,b)\times O$ to $\bbR$ such that its derivatives with respect to $z$ of any order exist and are continuous in $(t,z)$ on $[a,b)\times O$ and by $C^{1,\infty}([a,b)\times O)$ the set of functions $g(t,z)$ from $[a,b)\times O$ to $\bbR$ such that its first-order derivative with respect to $t$ and its derivatives with respect to $z$ of any order exist and are continuous in $(t,z)$ on $[a,b)\times O$. 

For any $x\in \bbR^l$ and $\delta\ge 0$, denote by $B^\ell_\delta(x):=\{y\in \bbR^l\mid \|y-x\|\le \delta\}$.

The proof of Theorem \ref{th:EquilibriumLowerLimit} is divided into three parts: In Section \ref{subsubse:CalDerive}, we prove a crucial lemma. In Section \ref{subsubse:Sufficiency}, we prove the sufficiency part of the theorem. In Section \ref{subsubse:Necessity}, we prove the remaining part of the theorem. Because the proof is involved, we present it by summarizing important intermediate steps of the proofs as lemmas and relegate all proofs in Section \ref{subsubse:proofs}.

\subsubsection{Calculation of Derivatives}\label{subsubse:CalDerive}
\begin{lemma}\label{le:QuantileDifferentiability}
  Consider any $\hat{\bm{\pi}}\in \mathbb{A}$, i.e., $\hat{\bm{\pi}}(t,x) = \theta_0(t) + \theta_1(t) x,\; t\in[0,T),x\in \bbR$ with $\theta_0,\theta_1\in C_{\mathrm{pw}}([0,T))$, and define
  \begin{align}
    t^*:&=\inf\{t\in[0,T):\theta_0(s)=\theta_1(s)=0,\forall s\in[t,T)\},\label{eq:T0}\\
    t_*:&=\inf\{t\in[0,t^*):\theta_0(s)+\xi\theta_1(s)=0,\forall s\in[t,t^*)\text{ and some }\xi \in \bbR\}.\label{eq:T0prime}
  \end{align}
  Then, if $t_*<t^*$, there exists unique $\xi\in \bbR$ such that $\theta_0(s)+\xi\theta_1(s)=0,\forall s\in[t_*,t^*)$. Denote
   \begin{align}\label{eq:SingularPoints}
     \mathbb{S}^{\hat{\bm{\pi}}}_t=\emptyset,\; t\in [0,t_*),\quad  \mathbb{S}^{\hat{\bm{\pi}}}_t=\{\xi\},\; t\in[t_*,T).
   \end{align}
   Then, for any $t\in [0,t^*)$, there exists $\eta\in (0,t^*-t)$ such that for any $x \in \bbR\backslash \mathbb{S}^{\hat{\bm{\pi}}}_t$,
  $\alpha \in(0,1)$, and $\pi\in \bbR^m$, $G^{\hat{\bm{\pi}}_{t,\epsilon,\pi}}(t,x,\alpha)$ is continuous in $\epsilon \in[0,\eta)$, $F^{\hat{\bm{\pi}}}_y(t,x,G^{\hat{\bm{\pi}}}(t,x,\alpha))>0$, and
  \begin{align}
    &\lim_{\epsilon\downarrow 0}\frac{G^{\hat{\bm{\pi}}_{t,\epsilon,\pi}}(t,x,\alpha) - G^{\hat{\bm{\pi}}}(t,x,\alpha)}{\epsilon}
    \notag \\
    &= -\frac{{\cal A}^{\pi}F^{\hat{\bm{\pi}}}(t,x,G^{\hat{\bm{\pi}}}(t,x,\alpha))}{F^{\hat{\bm{\pi}}}_y(t,x,G^{\hat{\bm{\pi}}}(t,x,\alpha))} =\frac{\varphi^{\hat{\bm{\pi}}}_{t,x,\alpha}(\hat{\bm{\pi}}(t,x)) - \varphi^{\hat{\bm{\pi}}}_{t,x,\alpha}(\pi)}{F^{\hat{\bm{\pi}}}_y(t,x,G^{\hat{\bm{\pi}}}(t,x,\alpha))},
    \label{eq:QuanDerivativeWRTeps}
  \end{align}
  where ${\cal A}^{\pi}$ is applied to $F^{\hat{\bm{\pi}}}(t,x,y)$ as a function of $(t,x)$ with ${\cal A}^{\pi}f(t,x):=f_t(t,x) + b(t)\tran \pi f_x(t,x) + \frac{1}{2}\|\sigma(t)\tran \pi\|^2 f_{xx}(t,x)$ and
  \begin{align}\label{eq:NeccOptObj}
  \varphi^{\hat{\bm{\pi}}}_{t,x,\alpha}(v):=F^{\hat{\bm{\pi}}}_x(t,x,G^{\hat{\bm{\pi}}}(t,x,\alpha)) b(t)\tran v +\frac{1}{2}F^{\hat{\bm{\pi}}}_{xx}(t,x,G^{\hat{\bm{\pi}}}(t,x,\alpha))\|\sigma(t)\tran v\|^2.
\end{align}
  \end{lemma}

Lemma \ref{le:QuantileDifferentiability} provides the derivative $G^{\hat{\bm{\pi}}_{t,\epsilon,\pi}}(t,x,\alpha)$ in $\epsilon=0$. For $t\in[t^*,T)$, it is obvious that $F^{\hat{\bm{\pi}}}(t,x,y)=\mathbf 1_{x\le y}$ and thus is not differentiable in $x$ and $y$. For $t\in [0,t^*)$, we have desired differentiability except at singular points $x\in \mathbb{S}^{\hat{\bm{\pi}}}_t$.

\subsubsection{Sufficiency}\label{subsubse:Sufficiency}

\begin{proposition}\label{prop:Sufficiency}
$\hat{\bm{\pi}}$ as given by \eqref{eq:EquiStrategyPropTarget} is an intra-personal equilibrium for $\alpha=1/2$ and $\hat{\bm{\pi}}$ as given by \eqref{tail risk equilibrium Strategy} is an intra-personal equilibrium for $\alpha\in(0,1/2)$.
\end{proposition}

\subsubsection{Necessity}\label{subsubse:Necessity}

\begin{lemma}\label{le:NeccOpt}
  Consider $\hat{\bm{\pi}}\in \mathbb{A}$ and define $t^*$, $t_*$, $\xi$, and $\mathbb{S}^{\hat{\bm{\pi}}}_t$ as in Lemma \ref{le:QuantileDifferentiability}. Suppose that $\hat{\bm{\pi}}$ is an intra-personal equilibrium  strategy for a given $\alpha \in (0,1)$. Then, for any $t\in[0,t^*)$ and $x\in \bbX^{x_0,\hat{\bm{\pi}}}_t\backslash \mathbb{S}^{\hat{\bm{\pi}}}_t$, we have $F^{\hat{\bm{\pi}}}_{xx}(t,x,G^{\hat{\bm{\pi}}}(t,x,\alpha))>0$, $F^{\hat{\bm{\pi}}}_{x}(t,x,G^{\hat{\bm{\pi}}}(t,x,\alpha))<0$, and
  \begin{align}\label{eq:NeccOptSolution}
    \hat{\bm{\pi}}(t,x) = -\frac{F^{\hat{\bm{\pi}}}_{x}(t,x,G^{\hat{\bm{\pi}}}(t,x,\alpha))}{F^{\hat{\bm{\pi}}}_{xx}(t,x,G^{\hat{\bm{\pi}}}(t,x,\alpha))}v^*(t).
  \end{align}
\end{lemma}

Lemma \ref{le:NeccOpt} proves a necessary condition for $\hat{\bm{\pi}}\in \mathbb{A}$ to be an equilibrium strategy.

\begin{lemma}\label{le:NeccNonDegAlphaLarge}
   Consider $\hat{\bm{\pi}}\in \mathbb{A}$, recall $t^*$, $t_*$, and $\xi$ as defined in Lemma \ref{le:QuantileDifferentiability}, and define
  \begin{align}
    \underline{t}:=\sup\{s\in [0, T]: \theta_0(\tau)+\theta_1(\tau)x_0=0, \forall \tau\in [0, s]  \}.
  \end{align}
  Then, $\underline{t}\ge t^*$ if and only if $\underline{t}=T$. Moreover, if $\hat{\bm{\pi}}$ is an intra-personal equilibrium  strategy for $\alpha \in[1/2,1)$, then $t^*=T$ and $\underline{t} = 0$.
\end{lemma}

The differentiability result in Lemma \ref{le:QuantileDifferentiability} applies to $t\in[0,t^*)$ only. On the other hand, for $t\in[0,\underline{t}]$, $\bbX^{x_0,\hat{\bm{\pi}}}_t=\{x_0\}$ is a singleton, so we cannot obtain too much information about the property of an equilibrium strategy. Thus, we expect to conduct analysis of the equilibrium strategy for $t\in(\underline{t},t^*)$, and Lemma \ref{le:NeccNonDegAlphaLarge} provides some properties of the interval $(\underline{t},t^*)$.

\begin{lemma}\label{le:NeccPDE}
   Consider $\hat{\bm{\pi}}\in \mathbb{A}$, recall $t^*$, $t_*$, $\xi$, and $\mathbb{S}^{\hat{\bm{\pi}}}_t$ as defined in Lemma \ref{le:QuantileDifferentiability} and $\underline{t}$ as defined in Lemma \ref{le:NeccNonDegAlphaLarge}. Suppose $\underline{t}< t^*$ and $\hat{\bm{\pi}}$ is an intra-personal equilibrium  strategy for a given $\alpha \in (0,1)$. Then,
   \begin{enumerate}
     \item[(i)] For each $t\in(\underline{t},T]$, $\bbX^{x_0,\hat{\bm{\pi}}}_t$ is either $(\underline{x}(t),+\infty)$ for some $\underline{x}(t)\in \bbR$, or $(-\infty, \bar{x}(t))$ for some $\bar{x}(t)\in \bbR$, or $\bbR$. Moreover, $\bbX^{x_0,\hat{\bm{\pi}}}_t$ is increasing in $t\in [0,T]$.
     \item[(ii)] There exists $a_0,a_1\in C([\underline{t},t^*])$ taking values in $\bbR$ such that
   \begin{align}
     &-\frac{F^{\hat{\bm{\pi}}}_{x}(t,x,G^{\hat{\bm{\pi}}}(t,x,\alpha))}{F^{\hat{\bm{\pi}}}_{xx}(t,x,G^{\hat{\bm{\pi}}}(t,x,\alpha))} = a_0(t) + a_1(t)x>0,\quad t\in (\underline{t},t^*), x\in \bbX^{x_0,\hat{\bm{\pi}}}_t\backslash \mathbb{S}^{\hat{\bm{\pi}}}_t,\label{eq:NeccStrategy1}\\
     &\hat{\bm{\pi}}(t,x) = \big(a_0(t) + a_1(t)x\big)v^*(t),\quad (t,x)\in (\underline{t},t^*)\times \bbR.\label{eq:NeccStrategy2}
   \end{align}
   Consequently, the following PDE holds:
    \begin{align}\label{quantile equation}
    \left\{
    \begin{array}{l}
G^{\hat{\bm{\pi}}}_t(t,x,\alpha)  + \frac{1}{2}G^{\hat{\bm{\pi}}}_x(t,x,\alpha) \rho(t)\big(a_0(t)+a_1(t)x\big)=0 , \quad t\in (\underline{t},t^*),x\in \bbX^{x_0,\hat{\bm{\pi}}}_t\backslash \mathbb{S}^{\hat{\bm{\pi}}}_t,   \\
\lim_{t \uparrow t^*,x'\rightarrow x}G^{\hat{\bm{\pi}}}_t(t,x',\alpha)=x, \quad x\in \bbR,
\end{array}
\right.
 \end{align}
 where
 \begin{align}\label{eq:rhot}
   \rho(t):=2b(t)\tran v^*(t)-\|\sigma(t)\tran v^*(t)\|^2.
 \end{align}
   \end{enumerate}
\end{lemma}

Lemma \ref{le:NeccPDE} shows that if $\hat{\bm{\pi}}\in \mathbb{A}$ is an intra-personal equilibrium, then it must take the form \eqref{eq:NeccStrategy2} in the time interval $(\underline{t},t^*)$ and the quantile of the terminal wealth under $\hat{\bm{\pi}}\in \mathbb{A}$ satisfies the PDE \eqref{quantile equation}. Note that this PDE is a linear transportation equation defined on a possibly strict subset of $(\underline{t},t^*)\times \bbR$. In order to identify the equilibrium strategy $\hat{\bm{\pi}}$, we need to solve $a_0$ and $a_1$, so it is crucial to solve the transportation equation.

\begin{lemma}\label{le:TransportEq}
  Fix $\tau_1<\tau_2$, $\bar c\in \bbR$, $c,\gamma_0,\gamma_1\in C([\tau_1,\tau_2])$ taking values in $\bbR$, and $\alpha_0,\alpha_1\in \bbR$. Define
  \begin{align}\label{eq:TransportEqSolution}
    \hat g(t,x):=\alpha_0+\alpha_1\left[\int_t^{\tau_2} \gamma_0(s)e^{\int_s^{\tau_2} \gamma_1(\tau)d \tau} d s+ x e^{\int_t^{\tau_2} \gamma_1(s)d s}\right],\quad t\in[\tau_1,\tau_2]\times \bbR.
  \end{align}
  Denote ${\cal X}_{1,t}:=(\bar c,+\infty)$, ${\cal X}_{2,t}:=(-\infty,\bar c)$, and ${\cal X}_{3,t}:=(c(t),+\infty)$, $t\in[\tau_1,\tau_2]$. For each $i=1,2,3$, define
  \begin{align*}
    {\cal D}_i=\{(t,x)\mid x\in  {\cal X}_{i,t},t\in [\tau_1,\tau_2)\},\quad \bar {\cal D}_i=\{(t,x)\mid x\in  {\cal X}_{i,t},t\in [\tau_1,\tau_2]\},
  \end{align*}
 denote by $C^{1,1}(\bar {\cal D}_i)$ the set of real-valued functions that are continuous on $\bar {\cal D}_i$ and differentiable on ${\cal D}_i$, and consider
 \begin{align}\label{eq:TransportEq}
  \left\{
  \begin{array}{l}
    g_t(t,x) + g_x(t,x)(\gamma_0(t)+\gamma_1(t)x) = 0,\quad (t,x)\in {\cal D}_i,\\
    g(\tau_2,x) = \alpha_0 + \alpha_1 x,\quad x\in {\cal X}_{i,\tau_2}.
  \end{array}
  \right.
\end{align}
\begin{enumerate}
  \item[(i)] If $g\in C^{1,1}(\bar {\cal D}_1)$ is the solution to \eqref{eq:TransportEq} with $i=1$, then there exists $\underline{x}\ge \bar c$ such that $g(t,x) = \hat g(t,x)$, $(t,x)\in [\tau_1,\tau_2]\times (\underline{x},+\infty)$.
\item[(ii)] If $g\in C^{1,1}(\bar {\cal D}_1)$ is the solution to \eqref{eq:TransportEq} with $i=2$,
then there exists $\bar{x}\le \bar c$ such that $g(t,x) = \hat g(t,x)$, $(t,x)\in [\tau_1,\tau_2]\times (-\infty,\bar{x})$.
\item[(iii)] If $c(t)$ is decreasing in $t\in[\tau_1,\tau_2]$, $\gamma_0(t)+\gamma_1(t)x>0$ for all $(t,x)\in {\cal D}_3$, and $g\in C^{1,1}(\bar {\cal D}_1)$ is the solution to \eqref{eq:TransportEq} with $i=3$, then $g(t,x) = \hat g(t,x)$, $(t,x)\in {\cal D}_3$.
\end{enumerate}
\end{lemma}

Lemma \ref{le:TransportEq} solves the linear transportation equation with affine terminal condition. The solution is also an affine function of $x$ for each $t$.

\begin{lemma}\label{le:EquilibriumQuantile}
   Consider $\hat{\bm{\pi}}\in \mathbb{A}$, recall $t^*$, $t_*$, $\xi$, and $\mathbb{S}^{\hat{\bm{\pi}}}_t$ as defined in Lemma \ref{le:QuantileDifferentiability} and $\underline{t}$ as defined in Lemma \ref{le:NeccNonDegAlphaLarge}. Suppose $\underline{t}< t^*$ and $\hat{\bm{\pi}}$ is an intra-personal equilibrium  for a given $\alpha \in (0,1)$, and recall $a_0$ and $a_1$ as defined in Lemma \ref{le:NeccPDE}. Define
   \begin{align}\label{eq:EquiQuantCoeff}
     \beta_1(t,\alpha) = e^{\frac{1}{2}\int_t^{t^*}a_1(s)\rho(s)ds},\quad \beta_0(t,\alpha):=\frac{1}{2}\int_t^{t^*}a_0(s)\rho(s)\beta_1(s,\alpha)ds,\quad t\in [\underline{t},t^*].
   \end{align}
   Fix any $t\in (\underline{t},t^*)$.
   \begin{enumerate}
   \item[(i)] There exists $s\in[t,t^*)$ such that $a_1(s)\neq 0$. Consequently,
   \begin{align}\label{eq:a1Eq0}
     \int_t^{t^*}a_1(s)^2\|\sigma(s)\tran v^*(s)\|^2ds>0.
   \end{align}
     \item[(ii)] If $\bbX^{x_0,\hat{\bm{\pi}}}_t \supseteq (\underline{x}_t,+\infty)$ for certain $\underline{x}_t\in \bbR$, then there exists $\underline{c}_t> \max(\underline{x}_t,\xi)$ such that
         \begin{align}\label{eq:EquilibriumQuantile}
           G^{\hat{\bm{\pi}}}(s,x,\alpha) = \beta_1(s,\alpha)x + \beta_0(s,\alpha)
         \end{align}
         for all $(s,x)\in [t,t^*]\times [\underline{c}_t,+\infty)$. Moreover,
         \begin{align}\label{eq:a1Eq1}
         &\frac{1}{2}\int_t^{t^*} a_1(s)\big(1-a_1(s)\big)\|\sigma(s)\tran v^*(s)\|^2ds+ \Phi^{-1}(\alpha) \sqrt{\int_t^{t^*} a_1(s)^2 \|\sigma(s)\tran v^*(s)\|^2 ds}=0.
         \end{align}
          \item[(iii)] If $\bbX^{x_0,\hat{\bm{\pi}}}_t \supseteq (-\infty,\bar{x}_t)$ for certain $\bar{x}_t\in \bbR$, then there exists $\bar c_t<\min(\xi, \bar{x}_t)$ such that \eqref{eq:EquilibriumQuantile} holds all $(s,x)\in [t,t^*]\times (-\infty,\bar{c}_t)$. Moreover,
         \begin{align}\label{eq:a1Eq2}
         &\frac{1}{2}\int_t^{t^*} a_1(s)\big(1-a_1(s)\big)\|\sigma(s)\tran v^*(s)\|^2ds+ \Phi^{-1}(1-\alpha) \sqrt{\int_t^{t^*} a_1(s)^2 \|\sigma(s)\tran v^*(s)\|^2 ds}=0.
         \end{align}
   \end{enumerate}
\end{lemma}

Lemma \ref{le:NeccPDE} already proved that an equilibrium strategy $\hat{\bm{\pi}}\in \mathbb{A}$ must take the form \eqref{eq:NeccStrategy2} with $a_0$ and $a_1$ undetermined. Lemma \ref{le:EquilibriumQuantile} provides a necessary condition for $a_1$, which is then used in the following Proposition \ref{prop:NeccAlphaNotHalf} to characterize equilibrium strategies for $\alpha\neq 1/2$ and in the following Proposition \ref{prop:NeccAlphaHalf} to characterize equilibrium strategies for $\alpha=1/2$.

\begin{proposition}\label{prop:NeccAlphaNotHalf}
  For $\alpha\in(0,1/2)$, $\hat{\bm{\pi}}\in\mathbb{A}$ is an intra-personal equilibrium  if and only if $\hat{\bm{\pi}}$ is given by \eqref{tail risk equilibrium Strategy} for some $\theta\in C_{\mathrm{pw}}([0,T))$ taking values in $\bbR^m$. For $\alpha\in (1/2,1)$, any $\hat{\bm{\pi}}\in\mathbb{A}$ is not an intra-personal equilibrium.
\end{proposition}

\begin{proposition}\label{prop:NeccAlphaHalf}
  For $\alpha=1/2$, $\hat{\bm{\pi}}\in\mathbb{A}$ is an intra-personal equilibrium  if and only if it is given by \eqref{eq:EquiStrategyPropTarget} for some $\xi<x_0$.
\end{proposition}

\subsubsection{Detailed Proofs}\label{subsubse:proofs}
  \begin{pfof}{Lemma \ref{le:QuantileDifferentiability}}
  Suppose $t_*<t^*$ and fix any $t\in[t_*,t^*)$. By  Corollary 3 and Theorem 2-(i) in \citet{HeJiang2020:LinearSDE}, there exists $\eta\in (0,t^*-t)$ such that $F^{\hat{\bm{\pi}}}(s,x,y)\in C^{1,\infty}([t,t+\eta]\times \bbR^2\backslash(\xi,\xi))$, $F_t^{\hat{\bm{\pi}}}(s,x,y)\in C^{0,\infty}([t,t+\eta]\times \bbR^2\backslash(\xi,\xi))$, $F^{\hat{\bm{\pi}}}_t$ is bounded on $[t,t+\eta]\times \bbR^2\backslash(\xi,\xi)$, and the derivatives of $F^{\hat{\bm{\pi}}}$ and $F^{\hat{\bm{\pi}}}_t$ with respect to $x$ and $y$ of any order are bounded on $[t,t+\eta]\times \bbR^2\backslash B^2_\delta(\xi)$ for any $\delta>0$.

 Fix $\pi\in \bbR^m$. For any fixed $x\in \bbR$, $y\neq \xi$, noting that $F^{\hat{\bm{\pi}}_{t,\epsilon,\pi}}(t,x,y) = \expect[F^{\hat{\bm{\pi}}}(t+\epsilon,X_{t,x}^\pi(t+\epsilon),y)]$ and recalling that for any $\delta>0$, $F^{\hat{\bm{\pi}}}_t$ and $F^{\hat{\bm{\pi}}}_x$ are bounded on $[t,t+\eta]\times \bbR^2\backslash B^2_\delta(\xi)$, we derive by the dominated convergence theorem that for any $\epsilon_0\in [0,\eta)$,
  \begin{align*}
    \lim_{\epsilon\rightarrow \epsilon_0}\sup_{x\in \bbR,y\in \bbR\backslash B_\delta(\xi)} \left|F^{\hat{\bm{\pi}}_{t,\epsilon,\pi}}(t,x,y)-F^{\hat{\bm{\pi}}_{t,\epsilon_0,\pi}}(t,x,y)\right| = 0.
  \end{align*}
  The above, together with the continuity of $F^{\hat{\bm{\pi}}_{t,\epsilon,\pi}}(t,x,y)$ in $y\neq \xi$ as implied by  Corollary 3 and Theorem 2-(iii) in \citet{HeJiang2020:LinearSDE}, yields that $F^{\hat{\bm{\pi}}_{t,\epsilon,\pi}}(t,x,y)$ is continuous in $(\epsilon,y)$ with $\epsilon\in [0,\eta)$ and $y\neq \xi$. Now, for any $x\neq \xi$,  Corollary 3 and Corollary 2-(ii) in \citet{HeJiang2020:LinearSDE}  show that for any $\alpha\in (0,1)$, $F^{\hat{\bm{\pi}}_{t,\epsilon,\pi}}(t,x,G^{\hat{\bm{\pi}}_{t,\epsilon,\pi}}(t,x,\alpha))=\alpha$, $G^{\hat{\bm{\pi}}_{t,\epsilon,\pi}}(t,x,\alpha)\neq \xi$, and $F^{\hat{\bm{\pi}}_{t,\epsilon,\pi}}_y(t,x,G^{\hat{\bm{\pi}}_{t,\epsilon,\pi}}(t,x,\alpha))>0$. By  Corollary 3 and Theorem 2-(iii) in \citet{HeJiang2020:LinearSDE}, $F^{\hat{\bm{\pi}}_{t,\epsilon,\pi}}_y(t,x,y)$ is continuous in $y\neq \xi$. The implicit function theorem then yields that $G^{\hat{\bm{\pi}}_{t,\epsilon,\pi}}(t,x,\alpha)$ is continuous in $\epsilon \in [0,\eta)$.

  Fixing $x\in \bbR$, because $G^{\hat{\bm{\pi}}_{t,0,\pi}}(t,x,\alpha)=G^{\hat{\bm{\pi}}}(t,x,\alpha)\neq \xi$ and because $G^{\hat{\bm{\pi}}_{t,\epsilon,\pi}}(t,x,\alpha)$ is continuous in $\epsilon$, there exists $\delta>0$ such that $|G^{\hat{\bm{\pi}}_{t,\epsilon,\pi}}(t,x,\alpha)-\xi|>\delta$ for sufficiently small $\epsilon$. Because the derivatives of $F^{\hat{\bm{\pi}}}$ and $F^{\hat{\bm{\pi}}}_t$ with respect to $x$ and $y$ of any orders are bounded on $[t,t+\eta)\times \bbR^2\backslash B^2_\delta(\xi)$, there exists $L>0$ such that
  \begin{align}
    &\sup_{s\in[t,\eta)} \left|{\cal A}^\pi F^{\hat{\bm{\pi}}}(s,x,y_1)-{\cal A}^\pi F^{\hat{\bm{\pi}}}(s,x,y_2)\right|\notag\\
    &\le L(1+|x| + |x|^2)|y_1-y_2|,\quad
     x\in \bbR, (y_1,y_2)\in [\xi+\delta,+\infty)^2\cup (-\infty,\xi-\delta]^2.\label{eq:GeneratorLip}
  \end{align}
  Because $\expect\left[\sup_{s\in[t,\eta)}|X^{\pi}_{t,x}(s)|^p\right]<+\infty$ for any $p\ge 1$, we conclude that from \eqref{eq:GeneratorLip} that for sufficiently small $\epsilon>0$,
  \begin{align}
    &\expect\left[\sup_{s\in[t,\eta)}\left|{\cal A}^\pi F^{\hat{\bm{\pi}}}(s,X^{\pi}_{t,x}(s),G^{\hat{\bm{\pi}}_{t,\epsilon,\pi}}(t,x,\alpha))-{\cal A}^\pi F^{\hat{\bm{\pi}}}(s,X^{\pi}_{t,x}(s),G^{\hat{\bm{\pi}}}(t,x,\alpha))\right|\right]\notag\\
    &\le L'|G^{\hat{\bm{\pi}}_{t,\epsilon,\pi}}(t,x,\alpha)-G^{\hat{\bm{\pi}}}(t,x,\alpha)|\label{eq:LeQuanDiffEst1}
  \end{align}
  for some constant $L'>0$. On the other hand, we have
  \begin{align*}
  &F^{\hat{\bm{\pi}}_{t,\epsilon,\pi}}\left(t,x,G^{\hat{\bm{\pi}}_{t,\epsilon,\pi}}(t,x,\alpha)\right) - F^{\hat{\bm{\pi}}}\left(t,x,G^{\hat{\bm{\pi}}_{t,\epsilon,\pi}}(t,x,\alpha)\right)\\
   &= \expect[F^{\hat{\bm{\pi}}}(t+\epsilon,X^{\pi}_{t,x}(t+\epsilon),G^{\hat{\bm{\pi}}_{t,\epsilon,\pi}}(t,x,\alpha))] - F^{\hat{\bm{\pi}}}\left(t,x,G^{\hat{\bm{\pi}}_{t,\epsilon,\pi}}(t,x,\alpha)\right)\\
  &={\cal A}^\pi F^{\hat{\bm{\pi}}}(t,x,G^{\hat{\bm{\pi}}}(t,x,\alpha))\epsilon  +  \expect\Big[\int_t^{t+\epsilon}\Big({\cal A}^\pi F^{\hat{\bm{\pi}}}(s,X^{\pi}_{t,x}(s),G^{\hat{\bm{\pi}}_{t,\epsilon,\pi}}(t,x,\alpha))\\
  &\quad -{\cal A}^\pi F^{\hat{\bm{\pi}}}(s,X^{\pi}_{t,x}(s),G^{\hat{\bm{\pi}}}(t,x,\alpha))\Big)ds\Big]\\
  &\quad +  \expect\left[\int_t^{t+\epsilon}\Big({\cal A}^\pi F^{\hat{\bm{\pi}}}(s,X^{\pi}_{t,x}(s),G^{\hat{\bm{\pi}}}(t,x,\alpha))-{\cal A}^\pi F^{\hat{\bm{\pi}}}(t,x,G^{\hat{\bm{\pi}}}(t,x,\alpha))\Big)ds\right].
\end{align*}
Combining the above with \eqref{eq:LeQuanDiffEst1}, recalling that $F_t^{\hat{\bm{\pi}}}$, $F_{x}^{\hat{\bm{\pi}}}$, and $F_{xx}^{\hat{\bm{\pi}}}$ are bounded on $[t,t+\eta)\times \bbR^2\backslash B^2_\delta(\xi)$, noting that $\expect\left[\sup_{s\in[t,\eta)}|X^{\pi}_{t,x}(s)|^p\right]<+\infty$ for any $p\ge 1$, and applying the dominated convergence theorem, we conclude that
\begin{align}
  \lim_{\epsilon\downarrow 0}\frac{F^{\hat{\bm{\pi}}_{t,\epsilon,\pi}}\left(t,x,G^{\hat{\bm{\pi}}_{t,\epsilon,\pi}}(t,x,\alpha)\right) - F^{\hat{\bm{\pi}}}\left(t,x,G^{\hat{\bm{\pi}}_{t,\epsilon,\pi}}(t,x,\alpha)\right)}{\epsilon} = {\cal A}^\pi F^{\hat{\bm{\pi}}}(t,x,G^{\hat{\bm{\pi}}}(t,x,\alpha)).\label{eq:LeQuanDiffEq1}
\end{align}

Because $F^{\hat{\bm{\pi}}_{t,\epsilon,\pi}}(t,x,G^{\hat{\bm{\pi}}_{t,\epsilon,\pi}}(t,x,\alpha))=\alpha$ for any $\epsilon\in[0,\eta)$, we conclude that
\begin{align*}
 0 &= \frac{F^{\hat{\bm{\pi}}_{t,\epsilon,\pi}}(t,x,G^{\hat{\bm{\pi}}_{t,\epsilon,\pi}}(t,x,\alpha)) - F^{\hat{\bm{\pi}}}(t,x,G^{\hat{\bm{\pi}}_{t,\epsilon,\pi}}(t,x,\alpha))}{\epsilon}\\
  & \quad + \frac{F^{\hat{\bm{\pi}}}(t,x,G^{\hat{\bm{\pi}}_{t,\epsilon,\pi}}(t,x,\alpha)) - F^{\hat{\bm{\pi}}}(t,x,G^{\hat{\bm{\pi}}}(t,x,\alpha))}{\epsilon}.
\end{align*}
Because $F^{\hat{\bm{\pi}}}_y$ is continuous on $[t,t+\eta)\times \bbR^2\backslash(\xi,\xi)$, $G^{\hat{\bm{\pi}}}(t,x,\alpha)\neq \xi$, and $F^{\hat{\bm{\pi}}}_y(t,x,y)>0$ for $y$ with $F^{\hat{\bm{\pi}}}(t,x,y)\in (0,1)$, the above and \eqref{eq:LeQuanDiffEq1} immediately lead to the first two equality in \eqref{eq:QuanDerivativeWRTeps}. In addition,  Corollary 3 and Theorem 2-(i) in \citet{HeJiang2020:LinearSDE} shows that ${\cal A}^{\hat{\bm{\pi}}}F^{\hat{\bm{\pi}}}(s,x,y)=0$ on $[t,t+\eta)\times \bbR^2\backslash(\xi,\xi)$. The last equality in \eqref{eq:QuanDerivativeWRTeps} then follows.

Suppose $t_*>0$ and fix $t\in [0,t_*)$. By  Corollary 3 and Theorem 2-(ii) in \citet{HeJiang2020:LinearSDE}, there exists $\eta\in (0,t^*-t)$ such that $F^{\hat{\bm{\pi}}}(s,x,y)\in C^{1,\infty}([t,t+\eta)\times \bbR^2 )$, $F_t^{\hat{\bm{\pi}}}(s,x,y)\in C^{0,\infty}([t,t+\eta)\times \bbR^2)$, the derivatives of $F^{\hat{\bm{\pi}}}$ with respect to $x$ and $y$ of any order are bounded on $[t,t+\eta)\times \bbR^2$, and $\sup_{s\in[t,t+\eta),y\in \bbR}|\frac{\partial^{i+j}F^{\hat{\bm{\pi}}}_t}{\partial x^i\partial y^j}(s,x,y)|$ is of polynomial growth in $x$ for any $i,j\in \mathbb{N}_0$. The remaining proof then follows the same line as for the case $t\in [t_*,t^*)$.\halmos
  \end{pfof}

  \begin{pfof}{Proposition \ref{prop:Sufficiency}}
We first consider the case $\alpha=1/2$. It is straightforward to see that
\begin{align*}
  d(X_{0,x_0}^{\hat{\bm{\pi}}}(s)-\xi) = (X_{0,x_0}^{\hat{\bm{\pi}}}(s)-\xi)\left[v^*(s)\tran b(s)ds + v^*(s)\tran \sigma(s)dW(s)\right],\quad s\in[0,T)
\end{align*}
and $X_{0,x_0}^{\hat{\bm{\pi}}}(0)-\xi=x_0-\xi>0$. Because $v^*(s)\neq 0$ and $\sigma(s)\sigma(s)\tran$ is positive definition, $s\in[0,T)$, we conclude $\bbX^{x_0,\hat{\bm{\pi}}}_0 = \{x_0\}$ and $\bbX^{x_0,\hat{\bm{\pi}}}_t = (\xi,+\infty)$, $t\in(0,T)$. For every $t\in[0,T)$, by the definition of $v^*(t)$, $Qv^*(t)\ge 0$. Because $X_{0,x_0}^{\hat{\bm{\pi}}}(t)>\xi$, we immediately conclude $Q\hat{\bm{\pi}}(t,X_{0,x_0}^{\hat{\bm{\pi}}}(t))\ge 0$. In addition, $v^*(t)$ is bounded in $t\in[0,T)$. Thus, $\hat{\bm{\pi}}\in \Pi$.

Now, fix any $t\in [0,T)$ and $x\in \bbX^{x_0,\hat{\bm{\pi}}}_t$. Straightforward calculation shows that
\begin{align*}
  F^{\hat{\bm{\pi}}}_x(t,x,G^{\hat{\bm{\pi}}}(t,x,1/2)) = -\frac{\phi(0)}{\sqrt{\int_t^T\|\sigma(s)\tran v^*(s)\|^2ds}(x-\xi)}<0,\\
  F^{\hat{\bm{\pi}}}_{xx}(t,x,G^{\hat{\bm{\pi}}}(t,x,1/2)) = -(x-\xi)^{-1}F^{\hat{\bm{\pi}}}_x(t,x,G^{\hat{\bm{\pi}}}(t,x,1/2))>0.
\end{align*}
Then, because $v^*(t)$ is the optimal solution to \eqref{general small alpha 1}, $\hat{\bm{\pi}}(t,x) = v^*(t)(x-\xi)$ is the unique optimal solution of
\begin{align}\label{eq:AKeyQuadraticProg}
\left\{
  \begin{array}{cl}
    \underset{\pi \in \bbR^m}{\min} & \varphi^{\hat{\bm{\pi}}}_{t,x,\alpha}(\pi)\\
    \text{subject to} & Q\pi\ge 0,
  \end{array}
\right.
\end{align}
with $\alpha =1/2$, where $\varphi^{\hat{\bm{\pi}}}_{t,x,\alpha}$ is defined by \eqref{eq:NeccOptObj}. As a result, for any $\pi\in  \bbR^m$ with $Q\pi\ge 0$, $\pi\neq \hat{\bm{\pi}}(t,x)$, Lemma \ref{le:QuantileDifferentiability} yields that
\begin{align*}
  \lim_{\epsilon\downarrow 0}\frac{G^{\hat{\bm{\pi}}_{t,\epsilon,\pi}}(t,x,1/2) - G^{\hat{\bm{\pi}}}(t,x,1/2)}{\epsilon}<0.
\end{align*}
Thus, $\hat{\bm{\pi}}$ is an equilibrium strategy.

Next, we consider the case $\alpha<1/2$. It is straightforward to see that $X_{0,x_0}^{\hat{\bm{\pi}}}(t)\equiv x_0$ and thus $\bbX^{x_0,\hat{\bm{\pi}}}_t = \{x_0\},t\in[0,T]$. As a result, $\hat{\bm{\pi}}(t,X_{0,x_0}^{\hat{\bm{\pi}}}(t))=0,t\in[0,T)$. In addition, $\theta(t)$ is bounded in $t\in[0,T)$. Thus, $\hat{\bm{\pi}}\in \Pi$.

Fix any $t\in[0,T)$ and $\pi \in \bbR^m$ with $\pi \neq \hat{\bm{\pi}}(t,x_0)=0$. For any $\epsilon\in(0,T-t)$, straightforward calculation leads to
 \begin{align}\label{process for non invest}
&X^{\hat{\bm{\pi}}_{t,\epsilon,\pi}}_{t,x_0}(T)-x_0 \nonumber \\
=&\left[  \int_{t}^{t+\epsilon} b(s)\tran \pi ds  +\int_{t}^{t+\epsilon}  \pi \tran \sigma(s) dW(s) \right] e^{  \int_{t+\epsilon}^T  [ b(\tau)\tran \theta(\tau) -\frac{1}{2} \|\sigma(\tau) \tran \theta(\tau) \|^2 ]d \tau +\int_{t+\epsilon}^T \theta(\tau) \tran \sigma(\tau) dW(\tau) }.
\end{align}
As a result, 
\begin{align}
 & F^{\hat{\bm{\pi}}_{t,\epsilon,\pi}}(t,x_0,x_0) = \prob(X^{\hat{\bm{\pi}}_{t,\epsilon,\pi}}_{t,x_0}(T) \leq x_0 )=\prob\left(\int_{t}^{t+\epsilon} b(s)\tran \pi ds+\int_{t}^{t+\epsilon}  \pi \tran \sigma(s) dW(s) \leq 0\right)\nonumber\\
&=\Phi\left(\frac{-\int_{t}^{t+\epsilon} b(s)\tran \pi ds}{ \sqrt{ \int_{t}^{t+\epsilon} \|\sigma(s)\tran \pi\|^2 ds }   } \right).\label{prob for non invest}
\end{align}
Because $\alpha<1/2$, the right-hand side of the above is strictly larger than $\alpha$ when $\epsilon$ is sufficiently small. As a result, $G^{\hat{\bm{\pi}}_{t,\epsilon,\pi}}(t,x_0,\alpha)\le x_0$ for sufficiently small $\epsilon$. Thus, $\hat{\bm{\pi}}$ is an equilibrium strategy.\halmos
\end{pfof}

\begin{pfof}{Lemma \ref{le:NeccOpt}}
  For any $t\in[0,t_*)$ and $x\in \bbX^{x_0,\hat{\bm{\pi}}}_t$ and for any $t\in[t_*,t^*)$ and $x\in \bbX^{x_0,\hat{\bm{\pi}}}_t$ with $x\neq \xi$, Lemma \ref{le:QuantileDifferentiability} implies that
\begin{align}
  \varphi^{\hat{\bm{\pi}}}_{t,x,\alpha}(\hat{\bm{\pi}}(t,x))\le \varphi^{\hat{\bm{\pi}}}_{t,x,\alpha}(\pi),\quad \forall \pi\neq \hat{\bm{\pi}}(t,x)\text{ with }Q\pi \ge 0.\label{eq:LeNeccOptEq1}
\end{align}
 According to Corollary 3, Theorem 2-(iv) and Corollary 2-(ii) in \citet{HeJiang2020:LinearSDE}, we have $F^{\hat{\bm{\pi}}}_x(t,x,G^{\hat{\bm{\pi}}}(t,x,\alpha))<0$. On the other hand, because of Assumption \ref{as:Feasibility2}, we can find $v_0\in \bbR^m$ with $Qv_0\ge 0$ such that $b(t)\tran v_0>0$. Then $Q(\lambda v_0)\ge 0$ for any $\lambda>0$. If $F^{\hat{\bm{\pi}}}_{xx}(t,x,G^{\hat{\bm{\pi}}}(t,x,\alpha))\le 0$, we consider $\lambda v_0$ for sufficiently large, positive scalar $\lambda$ so that $\lambda v_0 \neq \hat{\bm{\pi}}(t,x)$ and $\varphi^{\hat{\bm{\pi}}}_{t,x,\alpha}(\hat{\bm{\pi}}(t,x))>\varphi^{\hat{\bm{\pi}}}_{t,x,\alpha}(\lambda v_0)$, which contradicts \eqref{eq:LeNeccOptEq1}. Thus, we must have $F^{\hat{\bm{\pi}}}_{xx}(t,x,G^{\hat{\bm{\pi}}}(t,x,\alpha))>0$. Then, \eqref{eq:LeNeccOptEq1} immediately implies that
\begin{align*}
-\hat{\bm{\pi}}(t,x)F^{\hat{\bm{\pi}}}_{xx}(t,x,G^{\hat{\bm{\pi}}}(t,x,\alpha))/F^{\hat{\bm{\pi}}}_{x}(t,x,G^{\hat{\bm{\pi}}}(t,x,\alpha))
 \end{align*}
 is the optimizer of \eqref{general small alpha 1}, i.e., \eqref{eq:NeccOptSolution} holds.\halmos
\end{pfof}

\begin{pfof}{Lemma \ref{le:NeccNonDegAlphaLarge}}
If $\underline{t}=T$, then it is obvious that $\underline{t}\ge t^*$. On the other hand, by the definition of $t^*$, $\theta_0(\tau)+\theta_1(\tau)x_0=0,\tau\in[t^*,T)$. Thus, $\underline{t}\ge t^*$ implies $\underline{t}=T$.

Next, we fix $\alpha \in[1/2,1)$. For the sake of contradiction, suppose $t^*<T$. Then, for any $t\in [t^*,T)$ and $x\in \bbX^{x_0,\hat{\bm{\pi}}}_t$, we have $\hat{\bm{\pi}}(t,x)=0$. Choose any $\pi\in \bbR^m$ with $b(t)\tran \pi>0$ and $Q\pi\ge 0$. Note that such $\pi$ exists due to Assumption \ref{as:Feasibility2}. Then, we must have $\pi \neq 0=\hat{\bm{\pi}}(t,x)$. For any $\epsilon\in(0,T-t)$, straightforward calculation leads to
 \begin{align*}
X^{\hat{\bm{\pi}}_{t,\epsilon,\pi}}_{t,x}(T) =x + \int_{t}^{t+\epsilon} b(s)\tran \pi ds  +\int_{t}^{t+\epsilon}  \pi \tran \sigma(s) dW(s) .
\end{align*}
Because $\alpha\ge 1/2$, $b(t)\tran \pi>0$, and $b$ is right-continuous, straightforward calculation shows that for sufficiently small $\epsilon>0$, $G^{\hat{\bm{\pi}}_{t,\epsilon,\pi}}(t,x,\alpha)>x = G^{\hat{\bm{\pi}}}(t,x,\alpha)$. This contradicts the assumption that $\hat{\bm{\pi}}$ is an equilibrium strategy. Thus, we must have $t^*=T$.

Next, for the sake of contradiction, suppose $\underline{t}>0$. Then, $X_{0,x_0}^{\hat{\bm{\pi}}}(s)=x_0$ and thus $\bbX^{x_0,\hat{\bm{\pi}}}_s=\{x_0\}$ for all $s\in[0,\underline{t})$. Recall $t_*$ as defined in \eqref{eq:T0prime} and $\theta_0(s) +  \theta_1(s)\xi=0,s\in [t_*,t^*)=[t_*,T)$ for certain uniquely determined $\xi \in \bbR$. When $t_*>0$ we can choose any $t\in[0,t_*\wedge \underline{t})$ and when $\xi \neq x_0$ we can choose any $t\in[0,\underline{t})$. In either case, Lemma \ref{le:NeccOpt} can apply to this particular $t$ together with $x_0 \in \bbX^{x_0,\hat{\bm{\pi}}}_t$, leading to $F^{\hat{\bm{\pi}}}_{x}(t,x_0,G^{\hat{\bm{\pi}}}(t,x_0,\alpha))<0$, $F^{\hat{\bm{\pi}}}_{xx}(t,x_0,G^{\hat{\bm{\pi}}}(t,x_0,\alpha))>0$, and
  \begin{align*}
    \hat{\bm{\pi}}(t,x_0) = -\frac{F^{\hat{\bm{\pi}}}_{x}(t,x_0,G^{\hat{\bm{\pi}}}(t,x_0,\alpha))}{F^{\hat{\bm{\pi}}}_{xx}(t,x_0,G^{\hat{\bm{\pi}}}(t,x_0,\alpha))}v^*(t).
  \end{align*}
The above is a contradiction because $\hat{\bm{\pi}}(t,x_0)=0$ and $v^*(t)\neq 0$.

When $t_*=0$ and $\xi =x_0$, we have $\underline{t}=T$ and thus $\bbX^{x_0,\hat{\bm{\pi}}}_s=\{x_0\}$ and $X_{s,x_0}^{\hat{\bm{\pi}}}(T) = X_{0,x_0}^{\hat{\bm{\pi}}}(T)$ for all $s\in[0,\underline{t})$. Fix any $t\in[0,T)$ and choose any $\pi\in \bbR^m$ with $b(t)\tran \pi>0$ and $Q\pi\ge 0$, and such $\pi$ can be found due to Assumption \ref{as:Feasibility2}. For any $\epsilon\in(0,T-t)$, \eqref{prob for non invest} holds. Thus, because $\alpha\ge 1/2$, for sufficiently small $\epsilon>0$, $F^{\hat{\bm{\pi}}_{t,\epsilon,\pi}}(t,x_0,x_0)<\alpha$. In addition, because $\pi\neq 0$, $F^{\hat{\bm{\pi}}_{t,\epsilon,\pi}}(t,x_0,y)$ is strictly increasing and continuous in $y$ for any $\epsilon \in (0,T-t)$. Consequently, we conclude that $G^{\hat{\bm{\pi}}_{t,\epsilon,\pi}}(t,x_0,\alpha)>x_0=G^{\hat{\bm{\pi}}}(t,x_0,\alpha)$ for sufficiently small $\epsilon>0$. This contradicts the assumption that $\hat{\bm{\pi}}$ is an equilibrium policy. The proof then completes.\halmos
\end{pfof}

\begin{pfof}{Lemma \ref{le:NeccPDE}}
  Part (i) is an immediate consequence of   Corollary 4 and Theorem 3 in \citet{HeJiang2020:LinearSDE}. We prove part (ii) in the following.

  Fix any $t\in (\underline{t},t^*)$, Lemma \ref{le:NeccOpt} implies that
  \begin{align}\label{eq:LeNeccPDEEq1}
     -\frac{F^{\hat{\bm{\pi}}}_{x}(t,x,G^{\hat{\bm{\pi}}}(t,x,\alpha))}{F^{\hat{\bm{\pi}}}_{xx}(t,x,G^{\hat{\bm{\pi}}}(t,x,\alpha))}v^*(t) = \hat{\bm{\pi}}(t,x) = \theta_0(t)+\theta_1(t)x,\quad  x\in \bbX^{x_0,\hat{\bm{\pi}}}_t\backslash \mathbb{S}^{\hat{\bm{\pi}}}_t.
  \end{align}
  Multiplying $v^*(t)\tran$ on both sides of the above equality and noting that $v^*(t)\neq 0$, we conclude
  \begin{align*}
    -\frac{F^{\hat{\bm{\pi}}}_{x}(t,x,G^{\hat{\bm{\pi}}}(t,x,\alpha))}{F^{\hat{\bm{\pi}}}_{xx}(t,x,G^{\hat{\bm{\pi}}}(t,x,\alpha))} = \|v^*(t)\|^{-2}\left(v^*(t)\tran\theta_0(t)+v^*(t)\tran\theta_1(t)x\right),\quad x\in \bbX^{x_0,\hat{\bm{\pi}}}_t\backslash \mathbb{S}^{\hat{\bm{\pi}}}_t.
  \end{align*}
  Then,
  \begin{align*}
  \left[\left(-\frac{F^{\hat{\bm{\pi}}}_{x}(t,x',G^{\hat{\bm{\pi}}}(t,x',\alpha))}{F^{\hat{\bm{\pi}}}_{xx}(t,x',G^{\hat{\bm{\pi}}}(t,x',\alpha))}\right) - \left(-\frac{F^{\hat{\bm{\pi}}}_{x}(t,x'',G^{\hat{\bm{\pi}}}(t,x'',\alpha))}{F^{\hat{\bm{\pi}}}_{xx}(t,x'',G^{\hat{\bm{\pi}}}(t,x'',\alpha))}\right)\right]/(x'-x'')
  \end{align*}
  does not depend on the choice of $x',x''\in \bbX^{x_0,\hat{\bm{\pi}}}_t\backslash \mathbb{S}^{\hat{\bm{\pi}}}_t$ with $x'\neq x''$, and we denote this common value by $a_1(t)$. Because $ \bbX^{x_0,\hat{\bm{\pi}}}_t$ is a nonempty open interval and $\mathbb{S}^{\hat{\bm{\pi}}}_t$ is either the empty set or a singleton, we can always find $x',x''\in \mathbb{S}^{\hat{\bm{\pi}}}_t$ with $x'\neq x''$ and thus $a_1(t)$ is well defined. Then, fixing any $\bar x \in  \bbX^{x_0,\hat{\bm{\pi}}}_t\backslash \mathbb{S}^{\hat{\bm{\pi}}}_t$, we have
  \begin{align*}
 -\frac{F^{\hat{\bm{\pi}}}_{x}(t,x,G^{\hat{\bm{\pi}}}(t,x,\alpha))}{F^{\hat{\bm{\pi}}}_{xx}(t,x,G^{\hat{\bm{\pi}}}(t,x,\alpha))} - a_1(t)x  = -\frac{F^{\hat{\bm{\pi}}}_{x}(t,\bar x,G^{\hat{\bm{\pi}}}(t,\bar x,\alpha))}{F^{\hat{\bm{\pi}}}_{xx}(t,\bar x,G^{\hat{\bm{\pi}}}(t,\bar x,\alpha))} - a_1(t)\bar x,\quad x\in \bbX^{x_0,\hat{\bm{\pi}}}_t\backslash \mathbb{S}^{\hat{\bm{\pi}}}_t.
  \end{align*}
 It is obvious that the right-hand side does not depend on the choice of $\bar x$, and we denote it by $a_0(t)$. Consequently,
 \begin{align*}
   -\frac{F^{\hat{\bm{\pi}}}_{x}(t,x,G^{\hat{\bm{\pi}}}(t,x,\alpha))}{F^{\hat{\bm{\pi}}}_{xx}(t,x,G^{\hat{\bm{\pi}}}(t,x,\alpha))} = a_0(t) + a_1(t)x,\quad x\in \bbX^{x_0,\hat{\bm{\pi}}}_t\backslash \mathbb{S}^{\hat{\bm{\pi}}}_t.
 \end{align*}
Combining the above with \eqref{eq:LeNeccPDEEq1}, we immediately conclude that $\theta_0(t) = a_0(t) v^*(t)$ and $\theta_1(t) = a_1(t)v^*(t)$.

Next, we prove that $a_0,a_1\in C([\underline t,t^*])$. For any fixed $t_1\in (\underline t,t^*)$, because $\bbX^{x_0,\hat{\bm{\pi}}}_{t_1}$ is a nonempty interval, we can find distinct $x'\neq x''\in \bbX^{x_0,\hat{\bm{\pi}}}_{t_1}\backslash\{\xi\}\subseteq \bbX^{x_0,\hat{\bm{\pi}}}_{t}\backslash \mathbb{S}^{\hat{\bm{\pi}}}_t$, $t\in [t_1,t^*)$. Then, we have
\begin{align*}
  a_1(t) =
  \left[\left(-\frac{F^{\hat{\bm{\pi}}}_{x}(t,x',G^{\hat{\bm{\pi}}}(t,x',\alpha))}{F^{\hat{\bm{\pi}}}_{xx}(t,x',G^{\hat{\bm{\pi}}}(t,x',\alpha))}\right) - \left(-\frac{F^{\hat{\bm{\pi}}}_{x}(t,x'',G^{\hat{\bm{\pi}}}(t,x'',\alpha))}{F^{\hat{\bm{\pi}}}_{xx}(t,x'',G^{\hat{\bm{\pi}}}(t,x'',\alpha))}\right)\right]/(x'-x''),\quad t\in[t_1,t^*).
\end{align*}
By   Corollary 3,  Theorem 2-(iii) and Corollary 2 in \citet{HeJiang2020:LinearSDE}, $F^{\hat{\bm{\pi}}}_{x}(t,x,G^{\hat{\bm{\pi}}}(t,x,\alpha))$ and $F^{\hat{\bm{\pi}}}_{xx}(t,x,G^{\hat{\bm{\pi}}}(t,x,\alpha))$ are continuous in $t\in [0,t^*)$ for any $x\neq \xi$. As a result, $a_1\in C([t_1,t^*))$. Similarly, $a_0\in C([t_1,t^*);\mathbb{R})$. Because $t_1$ is arbitrary, we have $a_0,a_1\in C((\underline t,t^*))$. For any $t\in (\underline{t},t^*)$, because $\theta_1(t)=a_1(t)v^*(t)$, we have $a_1(t) = \theta_1(t)\tran v^*(t)/\|v^*(t)\|^2$. Because the limits of $v^*(t)$ and $\theta_1(t)$ exist when $t$ converges to $\underline{t}$ and because the former limit is not zero, we conclude that the limit of $a_1(t)$ as $t$ goes to $\underline{t}$ exists. Similarly, the limit of $a_1(t)$ as $t$ goes to $t^*$ exists. Thus, we can extend the definition of $a_1$ to the domain $[\underline{t},t^*]$ and $a_1\in C([\underline t,t^*])$. Similarly, we can show that $a_0\in C([\underline t,t^*])$.

Finally, for any $t\in(\underline{t},t^*)$ and $x\in \bbX^{x_0,\hat{\bm{\pi}}}_{t}\backslash \mathbb{S}^{\hat{\bm{\pi}}}_t$, Corollary 3 and Theorem 2-(i),(ii) in \citet{HeJiang2020:LinearSDE}   yield that
\begin{align*}
  F^{\hat{\bm{\pi}}}_{t}(t,x,G^{\hat{\bm{\pi}}}(t,x,\alpha)) + F^{\hat{\bm{\pi}}}_{x}(t,x,G^{\hat{\bm{\pi}}}(t,x,\alpha))b(t)\tran \hat{\bm{\pi}}(t,x) + \frac{1}{2}F^{\hat{\bm{\pi}}}_{xx}(t,x,G^{\hat{\bm{\pi}}}(t,x,\alpha))\|\sigma(t)\hat{\bm{\pi}}(t,x)\|^2= 0.
\end{align*}
Combining the above with \eqref{eq:NeccStrategy1} and \eqref{eq:NeccStrategy2} and noting that $F^{\hat{\bm{\pi}}}_{x}(t,x,G^{\hat{\bm{\pi}}}(t,x,\alpha))\neq 0$ and thus $a_0(t)+a_1(t)x\neq 0$, we obtain
\begin{align*}
  F^{\hat{\bm{\pi}}}_{t}(t,x,G^{\hat{\bm{\pi}}}(t,x,\alpha)) + \frac{1}{2}F^{\hat{\bm{\pi}}}_{x}(t,x,G^{\hat{\bm{\pi}}}(t,x,\alpha))\rho(t)\big(a_0(t)+a_1(t)x\big)= 0.
\end{align*}
 Corollary 3 and Corollary 2-(ii), (iii) in \citet{HeJiang2020:LinearSDE}   show that
\begin{align*}
  G^{\hat{\bm{\pi}}}_x(t,x,\alpha) = -\frac{F^{\hat{\bm{\pi}}}_x(t,x,G^{\hat{\bm{\pi}}}(t,x,\alpha))}{F^{\hat{\bm{\pi}}}_y(t,x,G^{\hat{\bm{\pi}}}(t,x,\alpha))},\quad G^{\hat{\bm{\pi}}}_t(t,x,\alpha) = -\frac{F^{\hat{\bm{\pi}}}_t(t,x,G^{\hat{\bm{\pi}}}(t,x,\alpha))}{F^{\hat{\bm{\pi}}}_y(t,x,G^{\hat{\bm{\pi}}}(t,x,\alpha))}.
\end{align*}
As a result, we derive
  \begin{align*}
  G^{\hat{\bm{\pi}}}_t(t,x,\alpha)  + \frac{1}{2}G^{\hat{\bm{\pi}}}_x(t,x,\alpha) \rho(t)\big(a_0(t)+a_1(t)x\big)= 0.
\end{align*}
Corollary 3 and Corollary 2-(iv) in \citet{HeJiang2020:LinearSDE} show  that for any $x\in \bbR$,
\begin{align*}
  \lim_{t\uparrow t^*,x'\rightarrow x}G^{\hat{\bm{\pi}}}(t,x',\alpha) =x.
\end{align*}
The proof then completes.\halmos
\end{pfof}

\begin{pfof}{Lemma \ref{le:TransportEq}}
  We prove (i) first. Because $\gamma_0,\gamma_1\in C([\tau_1,\tau_2])$, there exists $L>0$ such that $\sup_{s\in[\tau_1,\tau_2]}|\gamma_i(s)|\le L$, $i=0,1$. Set
  \begin{align*}
    \underline{x}:=e^{L(\tau_2-\tau_1)}\big(\bar c+L e^{L(\tau_2-\tau_1)} (\tau_2-\tau_1)\big).
  \end{align*}
  Fix any $t\in[\tau_1,\tau_2)$ and $x> \underline{x}$, consider
  \begin{align*}
    \varphi(s) = g(s,h_1(s)x + h_0(s)),\quad s\in[t,\tau_2],
  \end{align*}
  where
  \begin{align*}
    h_0(s) = \int_t^se^{\int_\tau^s \gamma_1(z)dz}\gamma_0(\tau)d\tau,\quad h_1(s) =e^{\int_t^s\gamma_1(\tau)d\tau},\quad s\in [t,\tau_2].
  \end{align*}
 Straightforward calculation yields
  \begin{align}\label{eq:LeTransEqhfunEq}
    h_1'(s)x + h_0'(s) = \gamma_0(s) + \gamma_1(s)\big(h_1(s)x + h_0(s)\big),\quad s\in [t,\tau_2].
  \end{align}
  For any $s\in[t,\tau_2]$,
  \begin{align*}
    h_1(s)x + h_0(s)\ge e^{-L(\tau_2-\tau_1)} x - Le^{L(\tau_2-\tau_1)} (\tau_2-\tau_1)> e^{-L(\tau_2-\tau_1)} \underline{x} - L e^{L(\tau_2-\tau_1)} (\tau_2-\tau_1) =\bar c,
  \end{align*}
  where the first inequality is the case because $\sup_{s\in[\tau_1,\tau_2]}|\gamma_i(s)|\le L$, $i=0,1$ and the second inequality is the case because $x> \underline{x}$. Therefore, $\varphi(s), s\in [t,\tau_2]$ is well defined and is differentiable in $s\in[t,\tau_2)$. Moreover, applying the chain rule, we derive
  \begin{align*}
    \varphi'(s) &= g_t(s,h_1(s)x + h_0(s)) + g_x(s,h_1(s)x + h_0(s))\big(h_1'(s)x + h_0'(s)\big)\\
    & =  g_t(s,h_1(s)x + h_0(s)) + g_x(s,h_1(s)x + h_0(s))\big(h_1(s)\gamma_1(s) x + \gamma_0(s) + h_0(s) \gamma_1(s)\big)\\
    & = 0,\quad s\in[t,\tau_2),
  \end{align*}
  where the second equality follows from \eqref{eq:LeTransEqhfunEq} and the third follows from the differentiable equation satisfied by $g$. As a result,
  \begin{align*}
    g(t,x) = \varphi(t) = \varphi(\tau_2) = g\left(\tau_2,h_1(\tau_2)x + h_0(\tau_2)\right) = \alpha_0 + \alpha_1 \left(h_1(\tau_2)x + h_0(\tau_2)\right) = \hat g(t,x),
  \end{align*}
  where the fourth equality is the case due to the terminal condition satisfied by $g$ at $\tau_2$.

  Part (ii) can be proved similarly. For (iii), because $\gamma_0(t)+\gamma_1(t)x\ge 0$ for all $(t,x)\in {\cal D}_3$ and because $c(t)$ is decreasing, for each fixed $(t,x)\in {\cal D}_3$, \eqref{eq:LeTransEqhfunEq} shows that $h_1'(s)x + h_0'(s)\ge 0$, $h_1(s)x + h_0(s)$ is increasing in $s\in[t,\tau_2]$ and thus $(s,h_1(s)x + h_0(s))\in {\cal D}_3, s\in[t,\tau_2)$. Then, following the same proof as in part (i) of the lemma, we conclude that $g(t,x)=\hat g(t,x)$.
  \halmos
\end{pfof}

\begin{pfof}{Lemma \ref{le:EquilibriumQuantile}}
Recall the form of the strategy $\hat{\bm{\pi}}$ as in \eqref{eq:NeccStrategy2}. Suppose $\bbX^{x_0,\hat{\bm{\pi}}}_t \supseteq (\underline{x}_t,+\infty)$ for certain $\underline{x}_t\in \bbR$. Then, $\xi\notin [\underline{x}_t',+\infty)$ for certain $\underline{x}_t'\ge \underline{x}_t$ and \eqref{quantile equation} holds in the region $[t,t^*]\times (\underline{x}_t',+\infty)$. Because $G^{\hat{\bm{\pi}}}\in C([t,t^*]\times (\underline{x}_t',+\infty))$ (due to  Corollary 3 and Corollary 2-(ii) in \citet{HeJiang2020:LinearSDE}) and because there exists a partition of $[t,t^*]$, $t=:\tau_0<\tau_1<\dots<\tau_N=t^*$, such that $a_0,a_1,\rho\in C([\tau_{i-1},\tau_i])$ and $G^{\hat{\bm{\pi}}}\in C^{1,\infty}([\tau_{i-1},\tau_i)\times (\underline{x}_t',+\infty))$ (due to  Corollary 3 and Corollary 2-(iii) in \citet{HeJiang2020:LinearSDE}) $i=1,\dots, N$, by applying Lemma \ref{le:TransportEq} in $[\tau_{i-1},\tau_i)$ sequentially, we conclude that there exists $\underline{c}_t>\max(\underline{x}_t,\xi)$ such that \eqref{eq:EquilibriumQuantile} holds for all $(s,x)\in [t,t^*]\times [\underline{c}_t,+\infty)$. Similarly, in the case $\bbX^{x_0,\hat{\bm{\pi}}}_t \supseteq (-\infty,\bar{x}_t)$ for certain $\bar{x}_t\in \bbR$, we can find $\bar c_t<\min(\xi, \bar{x}_t)$ such that \eqref{eq:EquilibriumQuantile} holds all $(s,x)\in [t,t^*]\times (-\infty,\bar{c}_t)$.

Next, we prove (i). For the sake of contradiction, suppose $a_1(s)=0,s\in[t,t^*)$. Then, we must have $t_*=t^*$; otherwise $\theta_0(s)=-\xi \theta_1(s) = -\xi a_1(s)v^*(s)=0$ for $s\in [t_*\vee t,t^*)$, contradicting the definition of $t^*$. Moreover, by the definition of $t^*$, for any $s\in[t,t^*)$, there exists $\tau\in[s,t^*)$ with $\theta_0(\tau)\neq 0$ and thus $a_0(\tau)\neq 0$.
Also note that
\begin{align*}
  X^{\hat{\bm{\pi}}}_{s,x}(T) = x + \int_s^{t^*}a_0(\tau)v^*(\tau)\tran b(\tau)d\tau +\int_s^{t^*}a_0(\tau)v^*(\tau)\tran \sigma(\tau)dW(\tau),
\end{align*}
which is a normal random variable. Thus, we have
\begin{align}\label{eq:LeEquiQuantEq0}
  G^{\hat{\bm{\pi}}}(s,x,\alpha) = x + \int_s^{t^*}a_0(\tau)v^*(\tau)\tran b(\tau)d\tau + \Phi^{-1}(\alpha)\sqrt{\int_s^{t^*}\|\sigma(\tau)\tran v^*(\tau)\|^2a_0(\tau)^2d\tau}.
\end{align}
Lemma \ref{le:NeccPDE}-(i) shows that it is either the case in which $\bbX^{x_0,\hat{\bm{\pi}}}_t = (\underline{x}_t,+\infty)$ for certain $\underline{x}_t\in \bbR$, or the case in which $\bbX^{x_0,\hat{\bm{\pi}}}_t = (-\infty,\bar{x}_t)$ for certain $\bar{x}_t\in \bbR$, or the case $\bbX^{x_0,\hat{\bm{\pi}}}_t =\bbR$. In either case, as we already proved, there exists a nonempty open interval $I\subset \bbR$ such that \eqref{eq:EquilibriumQuantile} holds for any $s\in[t,t^*)$ and $x\in I$. Comparing \eqref{eq:EquilibriumQuantile} and \eqref{eq:LeEquiQuantEq0}, we conclude
\begin{align*}
    &\beta_0(s,\alpha) = \int_s^{t^*}a_0(\tau)v^*(\tau)\tran b(\tau)d\tau + \Phi^{-1}(\alpha)\sqrt{\int_s^{t^*}\|\sigma(\tau)\tran v^*(\tau)\|^2a_0(\tau)^2d\tau},\\
    &\beta_1(s,\alpha)=1,\quad s\in[t,t^*).
\end{align*}
The above immediately yields that $a_1(s)=0,s\in[t,t^*)$ and
\begin{align}
  &-\Phi^{-1}(\alpha)\sqrt{\int_s^{t^*}\|\sigma(\tau)\tran v^*(\tau)\|^2a_0(\tau)^2d\tau} = \int_s^{t^*}a_0(\tau)v^*(\tau)\tran b(\tau)d\tau -\frac{1}{2} \int_s^{t^*}a_0(\tau)\rho(\tau)d\tau \notag \\
  & = \frac{1}{2}\int_s^{t^*}a_0(\tau)\|\sigma(\tau)\tran v^*(\tau)\|^2d\tau,\quad s\in[t,t^*),\label{eq:LeEquiQuantEq1}
\end{align}
where the second equality is due to the definition of $\rho$. When $\alpha =1/2$, because $\|\sigma(s)\tran v^*(s)\|^2>0, s\in[0,T)$, \eqref{eq:LeEquiQuantEq1} implies that $a_0(s)=0,s\in[t,t^*)$, which is a contradiction. When $\alpha\neq 1/2$, taking square and then taking derivative with respect to $s$ on both sides of \eqref{eq:LeEquiQuantEq1}, and noting that $\|\sigma(s)\tran v^*(s)\|^2>0, s\in[0,T)$ and that $a_0(s)>0, s\in[t,t^*)$   because of \eqref{eq:NeccStrategy1} and $a_1(s)=0$, $s\in [t, t^*)$, we derive the following integral equation
\begin{align*}
  \big(\Phi^{-1}(\alpha)\big)^2 a_0(s) - \frac{1}{2}\int_s^{t^*}a_0(\tau)\|\sigma(\tau)\tran v^*(\tau)\|^2d\tau = 0,\quad s\in[t,t^*).
\end{align*}
Then, $g(s):=\int_s^{t^*}a_0(\tau)\|\sigma(\tau)\tran v^*(\tau)\|^2d\tau = 0$ satisfies
\begin{align*}
  \big(\Phi^{-1}(\alpha)\big)^2 g'(s) + \frac{1}{2}\|\sigma(s)\tran v^*(s)\|^2 g(s) = 0,\; s\in[t,t^*),\quad g(t^*)=0.
\end{align*}
Because $\big(\Phi^{-1}(\alpha)\big)^2\neq 0$ and $\|\sigma(s)\tran v^*(s)\|$ is bounded in $s\in[t,t^*)$, we derive $g(s)=0,s\in[t,t^*]$, i.e., $a_0(s)=0,s\in[t,t^*]$, which is a contradiction.

Next, we suppose $\bbX^{x_0,\hat{\bm{\pi}}}_t \supseteq (\underline{x}_t,+\infty)$ for certain $\underline{x}_t\in \bbR$ and prove \eqref{eq:a1Eq1}.
Straightforward calculation yields that
\begin{align*}
   X^{\hat{\bm{\pi}}}_{t,x}(T) = x \tilde Z_1(T;t) + \tilde Z_2(T;t).
\end{align*}
where
\begin{align}\label{eq:Ztilde1}
  d\tilde Z_1(s;t)&=\tilde Z_1(s;t) \left[b(s)\tran  \theta_1(s)ds  + \theta_1(s) \tran \sigma(s) dW(s)\right],
  \;s\in[t,T],\quad \tilde Z_1(t;t)=1.\\
    d\tilde Z_2(s;t)&=\left(b(s)\tran  \theta_0(s)+b(s)\tran  \theta_1(s)\tilde Z_2(s;t)\right)ds  +\left(\theta_0(s) \tran \sigma(s)+\theta_1(s) \tran \sigma(s) \tilde Z_2(s;t)\right) dW(s), \notag\\
& \qquad s \in [t, T], \quad \tilde Z_2(t;t)=0.\label{eq:Ztilde2}
\end{align}
For any $x>\max(\underline{c}_t,0)$, denote by $G^Z(\alpha),\alpha\in(0,1)$ and $G^Y(x,\alpha),\alpha\in(0,1)$ the right-continuous quantile functions of $\tilde Z_1(T;t)$ and $Y^{\hat{\bm{\pi}}}_{t,x}(T):=X^{\hat{\bm{\pi}}}_{t,x}(T)/x$, respectively. We already proved part (i) of the lemma, which implies that $\tilde Z_1(T;t)$ is a non-degenerate lognormal random variable and thus its quantile function is continuous. Also note that $\lim_{x\uparrow +\infty}Y^{\hat{\bm{\pi}}}_{t,x}(T)= \tilde Z_1(T;t)$ almost surely. As a result,
\begin{align*}
  \lim_{x\uparrow +\infty} \frac{ G^{\hat{\bm{\pi}}}(t,x,\alpha)}{x}=\lim_{x\uparrow +\infty} G^{Y}(x,\alpha)=G^{Z}(\alpha),\quad \forall \alpha \in (0,1).
\end{align*}
By computing $G^{Z}(\alpha)$ from \eqref{eq:Ztilde1} and recalling \eqref{eq:EquilibriumQuantile}, we then derive \eqref{eq:a1Eq1}.

Finally, \eqref{eq:a1Eq2} can be proved similarly.\halmos
\end{pfof}

\begin{pfof}{Proposition \ref{prop:NeccAlphaNotHalf}}
 Suppose that $\hat{\bm{\pi}}\in\mathbb{A}$ is an equilibrium strategy for a given $\alpha \neq 1/2$. Recall $t^*$ as defined in Lemma \ref{le:QuantileDifferentiability} and $\underline{t}$ as defined in Lemma \ref{le:NeccNonDegAlphaLarge}. We prove that it cannot be the case that $\underline{t}<t^*$.

 For the sake of contradiction, suppose $\underline{t}<t^*$. Combining Lemma \ref{le:NeccPDE}-(i) and Lemma \ref{le:EquilibriumQuantile}, we conclude that it is either the case in which \eqref{eq:a1Eq1} holds for any $t\in (\underline{t},t^*)$ or the case in which \eqref{eq:a1Eq2} holds for any $t\in (\underline{t},t^*)$. Denote $\lambda=\Phi^{-1}(\alpha)$ in the first case and $\lambda=\Phi^{-1}(1-\alpha)$ in the second case. Cauchy's inequality implies
 \begin{align*}
   \left|\int_t^{t^*} a_1(s)\|\sigma(s)\tran v^*(s)\|^2 ds\right|\le \sqrt{\int_t^{t^*} \|\sigma(s)\tran v^*(s)\|^2 ds}\sqrt{\int_t^{t^*} a_1(s)^2\|\sigma(s)\tran v^*(s)\|^2 ds}.
 \end{align*}
 Combining the above with \eqref{eq:a1Eq1} and \eqref{eq:a1Eq2}, we obtain
 \begin{align*}
  \left|\int_t^{t^*} a_1(s)^2\|\sigma(s)\tran v^*(s)\|^2ds-  2\lambda\sqrt{\int_t^{t^*} a_1(s)^2 \|\sigma(s)\tran v^*(s)\|^2 ds}\right|\\
  \le \sqrt{\int_t^{t^*} \|\sigma(s)\tran v^*(s)\|^2 ds}\sqrt{\int_t^{t^*} a_1(s)^2\|\sigma(s)\tran v^*(s)\|^2 ds}.
 \end{align*}
 Because of \eqref{eq:a1Eq0}, we can divide $\sqrt{\int_t^{t^*} a_1(s)^2\|\sigma(s)\tran v^*(s)\|^2 ds}$ on both sides of the above inequality to obtain
 \begin{align*}
  \left|\sqrt{\int_t^{t^*} a_1(s)^2\|\sigma(s)\tran v^*(s)\|^2ds}-  2\lambda\right|
  \le \sqrt{\int_t^{t^*} \|\sigma(s)\tran v^*(s)\|^2 ds}.
 \end{align*}
 Sending $t$ to $t^*$ on both sides of the above inequality and noting that $\lambda \neq 0$ because $\alpha\neq 1/2$, we arrive at contradiction.

 We already proved that $\underline{t}\ge t^*$. On the other hand, Lemma \ref{le:NeccNonDegAlphaLarge} shows that if $\hat{\bm{\pi}}\in\mathbb{A}$ is an equilibrium strategy for $\alpha \in (1/2,1)$, we must have $\underline{t}=0$ and $t^*=T$. Thus, any $\hat{\bm{\pi}}\in\mathbb{A}$ is not an equilibrium strategy for $\alpha \in (1/2,1)$. For $\alpha\in (0,1/2)$, by Lemma \ref{le:NeccNonDegAlphaLarge}, $\underline{t}\ge t^*$ implies $\underline{t}=T$, which is the case if and only if $\hat{\bm{\pi}}(s,x_0)=0,s\in[0,T)$, i.e., $\hat{\bm{\pi}}$ takes the form in \eqref{tail risk equilibrium Strategy}. On the other hand, we already proved that any $\hat{\bm{\pi}}$ as given by \eqref{tail risk equilibrium Strategy} is an equilibrium strategy. The proof then completes.\halmos
\end{pfof}

\begin{pfof}{Proposition \ref{prop:NeccAlphaHalf}}
  Suppose that $\hat{\bm{\pi}}\in\mathbb{A}$ is an equilibrium strategy for certain $\alpha = 1/2$. Recall $t^*$ as defined in Lemma \ref{le:QuantileDifferentiability} and $\underline{t}$ as defined in Lemma \ref{le:NeccNonDegAlphaLarge}. Lemma \ref{le:NeccNonDegAlphaLarge} shows that $\underline{t}=0$ and $t^*=T$. Then, Lemma \ref{le:EquilibriumQuantile}-(ii) and (iii) imply that
  \begin{align*}
    \frac{1}{2}\int_t^{t^*} a_1(s)\big(1-a_1(s)\big)\|\sigma(s)\tran v^*(s)\|^2ds=0,\quad t\in (0,t^*),
  \end{align*}
  which implies that $a_1(t)(1-a_1(t))=0,t\in(0,t^*)$. Lemma \ref{le:EquilibriumQuantile}-(i) shows for any $t\in (0,t^*)$, $a_1(s)\neq 0$ for some $s\in [t,t^*)$. Then we must have $a_1(t)=1,t\in (0,T)$ because $a_1\in C([0,T];\bbR)$ as shown by Lemma \ref{le:NeccPDE}-(ii). Then, \eqref{eq:NeccStrategy2} implies that
  \begin{align}\label{eq:PropNeccAlphaHalfEq1}
    \hat{\bm{\pi}}(t,x) = (a_0(t) + x)v^*(t),\quad t\in[0,T),x\in \bbR.
  \end{align}
  What remains is to prove that $a_0(t)$ is constant in $t\in [0,T]$.

According to Lemma \ref{le:NeccPDE}-(i), \eqref{eq:NeccStrategy1},  and $a_1(t)=1,t\in (0,T)$, we have  $\bbX^{x_0,\hat{\bm{\pi}}}_t=(\underline{x}(t),+\infty)$   for some $\underline{x}(t)\in \bbR$.
  Because of \eqref{eq:PropNeccAlphaHalfEq1},  Corollary 4   in \citet{HeJiang2020:LinearSDE} shows that $a_0(t)$ is increasing in $t\in(0,T]$ and $\bbX^{x_0,\hat{\bm{\pi}}}_t=(-a_0(t),+\infty), t\in (0,T]$. By the definition of $t_*$ and $\xi$, we have $t_*=\inf\{t\in[0,T]: a_0(s)=a_0(T),s\in[t,T]\}$ and $\xi=-a_0(T)$. Because $a_0(t)$ is increasing in $t\in (0,T]$, we conclude that $\bbX^{x_0,\hat{\bm{\pi}}}_t\backslash \mathbb{S}^{\hat{\bm{\pi}}}_t = (-a_0(t),+\infty),t\in(0,T]$. Now, Lemma \ref{le:NeccPDE}-(ii) and Lemma \ref{le:TransportEq}-(iii) yield that \eqref{eq:EquilibriumQuantile} holds for any $x>-a_0(t),t\in(0,T]$.

   Now, for the sake of contradiction, suppose $a_0(t)$ is not constant in $t\in [0,T]$, which is equivalent to $t_*>0$. Fix any $t\in (0,t_*)$. Then, $a_0(s)$ is not a constant in $s\in[t,t_*]$, which means that $da_0(s)$ defines a positive measure on $[t,T]$.
Applying It\^o's lemma, recalling \eqref{eq:PropNeccAlphaHalfEq1}, Lemma \ref{le:QuadProgSolution} and $\tilde Z_1$ as in \eqref{eq:Ztilde1}, and noting that $a_0\in C([0,T])$ and $a_1\equiv 1$, we derive
\begin{align*}
  d\left(\frac{a_0(s)+ X^{\hat{\bm{\pi}}}_{t,x}(s) }{\tilde Z_1(s;t)}\right) = \frac{da_0(s)}{\tilde Z_1(s;t)},\quad s\in[t,T),\; x\in \bbR.
\end{align*}
Then, we obtain
\begin{align}
  X^{\hat{\bm{\pi}}}_{t,-a_0(t)}(T) &= -a_0(T) + \int_t^T\frac{\tilde Z_1(T;t)}{\tilde Z_1(s;t)}da_0(s),\notag\\
  & = -a_0(T) + \int_t^Te^{\int_s^T \left(b(\tau)\tran v^*(\tau) - \frac{1}{2}\|\sigma(\tau)\tran v^*(\tau)\|^2d\tau\right)d\tau + \int_s^T v^*(\tau)\tran \sigma(\tau)dW(\tau) }da_0(s),\label{eq:XTwitha0}
\end{align}
where the second equality is the case due to the definition of $\tilde Z_1$ and because $a_1(s)=1$ and thus $\theta_1(s)=v^*(s),s\in[0,T)$. Then, recalling $t^*=T$, $\beta_0,\beta_1$ as defined in \eqref{eq:EquiQuantCoeff}, and that $\rho(s) = 2b(s)\tran v^*(s) - \|\sigma(s)\tran v^*(s)\|^2$, we derive
\begin{align}
  &X^{\hat{\bm{\pi}}}_{t,-a_0(t)}(T)+ a_0(t)\beta_1(t,1/2) - \beta_0(t,1/2)\notag\\
  &=-a_0(T) + \int_t^Te^{\frac{1}{2}\int_s^T\rho(\tau)d\tau + \int_s^T v^*(\tau)\tran \sigma(\tau)dW(\tau) }da_0(s)\notag\\
   &\quad +  a_0(t) e^{\frac{1}{2}\int_t^T\rho(s)ds} - \frac{1}{2}\int_t^T a_0(s)\rho(s) e^{\frac{1}{2}\int_s^T\rho(\tau)d\tau}ds\notag\\
   & = \int_t^Te^{\frac{1}{2}\int_s^T\rho(\tau)d\tau + \int_s^T v^*(\tau)\tran \sigma(\tau)dW(\tau) }da_0(s) - \int_t^T e^{\frac{1}{2}\int_s^T\rho(\tau)d\tau} da_0(s),\label{eq:PropNeccAlphaHalfEq3}
\end{align}
where second equality is the case because
\begin{align*}
  d\left(a_0(s) e^{\frac{1}{2}\int_s^T\rho(\tau)d\tau}\right) = e^{\frac{1}{2}\int_s^T\rho(\tau)d\tau} da_0(s) - \frac{1}{2}\rho(s)e^{\frac{1}{2}\int_s^T\rho(\tau)d\tau} a_0(s)ds.
\end{align*}

Next, we prove
\begin{align}\label{eq:PropNeccAlphaHalfEq2}
  \prob\big(X^{\hat{\bm{\pi}}}_{t,-a_0(t)}(T)\le -a_0(t)\beta_1(t,1/2) + \beta_0(t,1/2)\big)<1/2.
\end{align}
According to \eqref{eq:PropNeccAlphaHalfEq3}, one can see that \eqref{eq:PropNeccAlphaHalfEq2} is equivalent to
\begin{align}\label{eq:PropNeccAlphaHalfEq4}
  \prob\left(\int_t^T\left(e^{M(s)} - 1\right)d\mu(s)\le 0\right)<1/2,
\end{align}
where $M(s):=\int_s^T v^*(\tau)\tran \sigma(\tau) dW(\tau)$ and $d\mu(s):=e^{\frac{1}{2}\int_s^T\rho(\tau)d\tau}da_0(s)$, $s\in [0, T]$. Because $da_0(s)$ is a positive, atomless measure on $[t,T]$, so is $d\mu(s)$.
    In the following, we first prove that
\begin{align}\label{eq:PropNeccAlphaHalfEq5}
  \prob\left(\int_t^T\left(e^{M(s)} - 1\right)d\mu(s)>0,\int_t^TM(s)ds\le 0\right)>0.
\end{align}
To this end, we first construct a path $m(s),s\in[t,T]$ of the stochastic process $M(s),s\in[0,T]$ such that
\begin{align}\label{eq:PropNeccAlphaHalfEq6}
  \int_t^T\left(e^{m(s)} - 1\right)d\mu(s)>0,\quad \int_t^Tm(s)d\mu(s)<0.
\end{align}
The construction is as follows: First, there exists $\zeta_1\in \bbR$ such that $\int_t^T(T-s)(\zeta_1-s)d\mu(s) = 0$.
Because $e^{x}-1>x$ for any $x\neq 0$ and because $d\mu(s)$ is a positive, atomless measure on $[t,T]$, we have $\int_t^T\left(e^{(T-s)(\zeta_1-s)} - 1\right)d\mu(s)>0$. Thus, there exists $\zeta_2<\zeta_1$ such that $\int_t^T\left(e^{(T-s)(\zeta_2-s)} - 1\right)d\mu(s)>0$. As a result,
\begin{align*}
  \int_t^T(T-s)(\zeta_2-s)d\mu(s) = \int_t^T(T-s)(\zeta_1-s)d\mu(s) + (\zeta_2-\zeta_1)\int_t^T(T-s)d\mu(s)<0,
\end{align*}
where the inequality is the case because $\zeta_2<\zeta_1$ and $\int_t^T(T-s)d\mu(s)>0$. Thus, setting $m(s):=(T-s)(\zeta_2-s),s\in[0,T]$, \eqref{eq:PropNeccAlphaHalfEq6} holds. As a result, there exists $\epsilon>0$ such that
\begin{align}\label{eq:PropNeccAlphaHalfEq7}
  \int_t^T\left(e^{m(s)-\epsilon} - 1\right)d\mu(s)>0,\quad \int_t^T(m(s)+\epsilon)d\mu(s)<0.
\end{align}

Define $Y(s):=M(t)-M(s) = \int_t^sv^*(\tau)\tran\sigma(\tau)dW(\tau),s\in[t,T]$ and
\begin{align*}
 f(s):=m(t)-m(s) = (T-t)(\zeta_2-t)-(T-s)(\zeta_2-s),s\in[0,T],
\end{align*}
so $f$ can be considered to be a path of $Y$. Defining an absolutely continuous function $w(s),s\in[0,T]$ by
\begin{align*}
  w(t)=0,\quad w'(s) = \frac{\sigma(s)\tran v^*(s)}{\|\sigma(s)\tran v^*(s)\|^2}f'(s),\; s\in[t,T],
\end{align*}
so we have $f(s) =\int_t^s v^*(z)\tran \sigma(z) w'(z)dz,s\in[0,T]$. Then, Lemma 2 of \citet{HeJiang2020:LinearSDE} shows that $f(s),s\in[t,T]$ is in the support of $Y(s), s\in[t,T]$, with the latter being viewed as a random variable taking values in the space of continuous functions on $[t,T]$ with the maximum norm.
As a result, for any $\epsilon>0$, in particular the one satisfying \eqref{eq:PropNeccAlphaHalfEq7}, we have
\begin{align*}
  &\prob\left(\sup_{s\in[t,T]}|M(s)-m(s)|<\epsilon\right)=\prob\left(\sup_{s\in[t,T]}|Y(T)-Y(s)-(f(T)-f(s))|<\epsilon\right)\\
   &\ge  \prob\left(\sup_{s\in[t,T]}|Y(s)-f(s)|<\epsilon/2\right)>0.
\end{align*}
Combining the above with \eqref{eq:PropNeccAlphaHalfEq7}, we derive
\begin{align*}
  &\prob\left(\int_t^T\left(e^{M(s)} - 1\right)d\mu(s)>0,\int_t^TM(s)ds\le 0\right)\ge \prob\left(\sup_{s\in[t,T]}|M(s)-m(s)|<\epsilon\right)>0,
\end{align*}
i.e., \eqref{eq:PropNeccAlphaHalfEq5} is proved.

Now, recalling that \eqref{eq:PropNeccAlphaHalfEq5} and noting that  $\int_t^T\left(e^{M(s)} - 1\right)d\mu(s)\ge \int_t^TM(s)d\mu(s)$ because $e^x-1\ge x$ for any $x\in \bbR$, we derive
\begin{align*}
  &\prob\left(\int_t^T\left(e^{M(s)} - 1\right)d\mu(s)\le 0\right) \\
  =& \prob\left(\int_t^TM(s)d\mu(s)\le 0\right)  - \prob\left(\int_t^TM(s)d\mu(s)\le 0, \int_t^T\left(e^{M(s)} - 1\right)d\mu(s)>0\right)\\
  <&\prob\left(\int_t^TM(s)d\mu(s)\le 0\right) = 1/2,
\end{align*}
where the last equality is the case because $\int_t^TM(s)d\mu(s)$ is a normal random variable with zero mean. Thus,
\eqref{eq:PropNeccAlphaHalfEq2} is proved.

 Because $t<t_*$,  Corollary 3 and Corollary 2-(ii) in \citet{HeJiang2020:LinearSDE} shows that $  G^{\hat{\bm{\pi}}}(t,x,\alpha)$ is continuous in $(x,\alpha)\in \bbR \times (0, 1)$. Then, we conclude from \eqref{eq:PropNeccAlphaHalfEq2} and  Corollary 3 and Corollary 2-(ii) in \citet{HeJiang2020:LinearSDE} that $ G^{\hat{\bm{\pi}}}(t,- a_0(t),\frac{1}{2})-[\beta_0(t,\frac{1}{2}) - a_0(t)\beta_1(t,\frac{1}{2}) ]>0$, and  there exists $\delta_t>0$ such that
\begin{align}\label{MedianLimitIn a0 ForT}
G^{\hat{\bm{\pi}}}(t,x,1/2)-\big(\beta_0(t,1/2) +x\beta_1(t,1/2) \big) >0,\quad \forall x\in ( -a_0(t)-\delta_t, -a_0(t) +\delta_t ).
\end{align}
 Because have shown that \eqref{eq:EquilibriumQuantile} holds for any $x>-a_0(t)$,
we arrive at contradiction. Then, $a_0(t)$ must be a constant on $[0,T]$, i.e., $\hat{\bm{\pi}}$ must be given by \eqref{eq:EquiStrategyPropTarget} for some $\xi<x_0$.

On the other hand, we already proved that $\hat{\bm{\pi}}$ as given by \eqref{eq:EquiStrategyPropTarget} for any $\xi<x_0$ is an equilibrium strategy. The proof then completes.\halmos
\end{pfof}

\subsection{Other Proofs}

\begin{pfof}{Theorem \ref{prop:BoundaryPoint}}
   Fix any $t\in (0,T)$. Because $\bbX_s^{x_0,\hat{\bm{\pi}}}=(\xi,+\infty)$ for $s\in (0,T)$, because $\tilde{\bm{\pi}}$ agrees with $\hat{\bm{\pi}}$ for $x\in \bbX_s^{x_0,\hat{\bm{\pi}}},s\in[0,T)$, and because $\tilde{\bm{\pi}}(s,x)$ is continuous in $x$, we have $X^{\tilde{\bm{\pi}}}_{t,\xi}(s)=\xi,s\in[t,T]$. As a result, $G^{\tilde{\bm{\pi}}}(t,\xi,1/2) = \xi$. Recall that $v^*(t)\tran b(t)>0$, so there exists $\epsilon_0\in (T-t)$ such that $v^*(t)\tran b(s)>0,s\in [t,t+\epsilon_0]$. Then, for any $\epsilon\in (0,\epsilon_0)$, we have
\begin{align*}
  \prob\left(X_{t,\xi}^{\tilde{\bm{\pi}}_{t,\epsilon,v^*(t)}}(T)>\xi\right)\ge \prob\left(X_{t,\xi}^{v^*(t)}(t+\epsilon)>\xi\right)>1/2,
\end{align*}
where the first inequality is the case because on $\big\{X_{t,\xi}^{v^*(t)}(t+\epsilon)>\xi\big\}$, we have $X_{t,\xi}^{\tilde{\bm{\pi}}_{t,\epsilon,v^*(t)}}(T) = X_{t+\epsilon,X_{t,\xi}^{v^*(t)}(t+\epsilon)}^{\hat{\bm{\pi}}}(T)>\xi$ and the second inequality is the case because $\int_t^{t+\epsilon}v^*(t)\tran b(s)ds>0$. As a result, $G^{\tilde{\bm{\pi}}_{t,\epsilon,v^*(t)}}(t,\xi,1/2)>\xi$ due to the definition of right-continuous quantile functions. \halmos
\end{pfof}

\begin{pfof}{Proposition \ref{prop:CompwithFracKelly}}
Straightforward calculation leads to
 \begin{align}\label{eq:MedianFractionalKelly}
   G^{\bm{\pi}_{\gamma-\mathrm{Kelly}}}(t,x,1/2) = xe^{(\gamma-\frac{1}{2}\gamma^2)\int_t^T\|\sigma(s)\tran v^*(s)\|^2ds},\quad x>0,\; t\in [0,T),\;\gamma>0.
 \end{align}
 For any $s\in [0,T)$, by Assumption \ref{as:Feasibility2}, we have $\|\sigma(s)\tran v^*(s)\|=b(s)\tran v^*(s) >0$. Because $\gamma \in (0,1)$, we have $\gamma-\frac{1}{2}\gamma^2 \in (0, 1)$.
As a result, we conclude that $a_{t,\gamma}$ is well defined and strictly larger than 1.

Finally, the comparison of $G^{\hat{\bm{\pi}}_\xi}(t,x,1/2)$ and $G^{\bm{\pi}_{\gamma-\mathrm{Kelly}}}(t,x,1/2)$ is straightforward. \halmos
\end{pfof}

\begin{pfof}{Proposition \ref{eq:EUMaximization}}
It is a standard result in portfolio selection that $\hat{\bm{\pi}}_\xi$ maximizes $\expect[\ln (X^{\bm{\pi}}(T)-\xi)]$; see for instance \citet{KaratzasIShreveS:98momf}.\halmos
\end{pfof}

\begin{pfof}{Proposition \ref{prop:PreCommitted}}
Given a portfolio insurance level $\xi$, we consider more generally the pre-committed strategy planned by the agent at time $t$ with wealth level $x_t$; i.e., the portfolio strategy that maximizes the median, conditional at time $t$ with wealth level $x_t$, of the terminal wealth at time $T$. Recall that the risk-free rate $r$ is set to be zero for simplicity in our discussion.

According to Section 4 of \citet{HeXDZhouXY:2011PortfolioChoiceviaQuantiles}, the problem of finding the pre-committed strategy is equivalent to the following:
\begin{align}\label{eq:MedianMaxStatic}
  \begin{array}{cl}
    \underset{X}{\mathrm{Max}} & G_{X}(1/2)\\
    \text{Subject to} & \expect_t\left[\big(\hat \rho(T)/\hat \rho(t)\big) X\right]\le x_t,\quad X\ge \xi,
  \end{array}
\end{align}
where $X$ stands for a possible wealth profile attained at time $T$, $G_X(1/2)$ stands for the median of $X$, $\expect_t$ stands for the expectation conditional on the information at time $t$, and
\begin{align*}
  \hat\rho(s):&=e^{-\int_0^s\frac{1}{2}\|\hat\theta(\tau)\|^2d\tau- \int_0^s \hat\theta(\tau)\tran dW(\tau)},\quad s\in [t,T],\\
  \hat\theta(s):&=\underset{\theta\in \bbR^d}{\mathrm{argmin}}\{\|\theta\|^2:\sigma(s)\theta-b(s)\in K^*\},\quad s\in [t,T),
\end{align*}
where $K:=\{\pi \in \mathbb{R}^m:Q\pi\ge 0\}$ and $K^*$ is the dual cone of $K$, i.e., $K^*:=\{y\in \mathbb{R}^m:y\tran \pi \ge 0,\forall \pi \in K\}$. It is known that $\hat \rho(t)/\hat \rho (s),s\in [t,T]$ is the wealth process of the optimal portfolio under expected logarithmic utility maximization with initial wealth at time $t$ to be 1, and this optimal portfolio is indeed the Kelly portfolio \eqref{eq:Kelly}; see for instance \citet{KaratzasIShreveS:98momf}. Therefore, we immediately conclude that $\hat \theta(s)=\sigma(s)\tran v^*(s)$ and consequently
\begin{align}
  \hat\rho(s):&=e^{-\int_0^s\frac{1}{2}\|\sigma(\tau)\tran v^*(\tau)\|^2d\tau - \int_0^sv^*(\tau)\tran \sigma(\tau) dW(\tau)},\quad s\in [t,T].
\end{align}

By changing of variable, $\tilde X= X-\xi$, problem \eqref{eq:MedianMaxStatic} is equivalent to Problem (2) in \citet{HE2021} with $\rho$, $X$, $m$, and $x$ therein set to be $\hat\rho(T)/\hat \rho(t)$, $\tilde X$, the Dirac measure at $1/2$, and $x_t-\xi$, respectively. Then, Corollary 1 in \citet{HE2021} yields that the optimal solution to \eqref{eq:MedianMaxStatic} is $X^* = \xi + k^*_t\mathbf 1_{\hat \rho(T)/\hat \rho(t)\le \beta^*_t}$, where $\beta^*_t:=F^{-1}_{\hat \rho(T)/\hat \rho(t)}(1/2)$ and $k_t^*:= \frac{x_t-\xi}{\expect[\hat \rho(T)/\hat \rho(t) \mathbf 1_{\hat \rho(T)/\hat \rho(t) \le \beta^*_t}]}$. Denoting $Z(t):=\int_0^tv^*(s)\tran \sigma(s)dW(s),t\in[0,T]$, we have $\log(\hat \rho(t)) = -Z(t) - \frac{1}{2}\int_0^t\|\sigma(s)\tran v^*(s)\|^2ds$. As a result, we can rewrite the optimal solution to \eqref{eq:MedianMaxStatic} as follows:
\begin{align*}
  X^* &= \xi + \frac{x_t-\xi}{\expect\left[e^{-(Z(T)-Z(t)) - \frac{1}{2}\int_t^T\|\sigma(s)\tran v^*(s)\|^2ds} \mathbf 1_{ Z(T)-Z(t)\ge 0}\right]}\mathbf 1_{Z(T)-Z(t)\ge 0}\\
  & = \xi + \frac{x_t-\xi}{\Phi\left(-\sqrt{\int_t^T\|\sigma(s)\tran v^*(s)\|^2ds}\right)}\mathbf 1_{Z(T)-Z(t)\ge 0}.
\end{align*}

Denote by $\bm{\pi}_{t,\mathrm{pc}}$ the pre-committed strategy, i.e., the strategy that attains terminal wealth is equal to $ X^*$. Then, for any $s\in [t,T]$, we have
\begin{align}
  X^{\bm{\pi}_{t,\mathrm{pc}}}_{t,x_t}(s) &= \expect_s\left[\big(\hat\rho(T)/\hat \rho(s)\big)X^*\right]\notag \\
  &=\xi + \frac{x_t-\xi}{\Phi\left(-\sqrt{\int_t^T\|\sigma(s)\tran v^*(s)\|^2ds}\right)}\expect_s\left[e^{-(Z(T)-Z(s)) - \frac{1}{2}\int_s^T\|\sigma(\tau)\tran v^*(\tau)\|^2d\tau}\mathbf 1_{Z(T)-Z(t)\ge 0}\right]\notag \\
  & = \xi + \frac{x_t-\xi}{\Phi\left(-\sqrt{\int_t^T\|\sigma(s)\tran v^*(s)\|^2ds}\right)}\Phi\left(d(s,Z(s)-Z(t))\right)\label{eq:PreComTimetWealth}
\end{align}
where
\begin{align*}
  d(s,z):=\frac{z -\int_s^T\|\sigma(\tau)\tran v^*(\tau)\|^2d\tau }{\sqrt{\int_s^T\|\sigma(\tau)\tran v^*(\tau)\|^2d\tau}}.
\end{align*}
By It\^o's lemma and straightforward calculation, we obtain
\begin{align}\label{eq:PreComStrategyGeneral}
 \frac{ \bm{\pi}_{t,\mathrm{pc}}\left(s,X^{\bm{\pi}_{t,\mathrm{pc}}}_{t,x_t}(s)\right)}{X^{\bm{\pi}_{t,\mathrm{pc}}}_{t,x_t}(s)-\xi} = \frac{1}{\sqrt{\int_s^T\|\sigma(\tau)\tran v^*(\tau)\|^2d\tau}} \frac{\Phi'\left(d(s,Z(s)-Z(t))\right)}{\Phi\left(d(s,Z(s)-Z(t))\right)}v^*(s),\quad s\in [t,T).
\end{align}
For each $s\in (t,T]$, $d(s,z)$ is continuous, strictly increasing in $z$ and
 \begin{align*}
   \lim_{z\downarrow -\infty}d(s,z)=-\infty,\quad \lim_{z\uparrow +\infty}d(s,z)=+\infty.
 \end{align*}
 As a result, $ X^{\bm{\pi}_{t,\mathrm{pc}}}_{t,x_t}(s)$ is continuous, strictly increasing in $Z(s)-Z(t)$, and
\begin{align*}
  \lim_{Z(s)-Z(t)\downarrow -\infty}X^{\bm{\pi}_{t,\mathrm{pc}}}_{t,x_t}(s) = \xi,\quad \lim_{Z(s)-Z(t)\uparrow +\infty}X^{\bm{\pi}_{t,\mathrm{pc}}}_{t,x_t}(s) = \xi + \frac{x_t-\xi}{\Phi\left(-\sqrt{\int_t^T\|\sigma(s)\tran v^*(s)\|^2ds}\right)}.
\end{align*}
Denote the inverse function of the relation between $X^{\bm{\pi}_{t,\mathrm{pc}}}_{t,x_t}(s)$ and $Z(s)-Z(t)$ as $\bm{z}_t(s,X^{\bm{\pi}_{t,\mathrm{pc}}}_{t,x_t}(s))$. Then, we have
\begin{align*}
 \frac{ \bm{\pi}_{t,\mathrm{pc}}\left(s,x\right)}{x-\xi} = \frac{1}{\sqrt{\int_s^T\|\sigma(\tau)\tran v^*(\tau)\|^2d\tau}} \frac{\Phi'\left(d(s, \bm{z}_t(s,x))\right)}{\Phi\left(d(s, \bm{z}_t(s,x))\right)}v^*(s),\quad s\in [t,T),x \in (\xi, \xi+k_t^*).
\end{align*}
Moreover, $\big(\Phi'(d)/\Phi(d)\big)' = [-d\Phi'(d)\Phi(d) - \Phi'(d)^2]/\Phi(d)^2<0$, where the inequality is the case because $-d\Phi(d)<\Phi'(d)$ for any $d\in \bbR$. As a result,
\begin{align*}
  \Delta_{t,\mathrm{pc}}(s,x):=\frac{1}{\sqrt{\int_s^T\|\sigma(\tau)\tran v^*(\tau)\|^2d\tau}} \frac{\Phi'\left(d(s, \bm{z}_t(s,x))\right)}{\Phi\left(d(s, \bm{z}_t(s,x))\right)}
\end{align*}
is strictly decreasing in $x$,
\begin{align*}
  \lim_{x\downarrow \xi}\Delta_{t,\mathrm{pc}}(s,x) = \lim_{d\downarrow -\infty}\frac{1}{\sqrt{\int_s^T\|\sigma(\tau)\tran v^*(\tau)\|^2d\tau}} \frac{\Phi'\left(d\right)}{\Phi\left(d\right)} = +\infty,
\end{align*}
and, similarly, $\lim_{x\uparrow \xi+k_t^*}\Delta_{t,\mathrm{pc}}(s,x) = 0$.

Specializing the above at $t=0$, we complete the proof of part (i) of the Proposition. For part (ii), Corollary 1 in \citet{HE2021}  yields that the optimal median of the terminal wealth is $\xi+k^*_0$, i.e.,
\begin{align*}
  G^{\bm{\pi}_{0,\mathrm{pc}}}(0,x_0,1/2) = \xi + k^*_0 = \xi + \frac{x_0-\xi}{\Phi\left(-\sqrt{\int_0^T\|\sigma(s)\tran v^*(s)\|^2ds}\right)}.
\end{align*}
Consider the function $f(x):=\Phi(-x)-e^{-x^2/2}, x\ge 0$. We have $f'(x) = e^{-x^2/2}\big(x - (2\pi)^{-1/2}\big)$, so $f$ is strictly decreasing on $[0,(2\pi)^{-1/2}]$ and strictly increasing on $[(2\pi)^{-1/2},+\infty)$. Because $f(0) = -1/2<0$ and $\lim_{x\uparrow +\infty}f(x)=0$, we conclude that $\Phi(-x)<e^{-x^2/2}$ for any $x\ge 0$. As a result, $G^{\bm{\pi}_{0,\mathrm{pc}}}(0,x_0,1/2)>G^{\hat{\bm{\pi}}_{\xi}}(0,x_0,1/2)$.

For any $t\in (0,T)$ and $x\in (\xi, \xi + k^*_0)$, we have
\begin{align*}
 \prob_t(X^{\bm{\pi}_{0,\mathrm{pc}}}_{0,x_0}(T)=\xi) &=  1- \prob_t(X^{\bm{\pi}_{0,\mathrm{pc}}}_{0,x_0}(T)=\xi + k^*_0)= \prob_t(Z(T)< 0)\\
 & = \prob_t(Z(T)-Z(t)< -Z(t))\begin{cases}
   > 1/2, & Z(t)< 0,\\
   \le 1/2, & Z(t)\ge 0.
 \end{cases}
\end{align*}
 Also note that $Z(t)\ge 0$ if and only if
 \begin{align*}
    X^{\bm{\pi}_{0,\mathrm{pc}}}_{0,x_0}(t) \ge \xi + \frac{x_0-\xi}{\Phi\left(-\sqrt{\int_0^T\|\sigma(s)\tran v^*(s)\|^2ds}\right)}\Phi\left(d(t,0)\right) = x_0.
 \end{align*}
 Therefore,
 \begin{align*}
   G^{\bm{\pi}_{0,\mathrm{pc}}}(t,x,1/2) = \begin{cases}
      \xi + k^*_0, & x\in[ x_0, \xi + k_0^*),\\
      \xi, & x\in (\xi, x_0).
   \end{cases}
 \end{align*}
Then, Proposition \ref{prop:PortfolioInsurance} yields \eqref{eq:EquiPreComMedianComparison} immediately. Recall that $\Phi(-x)<e^{-x^2/2}$ for any $x\ge 0$, so
\begin{align*}
  1<\frac{e^{-\frac{1}{2}\int_0^T\|\sigma(s)\tran v^*(s)\|^2ds}}{\Phi\left(-\sqrt{\int_0^T\|\sigma(\tau)\tran v^*(\tau)\|^2d\tau}\right)} \le \tilde a_t<\frac{1}{\Phi\left(-\sqrt{\int_0^T\|\sigma(\tau)\tran v^*(\tau)\|^2d\tau}\right)} = \frac{k_0^*}{x_0-\xi}.
\end{align*}
The proof then completes. \halmos
\end{pfof}

\begin{pfof}{Proposition \ref{prop:NaiveStrategy}}
  By definition,
  \begin{align*}
    \bm{\pi}_{\mathrm{na}}(t,x) = \bm{\pi}_{t,\mathrm{pc}}(t,X_{t,x}^{\bm{\pi}_{t,\mathrm{pc}}}(t)) = \frac{1}{\sqrt{\int_t^T\|\sigma(\tau)\tran v^*(\tau)\|^2d\tau}} \frac{\Phi'\left(d(t,0)\right)}{\Phi\left(d(t,0)\right)}v^*(t)(x-\xi),
  \end{align*}
  where the second equality is the case due to \eqref{eq:PreComStrategyGeneral}. Recall that $d(t,0) =- \sqrt{\int_t^T\|\sigma(\tau)\tran v^*(\tau)\|^2d\tau}$, we immediately obtain \eqref{eq:NaiveStrategy}.

  Because $\Phi'(d)>d\Phi(-d)$ for any $d> 0$, we conclude that $\Delta_{\mathrm{na}}(t)>1, t\in [0,T)$. Furthermore, it is easy to see that $\lim_{t\rightarrow T}\Delta_{\mathrm{na}}(t)=+\infty$.

  Finally, fix any $t\in [0,T)$ and $x>\xi$ and consider $\tau \in (t,T)$. Then, we have
  \begin{align*}
    X_{t,x}^{\bm{\pi}_{\mathrm{na}}}(\tau) &=\xi + (x-\xi)e^{\int_t^\tau \Delta_{\mathrm{na}}(s)b(s)\tran v^*(s)ds - \frac{1}{2}\int_t^\tau \Delta_{\mathrm{na}}(s)^2 \|\sigma(s)\tran v^*(s)\|^2ds + \int_t^{\tau}\Delta_{\mathrm{na}}(s)v^*(s)\tran \sigma(s)dW(s) }\\
     & = \xi + (x-\xi)e^{\int_t^\tau \left(\Delta_{\mathrm{na}}(s)-\frac{1}{2}\Delta_{\mathrm{na}}(s)^2\right) \|\sigma(s)\tran v^*(s)\|^2ds+ \int_t^{\tau}\Delta_{\mathrm{na}}(s)v^*(s)\tran \sigma(s)dW(s)},
  \end{align*}
  where the second equality is the case due to Lemma \ref{le:QuadProgSolution}. As a result, the median of $ X_{t,x}^{\bm{\pi}_{\mathrm{na}}}(\tau)$, denoted as $G^{\bm{\pi}_{\mathrm{na}}}(t,x,1/2;\tau)$, is
  \begin{align*}
    G^{\bm{\pi}_{\mathrm{na}}}(t,x,1/2;\tau) = \xi + (x-\xi)e^{\int_t^\tau \left(\Delta_{\mathrm{na}}(s)-\frac{1}{2}\Delta_{\mathrm{na}}(s)^2\right) \|\sigma(s)\tran v^*(s)\|^2ds}.
  \end{align*}
  Because $\inf_{s\in[0,T)}\|\sigma(s)\tran v^*(s)\|>0$ as shown in Lemma \ref{le:QuadProgSolution}, it is easy to see that
  \begin{align*}
    \int_t^T\Delta_{\mathrm{na}}(s)ds<+\infty, \quad \int_t^T\Delta_{\mathrm{na}}(s)^2ds=+\infty.
  \end{align*}
  As a result, $\lim_{\tau\uparrow T}G^{\bm{\pi}_{\mathrm{na}}}(t,x,1/2;\tau) = \xi$. \halmos
\end{pfof}

\subsection{Proof of Theorem \ref{th:QuantileMaximizationMultiple}}
The remainder of this session is devoted to the proof of Theorem \ref{th:QuantileMaximizationMultiple}.
Similar to the proof of Theorem \ref{th:EquilibriumLowerLimit}, we present the proof of Theorem \ref{th:QuantileMaximizationMultiple} by summarizing important intermediate steps of the proof as lemmas and relegate all proofs in Section \ref{subsubse:proofsMulti}.

\begin{lemma}\label{reachable state}
Fixed  $T_1 \in [0, T)$.  Consider $\hat{\bm{\pi}}\in \mathbb{A}$ and for any $x\in \bbR$ and $t\in [T_1,T]$, denote by $\bbX_t^{x,T_1,\hat{\bm{\pi}}}$ the set of reachable states of $X^{\hat{\bm{\pi}}}_{T_1,x}(t)$. Then, for any $x_1\in \bbX_{T_1}^{x_0,\hat{\bm{\pi}}}$, we have $\bbX_t^{x_1,T_1,\hat{\bm{\pi}}}\subseteq  \bbX_t^{x_0,\hat{\bm{\pi}}}$  for any $t\in [T_1, T]$.
\end{lemma}

Lemma \ref{reachable state} shows that any state that is reachable at given future time by a wealth equation starting from an intermediate time and state that is reachable from the initial wealth is also reachable from the initial wealth. As a result, if $\hat{\bm{\pi}}$ is an intrapersonal equilibrium at the initial time, it is also an intrapersonal equilibrium at intermediate time.

  \begin{proposition}\label{prop:NeccAlphaNotHalfMultipleT}
  For $\alpha\in(0,1/2)$, $\hat{\bm{\pi}}\in\mathbb{A}$ is an intra-personal equilibrium for the multiple-time-point for $\alpha$-level quantile maximization if and only if $\hat{\bm{\pi}}$ is given by \eqref{tail risk equilibrium Strategy} for some $\theta\in C_{\mathrm{pw}}([0,T))$ taking values in $\bbR^m$. For $\alpha\in (1/2,1)$, any $\hat{\bm{\pi}}\in\mathbb{A}$ is not an intra-personal equilibrium for multiple-time-point $\alpha$-level quantile maximization.
\end{proposition}

Proposition \ref{prop:NeccAlphaNotHalfMultipleT} shows that Theorem \ref{th:QuantileMaximizationMultiple}-(ii) and (iii) are true. In the following, we deal with the remaining case $\alpha=1/2$, i.e. the case of median maximization. To this end, for any $\hat{\bm{\pi}}\in\mathbb{A}$, define
\begin{align}
    \tau^*:&=\inf\{t\in[0,T_{N-1}):\theta_0(s)=\theta_1(s)=0,\forall s\in[t,T_{N-1})\},\label{eq:tau0}\\
    \tau_*:&=\inf\{t\in[0,\tau^*):\theta_0(s)+\xi_1\theta_1(s)=0,\forall s\in[t,\tau^*)\text{ and some }\xi_1 \in \bbR\}.\label{eq:tau0prime}
  \end{align}
    Then, there exists unique $\xi_1\in \bbR$ such that $\theta_0(s)+\xi_1 \theta_1(s)=0,\forall s\in[\tau_*,\tau^*)$. Denote \begin{align}\label{eq:SingularPointsNew}
     \tilde{\mathbb{S}}^{\hat{\bm{\pi}}}_t=\emptyset,\; t\in [0,\tau_*),\quad  \tilde{\mathbb{S}}^{\hat{\bm{\pi}}}_t=\{\xi_1\},\; t\in[\tau_*,T_{N-1}).
   \end{align}
   Recall $t^*$, $t_*$, and $\xi$ as defined in Lemma \ref{le:QuantileDifferentiability}. Recalling that $\mathbb{S}^{\hat{\bm{\pi}}}_t=\emptyset,\forall t\in [0,t_*)$, we derive that $\mathbb{S}^{\hat{\bm{\pi}}}_t \subseteq \tilde{\mathbb{S}}^{\hat{\bm{\pi}}}_t,\forall t\in [0,T_{N-1})$ because we have $\xi_1=\xi$ and $t_*=\tau_*$ in the case $t_*<T_{N-1}$.

\begin{lemma}\label{TwoMedianNeceCondi}
 Suppose $\hat{\bm{\pi}}\in\mathbb{A}$ is an intra-personal equilibrium for the multiple-time-point median maximization.
Then, the left end of $\bbX_{T_{N-1}}^{x_0,\hat{\bm{\pi}}} $ is finite and
\begin{align}\label{equilibrium for later time period}
   \hat{\bm{\pi}}(t,x)= v^*(t) (x-\xi), t\in [T_{N-1}, T), x\in \bbR
  \end{align}
for some constant $\xi$ that satisfies
\begin{align}\label{decreasing threshold}
 \xi< x,\quad \forall x\in  \bbX_{T_{N-1}}^{x_0,\hat{\bm{\pi}}}.
  \end{align}
Recall $t^*$, $t_*$, $\xi$, and $\mathbb{S}^{\hat{\bm{\pi}}}_t$ as defined in Lemma \ref{le:QuantileDifferentiability} and $\tau^*$, $ \tau_*$, and $\xi_1$ as defined in \eqref{eq:tau0} and \eqref{eq:tau0prime}. Then, $\tau^*=T_{N-1}$, $t^*=T$, and $t_*\leq T_{N-1}$.
\end{lemma}

Lemma \ref{TwoMedianNeceCondi} shows that if $\hat{\bm{\pi}}\in\mathbb{A}$ is an intra-personal equilibrium for the multiple-time-point median maximization, it must be a portfolio insurance strategy in the period $[T_{N-1},T)$, and the portfolio insurance level must be lower than any wealth level that is reachable at time $T_{N-1}$.

Denote by $F^{\bm{\pi}}(t,x,y;s):=\prob(X^{\bm{\pi}}_{t,x}(s)\le y), y\in \bbR$ the cumulative distribution function of the wealth at time $s$ given wealth level of $x$ at time $t\leq s$.
\begin{lemma}\label{le:NeccOptTwoMedian}
Recall $\underline{t}$ as defined in Lemma \ref{le:NeccNonDegAlphaLarge} and $\tilde{\mathbb{S}}^{\hat{\bm{\pi}}}_t$ as in \eqref{eq:SingularPointsNew}. Suppose $\hat{\bm{\pi}}\in\mathbb{A}$ is an intra-personal equilibrium for multiple-time-point median maximization.
For any $t\in[T_{N-2},T_{N-1})$, $x\in \bbX^{x_0,\hat{\bm{\pi}}}_t\backslash \tilde{\mathbb{S}}^{\hat{\bm{\pi}}}_t$, denoting
  \begin{align*}
   &\lambda_1^{\hat{\bm{\pi}}}(t,x)=(1-w_{N-1,N})\frac{F^{\hat{\bm{\pi}}}_x(t,x,G^{\hat{\bm{\pi}}}(t,x,\frac{1}{2};T_{N-1});T_{N-1})}{F^{\hat{\bm{\pi}}}_y(t,x, G^{\hat{\bm{\pi}}}(t,x,\frac{1}{2};T_{N-1});T_{N-1})}+w_{N-1,N}\frac{ F^{\hat{\bm{\pi}}}_x(t,x, G^{\hat{\bm{\pi}}}(t,x,\frac{1}{2}))}{ F^{\hat{\bm{\pi}}}_y(t,x, G^{\hat{\bm{\pi}}}(t,x,\frac{1}{2}))}, \\
   &\lambda_2^{\hat{\bm{\pi}}}(t,x)=(1-w_{N-1,N})\frac{F^{\hat{\bm{\pi}}}_{xx}(t,x, G^{\hat{\bm{\pi}}}(t,x,\frac{1}{2};T_{N-1});T_{N-1})}{F^{\hat{\bm{\pi}}}_y(t,x, G^{\hat{\bm{\pi}}}(t,x,\frac{1}{2};T_{N-1});T_{N-1})}+w_{N-1,N}\frac{ F^{\hat{\bm{\pi}}}_{xx}(t,x, G^{\hat{\bm{\pi}}}(t,x,\frac{1}{2}))}{ F^{\hat{\bm{\pi}}}_y(t,x, G^{\hat{\bm{\pi}}}(t,x,\frac{1}{2}))},
  \end{align*}
 we have $\lambda_2^{\hat{\bm{\pi}}}(t,x)>0$, $\lambda_1^{\hat{\bm{\pi}}}(t,x)<0$, and
  \begin{align}\label{eq:NeccOptSolutionTwoMedian}
    \hat{\bm{\pi}}(t,x) = -\frac{\lambda_1^{\hat{\bm{\pi}}}(t,x)}{\lambda_2^{\hat{\bm{\pi}}}(t,x)}v^*(t).
  \end{align}
Moreover, we have $\underline{t}\leq T_{N-2}$.
\end{lemma}

Lemma \ref{le:NeccOptTwoMedian} shows that any intra-personal equilibrium $\hat{\bm{\pi}}\in\mathbb{A}$ must take a particular form as in \eqref{eq:NeccOptSolutionTwoMedian} in the period $[T_{N-2},T_{N-1})$: the allocation across the risky assets in this period is the same as the Kelly strategy. 

\begin{lemma}\label{le:NeccPDETwoMedian}
  Suppose $\hat{\bm{\pi}}\in\mathbb{A}$ is an intra-personal equilibrium for the multiple-time-point median maximization problem. Recall $\xi_1$ and $\tilde{\mathbb{S}}^{\hat{\bm{\pi}}}_t$ as defined in \eqref{eq:tau0prime} and \eqref{eq:SingularPointsNew}, respectively. Then,
   \begin{enumerate}
     \item[(i)] For each $t\in(T_{N-2},T_{N-1}]$, $\bbX^{x_0,\hat{\bm{\pi}}}_t=(\underline{x}(t),+\infty)$ for some $\underline{x}(t)< x_0$, and $\underline{x}(t)$ is decreasing in $t \in(T_{N-2},T_{N-1}]$.
     \item[(ii)] There exists $\tilde a_0,\tilde a_1\in C([T_{N-2},T_{N-1}])$ taking values in $\bbR$ such that
   \begin{align}
     &-\frac{\lambda_1^{\hat{\bm{\pi}}}(t,x)}{\lambda_2^{\hat{\bm{\pi}}}(t,x)} = \tilde a_0(t) + \tilde a_1(t)x>0,\quad t\in (T_{N-2},T_{N-1}), x\in \bbX^{x_0,\hat{\bm{\pi}}}_t\backslash \tilde{\mathbb{S}}^{\hat{\bm{\pi}}}_t,\label{eq:NeccStrategyTwoMedian1}\\
     &\hat{\bm{\pi}}(t,x) = \big(\tilde a_0(t) + \tilde a_1(t)x\big)v^*(t),\quad (t,x)\in (T_{N-2},T_{N-1})\times \bbR.\label{eq:NeccStrategyTwoMedian2}
   \end{align}
   Consequently, $g^{\hat{\bm{\pi}}}(t,x):=w_{N-1,N}G^{\hat{\bm{\pi}}}(t,x,\frac{1}{2})+(1-w_{N-1,N}) G^{\hat{\bm{\pi}}}(t,x,\frac{1}{2};T_{N-1})  $satisfies following PDE:
    \begin{align}\label{TwoMedianEquation}
    \left\{
    \begin{array}{l}
g^{\hat{\bm{\pi}}}_t(t,x)  + \frac{1}{2}g^{\hat{\bm{\pi}}}_x(t,x) \rho(t)\big(\tilde a_0(t)+\tilde a_1(t)x\big)=0 , \; t\in (T_{N-2},T_{N-1}),x\in \bbX^{x_0,\hat{\bm{\pi}}}_t\backslash \tilde{\mathbb{S}}^{\hat{\bm{\pi}}}_t,   \\
\underset{t \uparrow  T_{N-1},x'\rightarrow x}{\lim}g^{\hat{\bm{\pi}}}(t,x')=w_{N-1,N} \left((x -\xi )e^{\int_{T_{N-1}}^T \frac{1}{2}\rho(s) ds }+\xi\right)+(1-w_{N-1,N})x, \; x\in \bbR \backslash \{\xi\},
\end{array}
\right.
 \end{align}
 where $\rho(t)$ is as defined in \eqref{eq:rhot}.
   \end{enumerate}
\end{lemma}

 \begin{lemma}\label{le:EquilibriumQuantileTwoMedian}
 Suppose $\hat{\bm{\pi}}\in\mathbb{A}$ is an intra-personal equilibrium for the multiple-time-point median maximization problem. Recall $\xi_1$ and $\tilde{\mathbb{S}}^{\hat{\bm{\pi}}}_t$ as defined in \eqref{eq:tau0prime} and \eqref{eq:SingularPointsNew}, respectively, and $\tilde a_0$, $\tilde a_1$, and $\underline{x}(t)$ as defined in Lemma \ref{le:NeccPDETwoMedian}. Define
   \begin{align}\label{eq:EquiMedianCoeff}
   \tilde  \beta_1(t):& =  e^{\frac{1}{2}\int_t^{T_{N-1}}\tilde a_1(s)\rho(s)ds}, \; \tilde \beta_0(t): = \frac{1}{2}\int_t^{T_{N-1}}\tilde a_0(s)\rho(s)\tilde\beta_1(s)ds,\; t\in [T_{N-2},T_{N-1}].
   \end{align}
   Fix any $t\in (T_{N-2},T_{N-1})$.
   \begin{enumerate}
   \item[(i)] There exists $s\in[t,T_{N-1})$ such that $\tilde a_1(s)\neq 0$. Consequently,
   \begin{align}\label{eq:a1Eq0TwoMedian}
     \int_t^{T_{N-1}}\tilde a_1(s)^2\|\sigma(s)\tran v^*(s)\|^2ds>0.
   \end{align}
     \item[(ii)]There exists $\underline{c}_t> \max(\underline{x}(t),\xi,\xi_1)$ such that
         \begin{align}\label{eq:EquilibriumWeightedQuantile}
         &w_{N-1,N} G^{\hat{\bm{\pi}}}(s,x,1/2)+(1-w_{N-1,N}) G^{\hat{\bm{\pi}}}(s,x,1/2;T_{N-1})\nonumber\\ & =w_{N-1,N} \big(\beta_1(s, 1/2)x
         + \beta_0(s, 1/2)\big)+(1-w_{N-1,N})\big(\tilde\beta_1(s)x + \tilde\beta_0(s)\big)
         \end{align}
         holds for all $(s,x)\in [t,T_{N-1}]\times [\underline{c}_t,+\infty)$, where $\beta_0(s,1/2),\beta_1(s,1/2)$ are as defined in \eqref{eq:EquiQuantCoeff} with $a_0(s)=-\xi,a_1(s)=1, s\in [T_{N-1}, T)$ and $a_0(s)=\tilde a_0(s),a_1(s)=\tilde a_1(s), s\in [T_{N-2}, T_{N-1})$. Moreover,
         \begin{align}\label{eq:a1Eq1TwoMedian}
         &\frac{1}{2}\int_t^{T_{N-1}} \tilde a_1(s)\big(1-\tilde a_1(s)\big)\|\sigma(s)\tran v^*(s)\|^2ds =0.
         \end{align}
   \end{enumerate}
\end{lemma}

\begin{proposition}\label{prop:NeccAlphaHalfMultipleT}
  Suppose $\alpha=1/2$. Then, $\hat{\bm{\pi}}\in \mathbb{A}$ is an intra-personal equilibrium for multiple-time-point $\alpha$-level quantile maximization if and only if $\hat{\bm{\pi}}$ is given by \eqref{eq:EquiStrategyPropTarget} for some constant $\xi<x_0$.
\end{proposition}

\subsubsection{Detailed Proofs}\label{subsubse:proofsMulti}
\begin{pfof}{Lemma \ref{reachable state}}
If $\underline{t}\geq T_1$, then $X_{0,x_0}^{\hat{\bm{\pi}}}(T_1)=x_0$ and thus $\bbX_{T_1}^{x_0,\hat{\bm{\pi}}}=\{x_0\}$. For any $t\in [T_1, T]$, we have $X^{\hat{\bm{\pi}}}_{T_1,x_0}(t)=X^{\hat{\bm{\pi}}}_{0,x_0}(t)$ and thus $\bbX_t^{x_0,T_1,\hat{\bm{\pi}}}=\bbX_t^{x_0,\hat{\bm{\pi}}}$.
In the following, we consider the case in which $\underline{t}< T_1$.

Fix $x_1\in \bbX_{T_1}^{x_0,\hat{\bm{\pi}}} $ and $t\in [T_1, T]$.
If $\theta_0(s)+\theta_1(s)x_1=0$ for any $s\in [T_1, t)$, then  Corollary 4   in \citet{HeJiang2020:LinearSDE} shows that $\bbX_t^{x_1,T_1,\hat{\bm{\pi}}}=\{x_1\} \subseteq \bbX_{T_1}^{x_0,\hat{\bm{\pi}}}\subseteq \bbX_{t}^{x_0,\hat{\bm{\pi}}}$. If $\theta_0(s)+\theta_1(s)x_1\neq 0$ for some $s\in [T_1, t)$, then  Corollary 4 and Theorem 3 in \citet{HeJiang2020:LinearSDE} show that $X^{\hat{\bm{\pi}}}_{0,x_0}(t)$ and $X^{\hat{\bm{\pi}}}_{T_1,x_1}(t)$ have density functions, so $\bbX_t^{x_0,\hat{\bm{\pi}}}$ and $\bbX_t^{x_1,T_1,\hat{\bm{\pi}}}$ are the interiors of the support of $X^{\hat{\bm{\pi}}}_{0,x_0}(t)$ and $X^{\hat{\bm{\pi}}}_{T_1,x_1}(t)$, respectively. It is well known that the support of $X^{\hat{\bm{\pi}}}_{T_1,x}(t)$ is contained in the support of $X^{\hat{\bm{\pi}}}_{0,x_0}(t)$ for each $x$ in the support of $X^{\hat{\bm{\pi}}}_{0,x_0}(T_1)$. As a result, $\bbX_t^{x_1,T_1,\hat{\bm{\pi}}}\subseteq \bbX_t^{x_1,\hat{\bm{\pi}}}$ for any $x_1\in \bbX_{T_1}^{x_0,\hat{\bm{\pi}}} $.\halmos
\end{pfof}

\begin{pfof}{Proposition \ref{prop:NeccAlphaNotHalfMultipleT}}
 For $\alpha\in(0,1/2)$, similar to Proposition \ref{prop:Sufficiency}, it's straightforward to show that $\hat{\bm{\pi}}$ as given by \eqref{tail risk equilibrium Strategy} is an intra-personal equilibrium for multiple-time-point $\alpha$-level quantile maximization.

Suppose that $\hat{\bm{\pi}}\in\mathbb{A}$ is an intra-personal equilibrium for the multiple-time-point $\alpha$-level quantile maximization. For $\alpha\in (1/2, 1)$, fix any $x_{1}\in \bbX_{T_{N-1}}^{x_0,\hat{\bm{\pi}}}$. According to Lemma \ref{reachable state}, $\hat{\bm{\pi}}(t,x), t\in [T_{N-1}, T), x\in \bbR$ is an intra-personal equilibrium for the $\alpha$-level quantile maximization with time interval $[T_{N-1}, T]$ and initial wealth $x_{1}$ at time $T_{N-1}$. According to Theorem \ref{th:EquilibriumLowerLimit}, however, $\hat{\bm{\pi}}$ cannot be an intrapersonal equilibrium so we derive a contradiction. Thus, there does not exist intra-personal equilibrium in $\mathbb{A}$ for the multiple-time-point $\alpha$-level quantile maximization with $\alpha \in (1/2,1)$.

For  $\alpha\in(0,1/2)$, recall $t^*$ as defined in Lemma \ref{le:QuantileDifferentiability}, $\underline{t}$ as defined in  Lemma \ref{le:NeccNonDegAlphaLarge}. We claim that it is either the case $t^*=0$ or the case $\underline{t}=T$ and, consequently, \eqref{tail risk equilibrium Strategy} holds. For the sake of contradiction, suppose $t^*>0$ and $\underline{t}<T$. Then, the definition of $t^*$ and $\underline{t}$ implies that $\underline{t} <t^*$. Moreover, there exist two target dates $T_{i-1}<T_{i}$, such that $t^*\in (T_{i-1}, T_{i}]$, $i\in \{1,2,\dots, N \}$. Fix any $T_{i-1}'\in (\max\{T_{i-1}, \underline{t}\}, t^*)$. For any $t\in [T_{i-1}', t^*)$, $\epsilon\in (0, t^*-t)$, $s\in [t^*, T]$, and $x\in \bbR$, we have $X^{\hat{\bm{\pi}}}_{t,x}(s)=X^{\hat{\bm{\pi}}}_{t,x}(t^*)$, $X^{\hat{\bm{\pi}}_{t,\epsilon,\pi}}_{t,x}(s)=X^{\hat{\bm{\pi}}_{t,\epsilon,\pi}}_{t,x}(t^*)$.  For any $x'\in \bbX_{T_{i-1}'}^{x_0,\hat{\bm{\pi}}}$, by Lemma \ref{reachable state}, $\hat{\bm{\pi}}$ is an intra-personal equilibrium for the $\alpha$-level quantile maximization for the period $[T_{i-1}', t^*]$ and with initial wealth $x'$ at time $T_{i-1}'$. Then, by Theorem \ref{th:EquilibriumLowerLimit}, $\hat{\bm{\pi}}(t,x) = \theta(t) (x-x'),  \quad \forall t\in[T_i', t^*), x\in\mathbb{R}$  for some $\theta\in C_{\mathrm{pw}}([T_{i-1}', t^*))$. Because $\underline{t}< T_{i-1}'$,  Corollary 4 and Theorem 3 in \citet{HeJiang2020:LinearSDE} imply that $\bbX_{T_{i-1}'}^{x_0,\hat{\bm{\pi}}}$ is a nonempty open interval. Then, taking arbitrary $x'\neq x''$ in $\bbX_{T_{i-1}'}^{x_0,\hat{\bm{\pi}}}$, we have $\hat{\bm{\pi}}(t,x) = \theta(t) (x-x')$ and $\hat{\bm{\pi}}(t,x)=\theta(t) (x-x'')$ for all $t\in[T_{i-1}', t^*)$ and $x\in\mathbb{R}$, so $\theta(t)=0 $ for any $ t\in[T_{i-1}', t^*)$. This, however, contradicts, the definition of $t^*$.
 \halmos
 \end{pfof}

\begin{pfof}{Lemma \ref{TwoMedianNeceCondi}}
Fix any $x_{N-1}\in \bbX_{T_{N-1}}^{x_0,\hat{\bm{\pi}}}$. According to Lemma \ref{reachable state}, $\hat{\bm{\pi}}$ is an intra-personal equilibrium for the $\alpha$-level quantile maximization in the period $[T_{N-1}, T]$ and with initial wealth $x_{N-1}$ at time $T_{N-1}$. Theorem \ref{th:EquilibriumLowerLimit} then yields that \eqref{equilibrium for later time period} holds for some constant $\xi<x_{N-1}$. Because $x_{N-1}\in \bbX_{T_{N-1}}^{x_0,\hat{\bm{\pi}}}$ is arbitrary, we derive \eqref{decreasing threshold} holds and, consequently, the left end of $\bbX_{T_{N-1}}^{x_0,\hat{\bm{\pi}}} $ is finite.

Next, it is straightforward to see from \eqref{equilibrium for later time period} that $t^*=T$ and $t_*\le T_{N-1}$. For the sake of contradiction, suppose $\tau^*<T_{N-1}$. Then, similar to the proof in
Lemma \ref{le:NeccNonDegAlphaLarge}, for any $t\in [\tau^* \vee T_{N-2},T_{N-1})$, $x\in \bbX^{x_0,\hat{\bm{\pi}}}_t$, and $\pi\in \bbR^m$ with $b(t)\tran \pi>0$ and $Q\pi\ge 0$, there exists $\epsilon_0>0$ such that $G^{\bm{\pi}_{t,\epsilon,\pi}}(t,x,\frac{1}{2};T_{N-1})>x =   G^{\hat{\bm{\pi}}}(t,x,\frac{1}{2};T_{N-1})$ for any $\epsilon \in (0, \epsilon_0]$. Recall $\underline{t}$ as defined in Lemma \ref{le:NeccNonDegAlphaLarge}.
If $\underline{t}\geq \tau^*$, then $\underline{t}\geq T_{N-1}$ and, consequently, $\bbX_{T_{N-1}}^{x_0,\hat{\bm{\pi}}}=\{x_0\}$. Then, \eqref{decreasing threshold} implies that $\xi<x_0$. As a result, for any $t\in [\tau^* \vee T_{N-2},T_{N-1})$, $\bbX^{x_0,\hat{\bm{\pi}}}_t\backslash \mathbb{S}^{\hat{\bm{\pi}}}_t=\{x_0\}$ is nonempty. If $\underline{t}< \tau^*$, then for any $t\in [\tau^* \vee T_{N-2},T_{N-1})$,  Corollary 4 and Theorem 3 in \citet{HeJiang2020:LinearSDE} show that $\bbX^{x_0,\hat{\bm{\pi}}}_t$ is a nonempty open interval and thus $\bbX^{x_0,\hat{\bm{\pi}}}_t\backslash \mathbb{S}^{\hat{\bm{\pi}}}_t$ is nonempty. Therefore, whether or not $\underline{t}\geq \tau^*$, for any $t\in [\tau^* \vee T_{N-2},T_{N-1})$, we can always find some $x\in \bbX^{x_0,\hat{\bm{\pi}}}_t\backslash \mathbb{S}^{\hat{\bm{\pi}}}_t$, and Corollary 3, Theorem 2-(iv) and Corollary 2-(ii) in \citet{HeJiang2020:LinearSDE} and
  show that $F^{\hat{\bm{\pi}}}_{x}(t,x,G^{\hat{\bm{\pi}}}(t,x,\frac{1}{2}))<0$. By the definition of $\tau^*$, we have $\hat{\bm{\pi}}(t,x)=0$, so we conclude from \eqref{eq:QuanDerivativeWRTeps} that there exists $\pi\in \bbR^m$ with $b(t)\tran \pi>0$ and $Q\pi\ge 0$ and $\epsilon_1\in (0, \epsilon_0]$ such that
$G^{\bm{\pi}_{t,\epsilon,\pi}}(t,x,\frac{1}{2})>G^{\hat{\bm{\pi}}}(t,x,\frac{1}{2}),\forall \epsilon\in (0, \epsilon_1]$. Recall that we already showed that for such $\pi$, $G^{\bm{\pi}_{t,\epsilon,\pi}}(t,x,\frac{1}{2};T_{N-1})>x =   G^{\hat{\bm{\pi}}}(t,x,\frac{1}{2};T_{N-1})$ for sufficiently small $\epsilon$. As a result, $\hat{\bm{\pi}}$ is not an intra-personal equilibrium, which is a contradiction. Thus, we must have $\tau^*=T_{N-1}$.
  \halmos
\end{pfof}

\begin{pfof}{Lemma \ref{le:NeccOptTwoMedian}}
For each $\alpha\in (0,1)$, denote by
 \begin{align}\label{eq:NeccOptObjNew}
 \tilde \varphi^{\hat{\bm{\pi}}}_{t,x,\alpha}(v):&=F^{\hat{\bm{\pi}}}_x(t,x, G^{\hat{\bm{\pi}}}(t,x,\alpha;T_{N-1});T_{N-1}) b(t)\tran v \notag\\
 &+\frac{1}{2}F^{\hat{\bm{\pi}}}_{xx}(t,x, G^{\hat{\bm{\pi}}}(t,x,\alpha;T_{N-1});T_{N-1})\|\sigma(t)\tran v\|^2.
\end{align}
For any $t\in[T_{N-2},T_{N-1})$  and $x\in \bbX^{x_0,\hat{\bm{\pi}}}_t\backslash \tilde{\mathbb{S}}^{\hat{\bm{\pi}}}_t$, Lemma \ref{le:QuantileDifferentiability} with $[t,T]$ therein replaced by $[t,T_{N-1}]$ yield that
\begin{align*}
\lim_{\epsilon\downarrow 0}\frac{ G^{\bm{\pi}_{t,\epsilon,\pi}}(t,x,\alpha;T_{N-1}) - G^{\hat{\bm{\pi}}}(t,x,\alpha;T_{N-1})}{\epsilon}=\frac{\tilde \varphi^{\hat{\bm{\pi}}}_{t,x,\alpha}(\hat{\bm{\pi}}(t,x)) - \tilde\varphi^{\hat{\bm{\pi}}}_{t,x,\alpha}(\pi)}{F^{\hat{\bm{\pi}}}_y(t,x, G^{\hat{\bm{\pi}}}(t,x,\alpha;T_{N-1});T_{N-1})}.
\end{align*}
Combining the above with Lemma \ref{le:QuantileDifferentiability}, we have
 \begin{align*}
    \lim_{\epsilon\downarrow 0}\frac{J(t,x; \bm{\pi}_{t,\epsilon,\pi} )-J(t,x; \hat{\bm{\pi}} )}{\epsilon}
    &=(1-w_{N-1,N})\frac{\tilde \varphi^{\hat{\bm{\pi}}}_{t,x,\alpha}(\hat{\bm{\pi}}(t,x)) - \tilde\varphi^{\hat{\bm{\pi}}}_{t,x,\alpha}(\pi)}{F^{\hat{\bm{\pi}}}_y(t,x, G^{\hat{\bm{\pi}} }(t,x,\alpha;T_{N-1});T_{N-1})}\\
    &+w_{N-1,N}\frac{\varphi^{\hat{\bm{\pi}}}_{t,x,\alpha}(\hat{\bm{\pi}}(t,x)) - \varphi^{\hat{\bm{\pi}}}_{t,x,\alpha}(\pi)}{F^{\hat{\bm{\pi}}}_y(t,x,G^{\hat{\bm{\pi}}}(t,x,\alpha))}.
  \end{align*}
Now, we set $\alpha=\frac{1}{2}$. Then, we have
\begin{align*}
  \lim_{\epsilon\downarrow 0}\frac{J(t,x; \bm{\pi}_{t,\epsilon,\pi} )-J(t,x; \hat{\bm{\pi}} )}{\epsilon} = \hat \varphi^{\hat{\bm{\pi}}}_{t,x}(\hat{\bm{\pi}}(t,x)) - \hat\varphi^{\hat{\bm{\pi}}}_{t,x}(\pi),
\end{align*}
where
   \begin{align*}
 &\hat \varphi^{\hat{\bm{\pi}}}_{t,x}(v):=\lambda_1^{\hat{\bm{\pi}}}(t,x)b(t)\tran v  +\frac{1}{2}\lambda_2^{\hat{\bm{\pi}}}(t,x) \|\sigma(t)\tran v\|^2.
\end{align*}
As a result, because $\hat{\bm{\pi}}$ is an intra-personal equilibrium, we have
\begin{align}\label{eq:LeNeccOptEq2}
  \hat\varphi^{\hat{\bm{\pi}}}_{t,x}(\hat{\bm{\pi}}(t,x))\le \hat\varphi^{\hat{\bm{\pi}}}_{t,x}(\pi),\quad \forall \pi\neq \hat{\bm{\pi}}(t,x)\text{ with }Q\pi \ge 0.
\end{align}
Similar to the proof of Lemma \ref{le:NeccOpt}, one can show that $\lambda_1^{\hat{\bm{\pi}}}(t,x)<0$ and $\lambda_2^{\hat{\bm{\pi}}}(t,x)>0$. Then, \eqref{eq:LeNeccOptEq2} immediately implies that $-\hat{\bm{\pi}}(t,x)\lambda_2^{\hat{\bm{\pi}}}(t,x)/\lambda_1^{\hat{\bm{\pi}}}(t,x)$ is the optimizer of \eqref{general small alpha 1}, i.e., \eqref{eq:NeccOptSolutionTwoMedian} holds.

Next, we prove $\underline{t}\leq T_{N-2}$. For the sake of contradiction, suppose $\underline{t}>T_{N-2}$. By the definition of $\underline{t}$, $X_{0,x_0}^{\hat{\bm{\pi}}}(s)=x_0$ and thus $\bbX^{x_0,\hat{\bm{\pi}}}_s=\{x_0\}$ for all $s\in[0,\underline{t}]$. When $\tau_*>T_{N-2}$, we choose any $t\in[T_{N-2}, \tau_*\wedge \underline{t} \wedge T_{N-1})$ and when $\tau_*\le T_{N-2}$ and $\xi_1\neq x_0$, we choose any $t\in[T_{N-2},\underline{t}\wedge T_{N-1} )$. For the above selected $t$, because $x_0\neq \xi_1$ and $\bbX^{x_0,\hat{\bm{\pi}}}_t=\{x_0\}$, we have $x_0\in \bbX^{x_0,\hat{\bm{\pi}}}_t\backslash \tilde{\mathbb{S}}^{\hat{\bm{\pi}}}_t$. As a result, $\lambda_1^{\hat{\bm{\pi}}}(t,x_0)<0$, $\lambda_2^{\hat{\bm{\pi}}}(t,x_0)>0$, and $\hat{\bm{\pi}}(t,x_0) = -\frac{\lambda_1^{\hat{\bm{\pi}}}(t,x_0)}{\lambda_2^{\hat{\bm{\pi}}}(t,x_0)}v^*(t)$.
The above is a contradiction because $t<\underline{t}$ and thus $\hat{\bm{\pi}}(t,x_0)=0$ and because $v^*(t)\neq 0$.

The remaining case is when $\tau_*\leq T_{N-2}$ and $\xi_1 =x_0$. In this case, we have $\underline{t}\geq T_{N-1}$, so $X_{0,x_0}^{\hat{\bm{\pi}}}(s)=x_0$ and $\bbX^{x_0,\hat{\bm{\pi}}}_s=\{x_0\}$ for all $s\in[0,\underline{t}]$. Fix $t\in[T_{N-2},T_{N-1})$. For any $\pi\in \bbR^m$ with $b(t)\tran \pi>0$ and $Q\pi\ge 0$, there exists $\epsilon_0\in(0,T_{N-1}-t)$, such that for any $\epsilon\in(0,\epsilon_0]$, $ F^{\bm{\pi}_{t,\epsilon,\pi}}(t,x_0,x_0;T_{N-1})=\Phi\left(\frac{-\int_{t}^{t+\epsilon} b(s)\tran \pi ds}{ \sqrt{ \int_{t}^{t+\epsilon} \|\sigma(s)\tran \pi\|^2 ds }   } \right)<\frac{1}{2}$.  In addition, because $\pi\neq 0$, $  F^{\bm{\pi}_{t,\epsilon,\pi}}(t,x_0,y;T_{N-1})$ is strictly increasing and continuous in $y$ for any $\epsilon \in (0,\epsilon_0]$. Consequently, we conclude that $ G^{\bm{\pi}_{t,\epsilon,\pi}}(t,x_0,\frac{1}{2};T_{N-1})>x_0=  G^{\hat{\bm{\pi}}}(t,x_0,\frac{1}{2};T_{N-1})$ for any $\epsilon \in (0,\epsilon_0]$, where the equality is the case because $X_{0,x_0}^{\hat{\bm{\pi}}}(T_{N-1})=x_0$. In other words, the median of wealth at time $T_{N-1}$ is improved by taking any strategy $\pi$ with $b(t)\tran \pi>0$ and $Q\pi\ge 0$. To complete the proof, we only need to show that there exists $\pi$ with $b(t)\tran \pi>0$ and $Q\pi\ge 0$ such that taking $\pi$ also improves the median of wealth at time $T$ because this implies that $\hat{\bm{\pi}}$ is not an intra-personal equilibrium and thus we have a contradiction.

If $\xi_1=\xi$, then we must have $\underline{t}=T$ because $\underline{t}>T_{N-1}$, $\xi_1=x_0$, and \eqref{equilibrium for later time period} holds. In this case, for any $\pi\in \bbR^m$ with $b(t)\tran \pi>0$ and $Q\pi\ge 0$, there exists $\epsilon_1\in(0,\epsilon_0]$ such that $F^{\bm{\pi}_{t,\epsilon,\pi}}(t,x_0,x_0)=\Phi\left(\frac{-\int_{t}^{t+\epsilon} b(s)\tran \pi ds}{ \sqrt{ \int_{t}^{t+\epsilon} \|\sigma(s)\tran \pi\|^2 ds }   } \right)<\frac{1}{2}$ for any $\epsilon \in (0,\epsilon_1]$. In addition, because $\pi\neq 0$, $  F^{\bm{\pi}_{t,\epsilon,\pi}}(t,x_0,y)$ is strictly increasing and continuous in $y$ for any $\epsilon \in (0,\epsilon_0]$. As a result, $ G^{\bm{\pi}_{t,\epsilon,\pi}}(t,x_0,\frac{1}{2})>x_0= G^{\hat{\bm{\pi}}}(t,x_0,\frac{1}{2}),\forall\epsilon \in (0,\epsilon_1]$, where the equality is the case because $\underline{t}=T$ and thus $X_{0,x_0}^{\hat{\bm{\pi}}}(T)=x_0$. 

If $\xi_1\neq \xi$, then we must have $t_*= T_{N-1}$. As a result,  Corollary 3, Theorem 2-(iv) and Corollary 2-(ii) in \citet{HeJiang2020:LinearSDE} imply that $F^{\hat{\bm{\pi}}}_{x}(t,x_0,G^{\hat{\bm{\pi}}}(t,x_0,\alpha))<0$. Because $\hat{\bm{\pi}}(t,x_0)=0$,
 we conclude from \eqref{eq:QuanDerivativeWRTeps} that there exists $\epsilon_1\in (0, \epsilon_0]$ and $\pi\in \bbR^m$ with $b(t)\tran \pi>0$ and $Q\pi\ge 0$ such that $G^{\bm{\pi}_{t,\epsilon,\pi}}(t,x,\frac{1}{2})>G^{\hat{\bm{\pi}}}(t,x,\frac{1}{2}),\forall \epsilon\in (0, \epsilon_1]$.
 \halmos
\end{pfof}

  \begin{pfof}{Lemma \ref{le:NeccPDETwoMedian}}
 Corollary 4   in \citet{HeJiang2020:LinearSDE} shows that $\bbX^{x_0,\hat{\bm{\pi}}}_t$ is increasing in $t\in [0,T]$ and that $\bbX^{x_0,\hat{\bm{\pi}}}_t$ is a half space or $\bbR$ for $t\in (\underline{t},T]$. Combining the above with \eqref{decreasing threshold} and recalling $\underline{t}\le T_{N-2}$ as shown in Lemma \ref{le:NeccOptTwoMedian}, we immediately derive part (i) of the Lemma.

 Next, we prove part (ii). Recall that Lemma \ref{TwoMedianNeceCondi} shows $\tau^*=T_{N-1}$ and $t^*=T$ and that Lemma \ref{le:NeccOptTwoMedian} shows $\underline{t}\le T_{N-2}$. Then, following the same proof as the one of Lemma \ref{le:NeccPDE}, we can derive \eqref{eq:NeccStrategyTwoMedian1} and \eqref{eq:NeccStrategyTwoMedian2} from \eqref{eq:NeccOptSolutionTwoMedian}. Because $\tau^*=T_{N-1}$ and $t^*=T$ and because $\mathbb{S}^{\hat{\bm{\pi}}}_t  \subseteq \tilde{\mathbb{S}}^{\hat{\bm{\pi}}}_t,\forall t\in [0,T_{N-1})$, we derive from Corollary 3, Theorem 2 and Corollary 2 in \citet{HeJiang2020:LinearSDE} that for $t\in(0,T_{N-1})$ and $x\in \bbX^{x_0,\hat{\bm{\pi}}}_{t}\backslash \mathbb{\tilde S}^{\hat{\bm{\pi}}}_t$,
\begin{align*}
&  F^{\hat{\bm{\pi}}}_{t}(t,x,G^{\hat{\bm{\pi}}}(t,x, 1/2)) + F^{\hat{\bm{\pi}}}_{x}(t,x,G^{\hat{\bm{\pi}}}(t,x, 1/2))b(t)\tran \hat{\bm{\pi}}(t,x) \\
&\qquad + \frac{1}{2}F^{\hat{\bm{\pi}}}_{xx}(t,x,G^{\hat{\bm{\pi}}}(t,x, 1/2))\|\sigma(t)\hat{\bm{\pi}}(t,x)\|^2= 0,\\
&   F^{\hat{\bm{\pi}}}_{t}(t,x, G^{\hat{\bm{\pi}}}(t,x, 1/2;T_{N-1});T_{N-1}) + F^{\hat{\bm{\pi}}}_{x}(t,x, G^{\hat{\bm{\pi}}}(t,x, 1/2;T_{N-1});T_{N-1})b(t)\tran \hat{\bm{\pi}}(t,x)\\
&\qquad  + \frac{1}{2}F^{\hat{\bm{\pi}}}_{xx}(t,x, G^{\hat{\bm{\pi}}}(t,x, 1/2;T_{N-1});T_{N-1})\|\sigma(t)\hat{\bm{\pi}}(t,x)\|^2= 0.
\end{align*}
Multiplying by $w_{N-1,N}/ F^{\hat{\bm{\pi}}}_{y}(t,x,G^{\hat{\bm{\pi}}}(t,x, \frac{1}{2}))$ and $(1-w_{N-1,N})/ F^{\hat{\bm{\pi}}}_{y}(t,x, G^{\hat{\bm{\pi}}}(t,x, \frac{1}{2};T_{N-1});T_{N-1})$ to both sides of the first and second equations, respectively, in the above and summing them up, we obtain
\begin{align*}
w_{N-1,N}\frac{ F^{\hat{\bm{\pi}}}_{t}(t,x,G^{\hat{\bm{\pi}}}(t,x, \frac{1}{2})) }{ F^{\hat{\bm{\pi}}}_{y}(t,x,G^{\hat{\bm{\pi}}}(t,x, \frac{1}{2})) }  +(1-w_{N-1,N})\frac{F^{\hat{\bm{\pi}}}_{t}(t,x,G^{\hat{\bm{\pi}}}(t,x, \frac{1}{2};T_{N-1});T_{N-1}) }{ F^{\hat{\bm{\pi}}}_{y}(t,x,G^{\hat{\bm{\pi}}}(t,x, \frac{1}{2};T_{N-1});T_{N-1}) }\\
+\lambda_1^{\hat{\bm{\pi}}}(t,x) b(t)\tran \hat{\bm{\pi}}(t,x) + \frac{1}{2}\lambda_2^{\hat{\bm{\pi}}}(t,x)  \|\sigma(t)\hat{\bm{\pi}}(t,x)\|^2= 0.
\end{align*}
Combining the above with \eqref{eq:NeccStrategyTwoMedian1} and \eqref{eq:NeccStrategyTwoMedian2}, we obtain
\begin{align}\label{OriginalTwoWeightEqua}
 w_{N-1,N}\frac{ F^{\hat{\bm{\pi}}}_{t}(t,x,G^{\hat{\bm{\pi}}}(t,x,\frac{1}{2})) }{ F^{\hat{\bm{\pi}}}_{y}(t,x,G^{\hat{\bm{\pi}}}(t,x,\frac{1}{2})) }  +(1-w_{N-1,N})\frac{F^{\hat{\bm{\pi}}}_{t}(t,x,G^{\hat{\bm{\pi}}}(t,x,\frac{1}{2};T_{N-1});T_{N-1}) }{ F^{\hat{\bm{\pi}}}_{y}(t,x,G^{\hat{\bm{\pi}}}(t,x,\frac{1}{2};T_{N-1});T_{N-1}) }\nonumber\\
  + \frac{1}{2}\lambda_1^{\hat{\bm{\pi}}}(t,x)\rho(t)\big(\tilde a_0(t)+\tilde a_1(t)x\big)= 0.
\end{align}
 Corollary 3 and    Corollary 2-(ii), (iii) in \citet{HeJiang2020:LinearSDE}  yield that
\begin{align*}
&  G^{\hat{\bm{\pi}}}_x(t,x,1/2) = -\frac{F^{\hat{\bm{\pi}}}_x(t,x,G^{\hat{\bm{\pi}}}(t,x,1/2))}{F^{\hat{\bm{\pi}}}_y(t,x,G^{\hat{\bm{\pi}}}(t,x,1/2))},\quad G^{\hat{\bm{\pi}}}_t(t,x,1/2) = -\frac{F^{\hat{\bm{\pi}}}_t(t,x,G^{\hat{\bm{\pi}}}(t,x,1/2))}{F^{\hat{\bm{\pi}}}_y(t,x,G^{\hat{\bm{\pi}}}(t,x,1/2))},\\
&  G^{\hat{\bm{\pi}}}_x(t,x,1/2;T_{N-1}) = -\frac{F^{\hat{\bm{\pi}}}_x(t,x, G^{\hat{\bm{\pi}}}(t,x,1/2;T_{N-1});T_{N-1})}{F^{\hat{\bm{\pi}}}_y(t,x, G^{\hat{\bm{\pi}}}(t,x,1/2;T_{N-1});T_{N-1})},\\
&  G^{\hat{\bm{\pi}}}_t(t,x,1/2;T_{N-1}) = -\frac{F^{\hat{\bm{\pi}}}_t(t,x, G^{\hat{\bm{\pi}}}(t,x,1/2;T_{N-1});T_{N-1})}{F^{\hat{\bm{\pi}}}_y(t,x, G^{\hat{\bm{\pi}}}(t,x,1/2;T_{N-1});T_{N-1})}  .
\end{align*}
Plugging above into \eqref{OriginalTwoWeightEqua}, we derive the differential equation in \eqref{TwoMedianEquation}. Because $\tau^* = T_{N-1}$, we conclude from Corollary 3 and   Corollary 2-(iv) in \citet{HeJiang2020:LinearSDE}   that $ \lim_{t\uparrow T_{N-1},x'\rightarrow x} G^{\hat{\bm{\pi}}}(t,x',1/2;T_{N-1}) =x,\forall x\in \bbR$. Because $t^*=T$, Corollary 3 and   Corollary 2-(iv) in \citet{HeJiang2020:LinearSDE} yield  that
\begin{align*}
  \lim_{t\uparrow T_{N-1},x'\rightarrow x} G^{\hat{\bm{\pi}}}(t,x',1/2) =G^{\hat{\bm{\pi}}}(T_{N-1},x,1/2)=(x -\xi )e^{\int_{T_{N-1}}^T (\rho(s)/2) ds }+\xi,\forall x\neq \xi,
\end{align*}
where the where the second equality is the case due to \eqref{equilibrium for later time period}. As a result, we derive the boundary condition in \eqref{TwoMedianEquation}. \halmos
    \end{pfof}

    \begin{pfof}{Lemma \ref{le:EquilibriumQuantileTwoMedian}}
We prove part (i) first. For the sake of contradiction, suppose $\tilde a_1(s)=0,s\in[t,T_{N-1})$. Because $\tau^*=T_{N-1}$ by Lemma \ref{TwoMedianNeceCondi}, for any $s\in[t,T_{N-1})$, there exists $\tau\in[s,T_{N-1})$ with $\theta_0(\tau)\neq 0$ and thus $\tilde a_0(\tau)\neq 0$. As a result, for any $x\in \bbX_t^{x_0,\hat{\bm{\pi}}}$,
\begin{align*}
   X^{\hat{\bm{\pi}}}_{t,x}(T_{N-1}) = x + \int_t^{T_{N-1}}\tilde a_0(\tau)v^*(\tau)\tran b(\tau)d\tau +\int_t^{T_{N-1}}\tilde a_0(\tau)v^*(\tau)\tran \sigma(\tau)dW(\tau)
\end{align*}
is a non-degenerate normal random variable, which, together with Lemma \ref{reachable state}, implies that $\bbX_{T_{N-1}}^{x_0,\hat{\bm{\pi}}} \supseteq \bbX_{T_{N-1}}^{x,t,\hat{\bm{\pi}}}=\bbR$. On the other hand, \eqref{decreasing threshold} implies that the left end of $\bbX_{T_{N-1}}^{x_0,\hat{\bm{\pi}}} $ is finite, so we arrive at contradiction.

Next, we prove part (ii). Because $\bbX^{x_0,\hat{\bm{\pi}}}_t=(\underline{x}(t),+\infty)$ by Lemma \ref{le:NeccPDETwoMedian} and because $\bbX^{x_0,\hat{\bm{\pi}}}_s$ is increasing in $s$, there exists $\underline{x}_t$ such that $[\underline{x}_t,+\infty)\subset \bbX^{x_0,\hat{\bm{\pi}}}_s\backslash\{\xi,\xi_1\},\forall s\in [t,T_{N-1}]$. As a result, \eqref{TwoMedianEquation} holds in the region $[t,T_{N-1}]\times (\underline{x}_t,+\infty)$. Similar to the proof of Lemma \ref{le:EquilibriumQuantile}, we can apply Lemma \ref{le:TransportEq} to conclude that there exists $\underline{c}_t>\max(\underline{x}(t),\xi,\xi_1)$ such that
 \begin{align*}
         g^{\hat{\bm{\pi}}}(s,x)=w_{N-1,N}\big(\beta_1(s, 1/2)x+ \beta_0(s, 1/2)\big)+(1-w_{N-1,N})\big(\tilde\beta_1(s)x + \tilde\beta_0(s)\big)
 \end{align*}
 for all $(s,x)\in [t,T_{N-1}]\times [\underline{c}_t,+\infty)$. Consequently, \eqref{eq:EquilibriumWeightedQuantile} holds.

 Finally, the proof of \eqref{eq:a1Eq1TwoMedian} is the same as the proof of \eqref{eq:a1Eq1} in Lemma \ref{le:EquilibriumQuantile}.\halmos
\end{pfof}

 \begin{pfof}{Proposition \ref{prop:NeccAlphaHalfMultipleT}}
 By Lemma \ref{le:EquilibriumQuantileTwoMedian}-(ii), we have
  \begin{align*}
    \frac{1}{2}\int_t^{T_{N-1}} \tilde a_1(s)\big(1-\tilde a_1(s)\big)\|\sigma(s)\tran v^*(s)\|^2ds=0,\quad \forall t\in (T_{N-2},T_{N-1}),
  \end{align*}
  which implies that $\tilde a_1(t)(1-\tilde a_1(t))=0,t\in(T_{N-2},T_{N-1})$. Lemma \ref{le:EquilibriumQuantileTwoMedian}-(i) shows for any $t\in (T_{N-2},T_{N-1})$, $\tilde a_1(s)\neq 0$ for some $s\in [t,T_{N-1})$. Then, we must have $\tilde a_1(t)=1,t\in (T_{N-2},T_{N-1})$ because $\tilde a_1\in C([T_{N-2},T_{N-1}))$ as shown by Lemma \ref{le:NeccPDETwoMedian}-(ii). Then, \eqref{eq:NeccStrategyTwoMedian2} implies that
  \begin{align}\label{eq:PropNeccAlphaHalfEq1TwoMedian}
    \hat{\bm{\pi}}(t,x) = (\tilde a_0(t) + x)v^*(t),\quad t\in[T_{N-2},T_{N-1}),x\in \bbR.
  \end{align}
  Define $\hat a_0$ by setting $\hat a_0(t) = \tilde a_0(t),t\in[T_{N-2},T_{N-1})$ and $\hat a_0(t) = -\xi,t\in[T_{N-1},T)$. Then, we have
  \begin{align*}
    \hat{\bm{\pi}}(t,x) = (\hat a_0(t) + x)v^*(t),\quad t\in[T_{N-2},T),x\in \bbR.
  \end{align*}
  Next, we prove that $\hat a_0(t)=-\xi$ for any $t\in [T_{N-2},T)$.

Lemma \ref{le:NeccPDETwoMedian}-(i) shows that $\bbX^{x_0,\hat{\bm{\pi}}}_t=(\underline{x}(t),+\infty),t\in (T_{N-2},T_{N-1}]$ for some decreasing function $\underline{x}$ on $(T_{N-2},T_{N-1}]$. Moreover, we derive from \eqref{decreasing threshold} that $\underline{x}(T_{N-1})\ge \xi$. Because the closure of $\bbX^{x_0,\hat{\bm{\pi}}}_{T_{N-1}}$, which is equal to $[\underline{x}(T_{N-1}),+\infty)$, is the support of $X^{\hat{\bm{\pi}}}_{0,x_0}(T_{N-1})$, we have $\prob(X^{\hat{\bm{\pi}}}_{0,x_0}(T_{N-1})\ge \xi)\ge \prob(X^{\hat{\bm{\pi}}}_{0,x_0}(T_{N-1})\ge \underline{x}(T_{N-1}))=1$. Because $\hat{\bm{\pi}}(s,x) = (- \xi + x)v^*(s),s\in [T_{N-1},T)$ and because $\underline{x}(T_{N-1})\ge \xi$, we conclude that for any $t\in (T_{N-1},T]$, $\prob(X^{\hat{\bm{\pi}}}_{0,x_0}(t)\ge \xi)=1$ and thus, by  Corollary 4   in \citet{HeJiang2020:LinearSDE}, $\bbX^{x_0,\hat{\bm{\pi}}}_t=(\underline{x}(t),+\infty)$ for some $\underline{x}(t)\in \bbR$.

Fix any $t_1\in (T_{N-2}, T)$ and $x_1\in \bbX^{x_0,\hat{\bm{\pi}}}_{t_1} \backslash \{-\hat a_0(t_1) \}$. Lemma \ref{reachable state} implies that $\bbX_t^{x_1,t_1,\hat{\bm{\pi}}} \subseteq \bbX^{x_0,\hat{\bm{\pi}}}_{t} = (\underline{x}(t),+\infty),\forall t\in [t_1,T] $. Corollary 4   in \citet{HeJiang2020:LinearSDE} then shows that $\hat a_0$ must be increasing on $(t_1,T)$ and $\bbX_t^{x_1,t_1,\hat{\bm{\pi}}} = (-\hat a_0(t),+\infty),\forall t\in (t_1,T)$; otherwise the lower end of $\bbX_t^{x_1,t_1,\hat{\bm{\pi}}}$ would be $-\infty$. For each $t\in (t_1,T)$, because $\bbX_t^{x_1,t_1,\hat{\bm{\pi}}} \subseteq \bbX^{x_0,\hat{\bm{\pi}}}_{t}=(\underline{x}(t),+\infty)$, we derive $-\hat a_0(t)\ge \underline{x}(t)$. Moreover,  Corollary 4  in \citet{HeJiang2020:LinearSDE} shows that $\bbX_s^{x_1,t_1,\hat{\bm{\pi}}}$ is increasing in $s\ge t_1$, so we have $\{x_1\} = \bbX_{t_1}^{x_1,t_1,\hat{\bm{\pi}}}\subseteq \bbX_s^{x_1,t_1,\hat{\bm{\pi}}}=(-\hat a_0(s),+\infty),\forall s\in (t_1,T_{N-1})$. Sending $s$ in the above to $t_1$ and recalling that $\hat a_0$ is right-continuous on $[T_{N-2},T)$, we conclude $x_1\ge -\hat a_0(t_1)$. Because $x_1\in \bbX^{x_0,\hat{\bm{\pi}}}_{t_1} \backslash \{-\hat a_0(t_1) \} = (\underline{x}(t_1),+\infty) \backslash \{-\hat a_0(t_1) \}$ is arbitrary, we derive that $\underline{x}(t_1)\ge -\hat a_0(t_1)$. Because $t_1$ is arbitrary, we conclude that $\hat a_0$ is increasing on $(T_{N-2},T)$ and that $\underline{x}(t)= -\hat a_0(t)$ and thus $\bbX^{x_0,\hat{\bm{\pi}}}_t=(-\hat a_0(t),+\infty)$ for all $t\in (T_{N-2},T)$.


By the definition of $\tau_*$ and $\xi_1$, we have $\tau_*=\inf\{t\in[0,T_{N-1}]: \hat a_0(s)=-\xi_1,s\in[t,T_{N-1}]\}$. For $t\in (T_{N-2},\tau_*)$, $\tilde{\mathbb{S}}^{\hat{\bm{\pi}}}_t=\emptyset$ and thus $\bbX^{x_0,\hat{\bm{\pi}}}_t\backslash \tilde{\mathbb{S}}^{\hat{\bm{\pi}}}_t = (-\hat a_0(t),+\infty)$, and for $t\in [\tau_*,T_{N-1})$, $\tilde{\mathbb{S}}^{\hat{\bm{\pi}}}_t=\{\xi_1\} = \{-\hat a_0(t)\}$ and thus $\bbX^{x_0,\hat{\bm{\pi}}}_t\backslash \tilde{\mathbb{S}}^{\hat{\bm{\pi}}}_t = (-\hat a_0(t),+\infty)$. Also recall that $\hat a$ is increasing on $(T_{N-2},T)$ and is equal to $\xi$ on $[T_{N-1},T)$, so $-\hat a_0(t)\ge \xi$ for all $t\in [T_{N-2},T_{N-1})$. As a result, Lemma \ref{le:NeccPDETwoMedian}-(ii) and Lemma \ref{le:TransportEq}-(iii) yield that \eqref{eq:EquilibriumWeightedQuantile} holds for any $x>-\hat a_0(s)$ and $s\in(T_{N-2},T_{N-1})$.

  Recall that $t_*\leq T_{N-1}$ and in the following, we prove that $t_* \leq T_{N-2}$. For the sake of contradiction, suppose that $t_*\in (T_{N-2},T_{N-1}]$. Fix any $t\in (T_{N-2},t_*)$.
 Because $\hat a_0$ is increasing in $(T_{N-2},T)$ and continuous on $[T_{N-2},T_{N-1})$ and on $[T_{N-1},T)$, following the same calculation of \eqref{eq:XTwitha0} in the proof of Proposition \ref{prop:NeccAlphaHalf}, we derive
\begin{align*}
  X^{\hat{\bm{\pi}}}_{t,-\hat a_0(t)}(T) 
  & = -\hat a_0(T) + \int_t^Te^{\int_s^T \left(b(\tau)\tran v^*(\tau) - \frac{1}{2}\|\sigma(\tau)\tran v^*(\tau)\|^2d\tau\right)d\tau + \int_s^T v^*(\tau)\tran \sigma(\tau)dW(\tau) }d\hat a_0(s).
\end{align*}
   Similar to the derivation of \eqref{eq:PropNeccAlphaHalfEq3}, we derive that
\begin{align}
  &X^{\hat{\bm{\pi}}}_{t,-\hat a_0(t)}(T) -\left(\beta_0(t,1/2) - \hat a_0(t)\beta_1(t,1/2)\right)  \notag\\
   & = \int_t^{T}e^{\frac{1}{2}\int_s^T\rho(\tau)d\tau + \int_s^{T} v^*(\tau)\tran \sigma(\tau)dW(\tau) }d\hat a_0(s) - \int_t^{T} e^{\frac{1}{2}\int_s^T\rho(\tau)d\tau} d\hat a_0(s).\label{eq:PropNeccAlphaHalfEqTwoMedian3}
\end{align}
Because $t<t_*$, $d\hat a_0(s)$ defines a positive measure on $(t,T)$. 
Similar to the proof in Proposition \ref{prop:NeccAlphaHalf}, we can prove that
there exists $\delta_t>0$ such that
\begin{align}\label{MedianLimitIn hata0 ForT}
G^{\hat{\bm{\pi}}}(t,x,1/2)-\big(\beta_0(t,1/2) +x\beta_1(t,1/2) \big) >0,\; \forall x\in ( -\hat a_0(t)-\delta_t, -\hat a_0(t) +\delta_t ).
\end{align}

We first consider the case in which $\tau_*\leq T_{N-2}$. Then, because $\tau^*=T_{N-1}$ as shown in Lemma \ref{TwoMedianNeceCondi}, we conclude $\hat a_0(t)=-\xi_1,\forall t\in [T_{N-2}, T_{N-1})$. Because $t_*>T_{N-2}$ and $\xi_1\ge \xi$, we must have $\xi_1>\xi$ and $t^*=T_{N-1}$. Then, straightforward calculation yields
\begin{align*}
  G^{\hat{\bm{\pi}}}(t,x,1/2;T_{N-1}) = \xi_1 + (x-\xi_1)e^{\frac{1}{2}\int_t^{T_{N-1}}\rho(s) ds}=\tilde\beta_1(t)x + \tilde\beta_0(t),\; x>\xi_1,\; t\in [T_{N-2},T_{N-1}).
\end{align*}
Combining the above with \eqref{MedianLimitIn hata0 ForT}, recalling that \eqref{eq:EquilibriumWeightedQuantile} holds for any $x>-\hat a_0(s)$ and $s\in(T_{N-2},T_{N-1})$, and noting that $w_{N-1, N}\in (0, 1]$, we derive a contradiction, so we must have $t_* \leq T_{N-2}$.

Next, we consider the case in which $\tau_*> T_{N-2}$. Fix any $t\in (T_{N-2},\tau_*)$.
    Then, $\hat a_0(s)$ is not constant in $s\in[t,\tau_*]$, which means that $d\hat a_0(s)$ defines a positive measure on $[t,\tau_*]$.
Similar to the derivation of \eqref{MedianLimitIn a0 ForT}, we can prove that there exists $\tilde \delta_t >0$, such that
\begin{align*}
 G^{\hat{\bm{\pi}}}(t,x,1/2;T_{N-1})-\big(\tilde\beta_0(t) +x\tilde\beta_1(t) \big) >0,\quad \forall x\in ( -\hat a_0(t)-\tilde \delta_t, -\hat a_0(t) +\tilde \delta_t ).
\end{align*}
Combining the above with \eqref{MedianLimitIn hata0 ForT}, recalling that \eqref{eq:EquilibriumWeightedQuantile} holds for any $x>-\hat a_0(s)$ and $s\in(T_{N-2},T_{N-1})$, and noting that $w_{N-1, N}\in (0, 1]$, we derive a contradiction, so we must have $t_* \leq T_{N-2}$.

Having proved that $t_* \leq T_{N-2}$ and thus $\hat a_0(s)=-\xi,\forall s\in [T_{N-2}, T)$, we conclude that $\hat{\bm{\pi}}(t,x)= v^*(t) (x-\xi), t\in [T_{N-2}, T), x\in \bbR$. Applying Lemmas \ref{TwoMedianNeceCondi}--\ref{le:EquilibriumQuantileTwoMedian} and the above proof to the periods $[T_{i-1},T_i)$, $i=N-2,N-1,\dots, 1$ sequentially, we can conclude that $ \hat{\bm{\pi}}(t,x) = (x-\xi)v^*(t), t\in[0,T),x\in \bbR$ for some $\xi<x_0$. The proof then completes.\halmos
\end{pfof}

\bibliography{LongTitles,BibFile}

\end{document}